\documentclass[10pt,authoryear]{elsarticle}

\usepackage[left=2.5cm, right=2.5cm, top=3.5cm, bottom=3cm, footskip=0.5cm]{geometry}
\usepackage{amsmath, amsfonts, amsthm, amssymb}
\usepackage{verbatim}
\usepackage{hyperref}
\usepackage{cleveref}
\usepackage{color}
\usepackage{graphicx}
\usepackage{enumerate}
\usepackage{wrapfig}
\usepackage{eurosym}
\usepackage{subcaption}
\usepackage{ulem}
\usepackage{cases} 
\usepackage{fancyhdr}
\usepackage{multirow}
\usepackage{mathrsfs}

\numberwithin{equation}{section}
\newtheorem{theorem}{Theorem}[section]
\newtheorem{corollary}[theorem]{Corollary}
\newtheorem{lemma}[theorem]{Lemma}
\newtheorem{proposition}[theorem]{Proposition}

\theoremstyle{definition}
\newtheorem{definition}[theorem]{Definition}
\newtheorem{remark}[theorem]{Remark}
\newtheorem{assumption}[theorem]{Assumption}
\newtheorem{sassumption}[theorem]{Standing Assumption}

\numberwithin{equation}{section}
\usepackage[none]{hyphenat}

\DeclareMathOperator*{\argmax}{arg\,max}

\newcommand{\RR}{\mathbb{R}}
\newcommand{\PP}{\mathbb{P}}

\newcommand{\GG}{\mathbb{G}}

\newcommand{\Ir}{\mathbf{I}}
\newcommand{\NN}{\mathbb{N}}
\newcommand{\Ff}{\mathcal{F}}

\newcommand{\Ttt}{\mathcal{T}}
\newcommand{\ft}{\mathfrak{t}}
\newcommand{\Jj}{\mathcal{I}}
\newcommand{\mfv}{\mathfrak{v}}
\newcommand{\cN}{\mathcal{N}}
\newcommand{\cLN}{\mathcal{LN}}
\newcommand{\cW}{\mathcal{W}}

\newcommand{\cF}{\mathcal{F}}
\newcommand{\cZ}{\mathcal{Z}}
\newcommand{\cV}{\mathcal{V}}
\newcommand{\cC}{\mathcal{C}}
\newcommand{\cG}{\mathcal{G}}
\newcommand{\cH}{\mathcal{H}}
\newcommand{\cA}{\mathcal{A}}
\newcommand{\cB}{\mathcal{B}}
\newcommand{\cD}{\mathcal{D}}

\newcommand{\cL}{\mathcal{L}}

\newcommand{\cU}{\mathcal{U}}
\newcommand{\EE}{\mathbb{E}}
\newcommand{\dr}{\mathrm{d}}
\newcommand{\OneN}{\{1,\hdots,N\}}

\newcommand{\VV}{\mathbb{V}}
\newcommand{\bOne}{\mathbf{1}}

\newcommand{\af}{\mathfrak{a}}
\newcommand{\pf}{\mathfrak{p}}
\newcommand{\ff}{\mathfrak{f}}
\newcommand{\cc}{\mathfrak{c}}
\newcommand{\cK}{\mathcal{K}}

\newcommand{\ee}{\mathfrak{e}}
\newcommand{\bb}{\mathfrak{b}}
\newcommand{\dd}{\mathfrak{d}}

\newcommand{\Pp}{{P}}
\newcommand{\Ww}{{W}}
\newcommand{\Cc}{{C}}
\newcommand{\Yy}{{Y}}
\newcommand{\Hh}{{H}}

\newcommand{\Ll}{{\Lambda}}
\newcommand{\llambda}{\boldsymbol{\lambda}}

\newcommand{\LGD}{\mathrm{LGD}}

\newcommand{\EAD}{\mathrm{EAD}}
\newcommand{\PD}{\mathrm{PD}}

\newcommand{\EL}{\mathrm{EL}}
\newcommand{\UL}{\mathrm{UL}}
\newcommand{\ES}{\mathrm{ES}}
\newcommand{\VaR}{\mathrm{VaR}}

\newcommand{\titre}{Impact of the carbon price on credit portfolio's loss with stochastic collateral}

\usepackage{fancyhdr}
\begin{document}

\begin{frontmatter}

\title{\titre}
\date{\today}

\author[1,2,3]{Lionel Sopgoui}

\address[1]{ Laboratoire de Probabilités, Statistique et Modélisation (LPSM), Université Paris Cité}
\address[2]{Department of Mathematics, Imperial College London}
\address[3]{Validation des modèles, Direction des risques, Groupe BPCE}

\journal{arXiv}

\begin{abstract}
The aim of this work is to propose an end-by-end modeling framework to evaluate the risk measures of a bank's portfolio of collateralized loans when the bank, the borrowers, as well as their guarantees operate in an economy subject to the climate transition. The economy, organized in sectors, is  driven by its productivity which is a multidimensional Ornstein-Uhlenbeck process while the climate transition is declined thanks to the carbon price and carbon intensities, continuous deterministic processes. We thus derive the dynamics of macroeconomic variables for each climate transition scenario.
By considering that a firm defaults if it is \textit{over-indebted} i.e. its market value -- depreciated due to the carbon price -- becomes less than the market value of its debt, 
we define each loan's loss at default as the difference between Exposure at Default (EAD) and the liquidated collateral, which will help us to define the Loss Given Default (LGD) -- the expected percentage of exposure that is lost if a debtor defaults. We consider two types of collateral: \textit{financial asset} such as invoices, cash, or investments or \textit{physical asset} such as real estate, business equipment, or inventory. First, if it is a \textit{financial asset}, we model the later by the continuous time version of the discounted cash flows methodology, where the cash flows growth is driven by the instantaneous output growth, the instantaneous growth of a carbon price function, and an arithmetic Brownian motion. Secondly, for \textit{physical asset}, we focus on the example of a \textit{property in housing market}. Therefore, we define, as \cite{sopgoui2024realestate}, its value as the difference between the price of an equivalent efficient building following an exponential Ornstein-Uhlenbeck (OU) as well as the actualized renovation costs and the actualized sum of the future additional energy costs due to the inefficiency of the building, before an optimal renovation date which depends on the carbon price process. Finally, we obtain expressions for risk measures of a portfolio of collateralized loans as a function of various parameters and variables, mainly those linked to the climate transition, such as the carbon price and the energy efficiency of buildings. These risk measures will be used by banks, depending on the climate transition scenarios, to define their operating expenses, the fees applied to clients, as well as their economic and regulatory capital. 
\end{abstract}


\end{frontmatter}

\footnotesize This research is part of the PhD thesis in Mathematical Finance of Lionel Sopgoui whose works are funded by a CIFRE grant from BPCE S.A. The opinions expressed in this research are those of the authors and are not meant to represent the opinions or official positions of BPCE S.A.

\footnotesize We would like to thank Jean-Fran\c{c}ois Chassagneux, Antoine Jacquier, Smail Ibbou, and Géraldine Bouveret for helpful comments on an earlier version of this work.


\normalsize
\newpage

\paragraph{Notations}
\begin{itemize}
    \item $\NN$ is the set of non-negative integers, $\NN^{*} := \NN\setminus\{0\}$, and $\mathbb{Z}$ is the set of integers.
    \item $\RR^d$ denotes the $d$-dimensional Euclidean space, $\RR_{+}$ is the set of non-negative real numbers, $\RR_{+}^{*} := \RR_{+}\setminus\{0\}$.
    \item {For $n,d\in\NN^*$, $\RR^{n\times d}$ is the set of real-valued $n\times d$ matrices ($\RR^{n\times 1} = \RR^{n}$), $\Ir_n$ is the identity $n\times n$ matrix and $\bOne := (1,\hdots,1) \in\RR^{n}$.}
    \item $x^i$ denotes the $i$-th component of the vector $x \in \RR^d$. For $A := (A^{ij})_{1\leq i,j\leq n}\in\RR^{n\times n}$, we denote by~$A^\top := (A^{ji})_{1\leq i,j\leq n}\in\RR^{n\times n}$ the transpose matrix, and $\lambda(A)$ denotes the spectrum of $A$.
    \item For a given finite set $S$, we define as the cardinal of $S$, $\#S$.
    \item For all $x,y\in\RR^d$, we denote the scalar product $x^\top y$, the Euclidean norm~$ | x | := \sqrt{x^\top x}$ and for a matrix~$M\in\RR^{d\times d}$, we denote
\begin{equation*}
    | M |:= \sup_{a\in\RR^d, |a| \leq 1}   |Ma|  \label{ct-eq:norm}.
\end{equation*}
\item $(\Omega, \mathcal{H}, \mathbb{P})$ is a complete probability space.

\item  For $p \in [1,\infty]$, ${E}$ is a finite dimensional Euclidian vector space and for a $\sigma$-field $\cH$, $\cL^p(\cH,{E})$, denotes the set of  $\cH$-meassurable random variable $X$ with values in ${E}$ such that $\Vert X \Vert_{p} := \left(\EE\left[ |X|^p\right] \right)^{\frac1p}<\infty$ for $p < \infty$ and for $p = \infty$, $\Vert X \Vert_{\infty} := \mathrm{esssup} |X(\omega)| < \infty$. 
\item For a filtration $\mathbb{G}$, $p \in [1,+\infty]$ and $I\in \NN^*$, $\mathscr{L}^p_{+}(\mathbb{G},(0,\infty)^I)$ is the set of continuous-time processes that are $\mathbb{G}$-adapted valued in $(0,\infty)^I$ and  which satisfy
\begin{equation*}
    \lVert X_t \rVert_p < \infty \text{ for all } t\in\RR_+.
\end{equation*}
\item If $X$ and $Y$ are two random variables $\RR^d$-valued, for $x\in\RR^d$, we note $Y|X=x$ the conditional distribution of $Y$ given $X=x$, and $Y|\Ff$ the conditional distribution of 
$Y$ given the filtration~$\Ff$.
\item {If X is a log-Normal random variable, we note $X \sim \cLN(\mu, \sigma^2)$ with $\mu := \EE[\log(X)]$ (the mean of $\log(X)$)  and $\sigma^2 := \VV[\log(X)]$ (the variance of $\log(X)$).}
\item {For $x,y\in\RR$, we note~$\Phi_2(x, y; \rho)$ is the cumulative distribution function of the bi-variate Gaussian vector $(X, Y)$ with correlation $\rho$ on the space $[-\infty, x]\times[-\infty, y]$. }
\item {For $A \subset \RR$ and $B\subset \RR^I$, if $f: A\to B, t\mapsto f(t)$ is a differentiable function, we note $\dot f$ its first derivative. We also note $\mathfrak{D}(A,B)$ the set of derivative functions from $A$ to $B$.}
\end{itemize}

\section{Introduction}
{There is now evidence that climate change has adverse effects on the environment,
on global warming, and on human societies. Climate risk has two components: \textit{physical risk} which arises from alterations in climate variables, like increasing temperatures and melting ice, or from extreme weather events, such as droughts or typhoons that can
damage infrastructure and endanger populations; and \textit{transition risk}
which comes from the necessity of transitioning to a low-carbon economy, leading to
the implementation of regulatory/political measures, potential technological changes, and the
evolution of consumer preferences. We focus in this work on \textit{transition risk} and particularly on  political changes. The most
well-known way to that is the carbon price. This means that the greenhouse gases (GHG) emissions of
economic agents are charged. For instance, the carbon price can modify the
three main components of credit risk, which are the \textit{borrower’s cash flows}, the \textit{value
of his assets}, and the \textit{value of his guarantees}. The aim of this work is to compute the losses of a credit portfolio when the bank, its borrowers, and their guarantees operate in an economy that undergoes the climate transition.} 

When an obligor (firm, government, or individual) defaults,  the creditor (bank) stands to lose its money. One way to ensure the stability of the banks' business and more generally the soundness of the whole financial system, is ideally, {to anticipate when the default will happen and the amount a bank could lose.} In order to achieve that, the~\cite{basel2017ead} introduces four parameters: the probability of default (PD) which measures the default risk associated with each borrower, the exposure at default (EAD) which quantifies the outstanding debt at the time of default, the loss given default (LGD) which captures the expected percentage of EAD that is lost if the debtor defaults, and the effective maturity T which represents the duration of the credit. By using these parameters, banks can compute various risk measures (such as expected, unexpected, and stressed losses) which help later on to determine provisions, as well as economic and regulatory capital. Using \cite{gordy2003risk}, we can obtain the losses as a function of PD, LGD, EAD, and T. An essential part of a bank risk division is to estimate how the risk measures change with various factors such as time and economic conditions.

{In \cite{bouveret2023propagation}, we propose a methodology in discrete time without collateral (i.e., with LGD independent of the economy and climate transition). However,} when a debtor defaults, the bank can lose all or part of its exposure. The fraction of the loss relative to EAD is LGD while the recovery rate is the fraction of EAD recovered so that $LGD = 1 - Recovery$. {So modeling LGD or recovery is equivalent.} According to~\cite{chalupka2008modelling}, there are three types of LGD: "\textit{market LGD} is
observed from market prices of defaulted bonds or marketable loans soon after the actual default event. \textit{Workout LGD} is derived from a set of estimated cash flows resulting from a workout and collection process, properly discounted to a date of default. Thirdly, implied
\textit{market LGD} is derived from risky but not defaulted bond prices using a theoretical asset pricing model". When consider the IRB approach\footnote{{Internal Rating Based (see~\cite{basel2017ead})}}, LGD refers to \textit{Workout LGD} and there are several techniques to model it. In economic modeling, as detailed by~\cite{bastos2010forecasting, roncalli2020handbook}, LGD is a function of different factors which can be external to the issuer, specific to the issuer, or specific to the debt issuance. That function can be obtained through logistic regression, regression trees, or neural networks. In stochastic modeling, it is assumed that LGD follows a given (parametric or non-parametric) distribution. For example, LGD is commonly modeled by a Beta distribution as \cite{roncalli2020handbook}[Page 193] and \cite{chalupka2008modelling}.
\cite{fermanian2020dependence}, for his part, proposes a joint modeling of PD and LGD by writing the potential loss at default as the difference between the debt amount (EAD) and the assets at the default date.

{In a credit portfolio, there are secured (or collateralized ) and unsecured loans. We call secured loans when there is collateral (also called guarantee) and unsecured loans when there is not. Of course, not all borrowers put up collateral when taking out loans, and there is evidence that loans with collateral may be riskier for lenders. \cite{berger1990collateral} note that banks require more collateral from riskier borrowers, and although seizing a guarantee when a borrower defaults reduces the bank’s loss, this does not always fully compensate the bank for the fact that the loan was initially riskier.} For a secured loan, when the counterpart defaults, the bank liquidates the collateral, and if the EAD is not reached, it can recover the remaining amount by liquidating other assets (called residual recovery). These guarantees can be tangible assets (buildings, business equipments, inventories, etc.) or intangible assets (cash deposits, public bonds, securities, etc.) as noted by~\cite{berger1990collateral}, \cite{blazy2013banks}. For secured loans, the recovery (that is, $1-\LGD$) is therefore made up of both the value of the collateral at the date of default and the value of the residual recovery \cite{frontczak2015modeling}, and~\cite{pelizza2020pricing}. {However, in most LGD modeling approaches, the presence or absence of the guarantees is often overlooked. We will tackle these limits here by considering two examples of guarantees: either a security or a (commercial or residential) building, which both will be affected by the climate transition. }

A security can represent ownership in a corporation in the form of stock, a creditor relationship with a governmental body, or a corporation represented by owning that entity's bond; or rights to ownership as represented by an option. A security generates a stream of cash flows. The proxy of the security value is the (in)finite sum of the present value of the future cash flows. As we already know from what we propose in~\cite{bouveret2023propagation}, the value of the security (particularly if it is a firm) will be affected by the transition risk. We will therefore revisit the results of~\cite{bouveret2023propagation} where carbon price dynamics affects the firm value and credit risk measures such as probability of default, expected and unexpected losses. In particular, we redesign the multisectoral model with carbon price, the firm valuation model, and the credit risk model proposed in continuous time. 

In the same way, a commercial or residential building price will be affected by the climate transition, for example through Energy Performance (or Energy Efficiency) as mentioned in~\cite{aydin2020capitalization, franke2019energy}. \cite{ter2021german} quantify the depreciation by writing the price difference per square meter between two properties with different energy efficiency as the sum of the discounted value of the (expected) energy cost differences. \cite{sopgoui2024realestate} enhances this work by assuming that initially the owner of the building incurs additional energy costs (written as a function of the carbon price) due to the inefficiency of his property. Then, he may decide to spend money on renovations to make his building energy efficient. After the renovation, he no longer incurs additional energy costs. Therefore, he obtains the value of his building as the difference between the price of an equivalent efficient building following an exponential Ornstein-Uhlenbeck as well as the actualized renovation costs and the actualized sum of the future additional energy costs. 

{Thanks to our micro-founded approach to a credit portfolio, we obtain explicit expressions of risk measures (especially probability of default, loss given default, expected and unexpected losses) depending on variables and parameters traditionally found in the literature, such as the the borrower's cash flows, its exposure, the calculation time, the nature of the collateral, the liquidation time, date, and costs, the interest rate, the productivity, etc. Additionally, they now depend on climate transition-specific parameters, such as carbon price, carbon intensities, energy efficiency, or renovation costs. In summary, our framework enables the calculation of the potential loss of a collateralized loan and also measures the variation in this loss depending on the borrower's and the collateral's virtuousness in terms of GHG emissions. Precisely, we will show that (1) expected and unexpected losses increase when the price of carbon increases, (2) the presence of collateral significantly reduces expected and unexpected losses, and (3) the positive effect of collateral on losses is reduced if the collateral is energy inefficient (for a building) or depends on a polluting sector (financial asset). And at each stage, variations can be precisely quantified.}

The rest of the present work is organized as follows. {In \cref{sec:problem}, we clarify the problem that this paper aims to solve.} We revisit in ~\cref{ct_sec_cont_time} the results of~\cite{bouveret2023propagation} in a continuous time setting, namely a multisectoral economic model with carbon price, a firm valuation model, and a credit risk model. In~\cref{ct-sec:LGD}, we define the loss at default as the difference between EAD and the liquidated collateral, which will help us to define LGD. If the collateral is a \textit{financial asset}, we model it in~\cref{ct-sec:col fin asset} by the continuous time version of the discounted cash flows, where the cash flows {growth} is driven by the instantaneous consumption growth, the instantaneous growth of a carbon price's function and a Brownian motion. If the collateral is a \textit{building}, we will use the housing valuation under the climate transition proposed in \cite{sopgoui2024realestate} to compute the loss of the portfolio. The last sections (\ref{sec:numerical exp} and \ref{sec:discussion}) are dedicated to estimations, simulations, and discussion. 

\section{The problem}\label{sec:problem}

We consider a bank credit portfolio composed of $N\in\NN^*$ firms in a closed economy (in other words, no import and no export). In credit risk assessment, one of the first steps is to create homogeneous sub-portfolios of firms. As we are dealing here with climate transition risk, we would like to classify firms by carbon intensity {-- the quantity of GHG in tons emits for each unit of production (in ton per euro)\footnote{To obtain this, one could simply divide each company's greenhouse gas emissions by its revenue.} --} so that firms with similar carbon intensities belong to a same homogeneous sub-portfolio. It should be noted that in the absence of a climate transition, firms are traditionally clustered in terms of industry, geography, size, and credit rating, for example. 

We thus assume $I\in\NN^*$ ($I\leq N$) homogeneous carbon emission sectors in the economy. Nevertheless, as we rarely have the firm individual carbon emissions/intensities {(especially for small and medium enterprises)}, 
we assume that each company has the carbon intensity of its industry sector.
This amounts to grouping "industry sectors" into $I$  {"carbon intensity sectors"}. From now on, sectors are to be interpreted as {carbon emission sectors}.

\begin{definition}\label{def:subport}
We divide our portfolio into $I$ disjunct sub-portfolios $g_1, \hdots, g_I$ so that each sub-portfolio represents a single risk class and the firms in each sub-portfolio belong to a single carbon emission sectors. From now on, we denote~$\Jj$ the set of sectors with cardinal $I \in\NN^*$.  We also fix $n_i := \min{\{n\in\{1,\hdots,N\} \text{ such that }n\in g_i\}}$ for each~$i\in\Jj$. Therefore, firm~$n_i$ is a representative of the group~$i$.
\end{definition}
We would like to know how the whole portfolio loss and sub-portfolios losses would be affected should the regulator introduce a carbon price in the economy, in order to mitigate the effects of climate change. This precisely amounts to quantifying the distortion over time of credit risk measures created by the introduction of a carbon price. 
For example, if the government decides to charge firms and households GHG emissions between 2025 and 2035, a bank would like to estimate today how the probability of a company to default in 2030 is impacted. 

The bank's potential loss caused by a firm depends essentially on the default date and on the liquidation of the guarantees if they exist. The firm as well as the guarantee belong to the same economy which is subject to the climate transition. Thus, we build in the first stage a dynamic, stochastic, and multisectoral economic model in which direct and indirect GHG emissions from companies as well as direct GHG emissions from households are charged. We choose a representative firm in each sector and a representative household for the whole economy. 
By observing that each firm belongs to a sector and its cash flows are a proportion of its sales. The latter are themselves a proportion of the sectoral output. We obtain the cash flows dynamics that we use to determine the value of firms in an environment where GHG emissions are charged. Then, starting from a default model in which a company defaults if its value falls below its debt, we calculate the probability of default of each firm. Finally, considering that the collateral is a GHG-emitting asset in the economy, we compute the distortion of the (associated statistics of) loss -- defined as the difference between the exposure and the liquidated collateral -- by the introduction of a carbon price.

\section{Main assumptions and results of~\cite{bouveret2023propagation} in continuous time}\label{ct_sec_cont_time}

{In~\cite{bouveret2023propagation}, we propose an end-to-end methodology, starting from a transition scenario modeled by the carbon price and carbon intensities, to the impact on different credit risk measures -- especially the probability of default, the expected and unexpected losses -- with both EAD and LGD deterministic and without collateral. }
In this section, we revisit that framework in continuous time. Precisely, we decline, in continuous time, the two standing assumptions as well as the three main results respectively on the dynamic stochastic multisectoral model with carbon emissions costs, on the firm valuation model, and on the structural credit risk model. Most of the proofs can be derived from the discrete time so we will skip them or detail them in \cref{ct-app-sec:multisec}.

\subsection{A Multisectoral Model with Carbon price}\label{subsec:multisec model}

{The goal here is to derive the dynamics of macroeconomic variables (GHG emissions, output, consumption, labor, and intermediary inputs) per sector and per climate transition scenario.
}\\

 Each sector $i\in\Jj$ has a representative firm which produces a single good, so that we can associate sector, firm, and good. We introduce the following standing assumption which describes the productivity, which is considered to have stationary Ornstein-Uhlenbeck dynamics.
\begin{sassumption}\label{ct-sassump:OU}
    We define the $\RR^I$-valued process~$\mathcal{A}$ which evolves  according to 
\begin{equation}\label{ct-eq:VAR}
     \left\{
     \begin{array}{rl}
      \dr\cZ_t &= -\Gamma\cZ_t \dr t + \Sigma \dr B_t^{\cZ}\\
     \dr \cA_t &= \left(\mu + \varsigma \cZ_t\right) \dr t
     \end{array}\quad\textrm{for all } t\in\RR_+,
     \right.
     \end{equation}
     where {$(B_t^{\cZ})_{t\in\RR_+}$} is a $I$-dimensional standard Brownian Motion, and where the constants $\mu, \mathcal{A}_0 \in \RR^I$, the matrices~$\Gamma,\Sigma \in\RR^{I\times I}$, $\cZ_0 \sim \cN\left(0,  \Sigma \Sigma^\top \right)$, and $0 < \varsigma \le 1$ is an intensity of noise parameter that is fixed: it will be used later to obtain a tractable proxy of the firm value. Moreover, $\Sigma$ is a positive definite matrix and $-\Gamma$ is a Hurwitz matrix i.e. its eigenvalues have strictly negative real parts.
\end{sassumption}
The processes $\cZ^i$ and $\cA^i$
play a major role in our factor productivity model since, for any $i\in\Jj$, 
the total factor productivity of sector~$i$
is defined as $A^i := \exp{(\cA^i)}$, so that $\cZ^i$ is the log-productivity growth and $\cA^i$ is the cumulative log-productivity growth.
In the rest of the paper, the terminology "productivity" will be used within a context that will allow the reader to understand if the term refers to $\cZ^i$, $\cA^i$, or $A^i$.

We also introduce the following filtration {$\mathbb{G}:=(\mathcal{G}_t)_{t\in\RR_+}$} with $\cG_0 := \sigma(\cZ_0)$ and for $t>0$, $\mathcal{G}_t := \sigma\left(\left\{\cZ_0, B_s^{\cZ}: s\leq t\right\}\right)$. 

\begin{remark}[OU process]\label{ct-rem:VAR1}
We have the following results on OU that we will use later:
    \begin{enumerate} 
        \item According to~\cite{gobet2016perturbation}[Proposition 1], if one assumes that $\cZ_0$ and $B^{\cZ}$ are independent and $\cZ_0$ is square integrable, then, there exists a unique square integrable solution to the $I$-dimentional Ornstein-Uhlenbeck process~$\cZ$ satisfying $\dr\cZ_t = -\Gamma\cZ_t \dr t + \Sigma \dr B_t^{\cZ}$, represented as
        \begin{equation*}
            \cZ_t = e^{-\Gamma t} \left(\cZ_0 + \int_{0}^{t} e^{\Gamma u} \Sigma \dr B_u^{\cZ} \right),\quad\textrm{for all } t\in\RR_+.
        \end{equation*}
         Additionally, for any $t,h \geq 0$, the distribution of $\cZ_{t+h}$ conditional on $\cG_t$ is Gaussian~$\cN\left(M^{\cZ,h}_{t}, \Sigma^{\cZ,h}_{t}\right)$, with the mean vector
        \begin{equation}
            M^{\cZ,h}_{t} := \EE[\cZ_{t+h}|\cG_t] = e^{-\Gamma h} \cZ_t,
        \end{equation}
     and the covariance matrix
     \begin{equation}
         \Sigma^{\cZ,h}_{t} := \VV[\cZ_{t+h}|\cG_t] = \int_{0}^{h} e^{-\Gamma u} \Sigma \Sigma^\top e^{-\Gamma^\top u} \dr u.
     \end{equation}
     \item Since $-\Gamma$ is a Hurwitz matrix, then if we note $\lambda_\Gamma :=\max_{\lambda\in\lambda(\Gamma)} Re(\lambda)$, there exists $c_\Gamma>0$ so that {$| e^{-\Gamma t}| < c_\Gamma e^{-\lambda_\Gamma t}$} for all~$t\geq 0$. Therefore, according to~\cite{gobet2016perturbation}[Proposition 2], $\cZ$ has a unique stationary distribution which is Gaussian with mean $0$ and covariance $\int_{0}^{+\infty} e^{-\Gamma u} \Sigma \Sigma^\top e^{-\Gamma^\top u}\dr u$. 
    \item We show in~\cite{sopgoui2024realestate}[Appendix A] that for any $t, h \geq 0$, we have
        \begin{equation*}
            \cA_{t+h} = \cA_{t} + \int_{t}^{t+h} (\mu + \varsigma \cZ_s) \dr s = \mu h + \varsigma \int_{t}^{t+h} \cZ_s \dr s,
        \end{equation*}
        and conditionally on $\cG_t$, $\cA_{t+h}$ has an $I$-dimensional normal distribution with the mean vector $M^{\cA,h}_{t} :=  \mu h + \varsigma\Upsilon_{h}\cZ_t +  \cA_t$,
    with \begin{equation}\label{ct-eq:Upsilon}
    \Upsilon_{h} := \int_{0}^{h} e^{-\Gamma s} \dr s = \Gamma^{-1}(\Ir_I-e^{-\Gamma h}),
\end{equation}
     and the covariance matrix
     \begin{equation}\label{ct-eq:Ma_ht}
         \Sigma^{\cA,h}_{t} := \varsigma^2 \Gamma^{-1} \left(\int_{0}^{h} \left(e^{-\Gamma u} - \Ir_I \right) \Sigma\Sigma^\top \left(e^{-\Gamma u} - \Ir_I \right) \dr u \right) (\Gamma^{-1})^\top = \varsigma^2\int_{0}^{h} \Upsilon_{u} \Sigma\Sigma^\top \Upsilon_{u}^\top \dr u .
     \end{equation}
    \item For later use, we define $\cA^\circ_t := \mathcal{A}_t - \mathcal{A}_0$,
        and observe that $(\cA^\circ_t,\cZ_t)_{t \ge 0}$ is a Markov process.
    \end{enumerate}
\end{remark}

Firms emit GHG when they consume intermediary input from other sectors and when they produce output. Likewise, households emit GHG when they consume. All these emissions are charged through a carbon price dynamics. For the whole economy, we introduce a deterministic and exogenous carbon price in {euro per ton of
CO$_2$-equivalent}.
It allows us to model the impact of the transition pathways on the whole economy. 
We shall then assume the following setting. 
\begin{sassumption}\label{ct-sassc:price}
We introduce the carbon price and the carbon intensities (the quantity of GHG in tons emits for each unit of production/consumption) processes:
\begin{enumerate}
\item Let $0 \le t_\circ < t_\star$ be given. 
The complete carbon price sequence noted $\delta$ satisfies
{\begin{itemize}
    \item for $t \in [0;t_\circ]$, $\delta_t = \delta_0\in \RR_+$, namely the carbon price is constant;
    \item for $t \in (t_\circ,t_\star)$, $\delta_t \in \RR_+$, the carbon price may evolve;
    \item for $t \ge t_\star$, $\delta_t = \delta_{t_\star} \in \RR_+$, namely the carbon price is constant.
\end{itemize}}
\noindent We assume moreover that $t\mapsto \delta_t$ is $\mathfrak{D}(\RR_+, \RR_+)$.
\item We also introduce carbon intensities as the sequences $\tau$, $\zeta$, and $\kappa$ being respectively $\RR_+^I$, $\RR_+^{I\times I}$, and $\RR_+^I$-processes, and representing respectively carbon intensities on firm's output, on firm's intermediary consumption, and on household’s consumption, and satisfying for all $t\in\RR_+$ and \\ $\mathfrak{y} \in\{\tau^1, \hdots, \tau^I, \zeta^{11}, \zeta^{12},\hdots, \zeta^{I I-1}, \zeta^{II}, \kappa^1, \hdots, \kappa^I\}$,
\begin{equation}\label{ct-eq:ghg_intensity}
\mathfrak{y}_t = \left\{
\begin{array}{ll}
 & \mathfrak{y}_0 \exp{\left(g_{\mathfrak{y}, 0} \frac{1-\exp{(-\theta_\mathfrak{y} t)}}{\theta_\mathfrak{y}}\right)}\qquad \text{if } 0\leq t\leq t_\star \\
 & \mathfrak{y}_0 \exp{\left(g_{\mathfrak{y}, 0} \frac{1-\exp{(-\theta_\mathfrak{y} t_\star)}}{\theta_\mathfrak{y}}\right)}
\qquad \text{else},
\end{array}
\right.
\end{equation} 
with $\mathfrak{y}_0, g_{\mathfrak{y}, 0}, \theta_\mathfrak{y} > 0$. For each~$t\geq 0$, we call $\mathfrak{y}_t\delta_t$ the \textit{emissions cost rate} at time $t$.\\
\item  For each~$i\in\Jj$ and for each~$t\in\RR_+$,
\begin{equation}
     \delta_t \max_{i\in\Jj}{\tau^i_0} < 1.\label{ct-eq:prod_vs_emiss}
\end{equation}
\end{enumerate}
\end{sassumption}
{\begin{remark}
    At each time $t\geq 0$ and for $i\in\Jj$, 
    \begin{itemize}
        \item $\tau_t^i$ represents the carbon intensity of the firms production of sector~$i$ at time $t$;
        \item $\kappa_t^i$, the carbon intensity of the household final consumption in sector $i$ at time $t$,
        \item and for $j\in\Jj$, $\zeta^{ji}_t$, the carbon intensity of the firms of sector $i$ intermediate consumption in sector $j$ at time $t$.
    \end{itemize}
\end{remark}}
\noindent In the following, we will note for all $t\geq 0$,
\begin{equation}\label{ct-eq:emiss cost rate}
    \dd_t := (\delta_t\tau_t,\delta_t\zeta_t,\delta_t\kappa_t).
\end{equation}

\paragraph{An example of carbon price process}
We assume the regulator fixes $t_\circ\geq 0$ when the transition starts and the transition horizon time~$t_\star>t_\circ$, the carbon price at the beginning of the transition~$P_{carbon} > 0$, at the end of the transition~$\delta_{t_\star} > P_{carbon}$, and the annual growth rate~$\eta_\delta > 0$. 
Then, for all~$t\geq 0$,
 \begin{equation}\label{ct-eq:carbon price}
    \delta_{t}= \left\{
    \begin{array}{ll}
    \displaystyle P_{carbon}, & \mbox{if } t \leq t_\circ,\\
    \displaystyle P_{carbon} e^{\eta_\delta(t-t_\circ)}, & \mbox{if } t\in(t_\circ,t_\star],\\
    \displaystyle \delta_{t_\star} = P_{carbon} e^{\eta_\delta(t_\star-t_\circ)}, &  \mbox{otherwise}.
    \end{array}
\right.
\end{equation} 
 In the example above that will be used in the rest of this work, we assume that the carbon price increases. However, there are several scenarios that could be considered, including a carbon price that would increase until a certain year before leveling off or even decreasing. We also assume an unique carbon price for the entire economy whereas we could proceed differently. For example, the carbon price could increase for production when stabilize or disappear on households in order to avoid social movements and so on. The framework can be adapted to various sectors as well as scenarios.\\

In our framework, a representative firm in each sector which maximizes its profits by choosing, at each time and for a given productivity, the quantities of labor and intermediary inputs, while, a representative household solves a dynamic optimization problem to decide how to allocate its consumption expenditures among the different goods and hours worked and among the different sectors. We assume that the utility function~$U: \RR_+^*\times \RR_+^* \to \RR, (c, h) \mapsto \log{c} - \frac{h^{1+\varphi}}{1+\varphi}$ with~$\varphi \geq 0$; {$c$ and $h$ represents respectively the household consumption and the labor}. Moreover, $\llambda$ (respectively, $\psi$) are matrix in $(\RR_+^*)^{I\times I}$ (respectively vector in~$(\RR_+^*)^I$) of the elasticities of intermediary inputs (respectively labor). We also assume a constant return to scale, namely
\begin{align}\label{ct-eq const ret to scale}
\psi^i + \sum_{j \in \mathcal{I}} \llambda^{ji} = 1,
\qquad\text{for each }i \in \mathcal{I}.
\end{align}
 Since the productivity and the carbon price processes are continuous, the firms and households problems are well posed and their solutions exist. More details are given in~\cref{ct-app-sec:multisec}. In the following proposition, we give an explicit expression of the output.
\begin{proposition} \label{ct-cor:output_consc}
For 
$(\overline\tau, \overline\zeta, \overline\kappa, \overline\delta) \in \RR_+^I\times \RR_+^{I\times I} \times \RR_+^I \times \RR_+$, let us note
\begin{equation}\label{ct-eq:Psi_Lambda}
    \Psi(\overline\dd) := \left( \psi^i \frac{1 - \overline \tau^i\overline\delta}{1 + \overline \kappa^i\overline\delta}
      \right)_{i \in \Jj}\ \quad\text{and}\quad 
    \Lambda(\overline \dd) := \left( \llambda^{ji} \frac{1 - \overline \tau^{i}\overline\delta}{1+\overline \zeta_t^{ji}\overline\delta} \frac{1 + \overline \kappa^{j}\overline\delta}{1 + \overline \kappa^{i}\overline\delta}
    \right)_{j,i \in \Jj},
\end{equation}
with $\overline{\dd} := (\overline\delta\overline\tau, \overline\delta\overline\zeta, \overline\delta\overline\kappa)$. Assume that 
\begin{enumerate}
    \item $\Ir_I - \llambda$ is not singular,
    \item $\Ir_I - \Ll(\dd_t)^\top$ is not singular for all $t\in\RR_+$.
\end{enumerate}
Then, for all~$t\in\RR_+$, there exists an unique couple of consumption and output~$(\Cc_t,\Yy_t)$ solving the (dynamic stochastic) multisectoral model.
Moreover, for all $t\in\RR_+$.
\begin{enumerate}
    \item if $\ee_{t}^i := \frac{Y^i_t}{C^i_t}$ for $i\in\Jj$, we have 
\begin{align}
    \ee_{t} = \ee(\dd_t) := (\Ir_I-\Ll(\dd_t)^\top)^{-1} \bOne. \label{ct-eq:eec}
\end{align}
    \item For 
$(\overline\tau, \overline\zeta, \overline\kappa, \overline\delta) \in \RR_+^I\times \RR_+^{I\times I} \times \RR_+^I \times \RR_+$,
\begin{equation} 
    v^i(\overline \dd) := \log\left((\ee(\overline \dd)^{i})^{-\frac{\varphi\psi^i}{1+\varphi}}   \left(\Psi^i(\overline \dd) \right)^{\frac{\psi^i}{1+\varphi}} \prod_{j\in\Jj}  \left(\Lambda^{ji}(\overline \dd)\right)^{\llambda^{ji}} \right)_{i\in\Jj} + ((\Ir_I-\llambda)\log{(\ee(\overline\dd))}),\label{ct-eq:pricefuncc}
\end{equation}
We obtain
\begin{equation}
    \Yy_{t} = \exp{\left( (\Ir_I-\llambda)^{-1}{\left( \cA_{t} + v(\dd_t)\right)}\right)}.\label{ct-eq:consumptionc}
\end{equation} 
    \item Furthermore, since $\dd \in \mathfrak{D}(\RR_+, [0, 1)^I\times(\RR_+)^{I\times I} \times (\RR_+)^I)$, we directly have $\Psi(\overline \dd_{\cdot}), \Lambda(\overline \dd_{\cdot}) \in\mathfrak{D}(\RR_+,\RR)$. Moreover, $\bar\dd \mapsto (\Ir_I-\Ll(\bar\dd)^\top)^{-1}$ on $\RR_+^I\times \RR_+^{I\times I} \times \RR_+^I$) is differentiable, then $(\Ir_I-\Ll(\dd_\cdot)^\top)^{-1}) \in\mathfrak{D}(\RR_+,\RR)$.
\end{enumerate}
\end{proposition}
 The output $\Yy$ is also positive, we can then introduce, from the third item, for all $t\geq 0$, $\dr\log{(\Yy_{t})}$ representing the instantaneous consumption growth. This proposition is an equivalent of \cite{bouveret2023propagation}[Theorem 1.10.] in continuous time. In the same spirit, we can also define and determine 
\begin{itemize}
    \item consumption, intermediary inputs, labor, and GHG emissions at each time as well as their logarithmic growth between two instants. 
    \item The sensitivities of the growth of variables to the carbon price. This allows us to answer the following types of questions: \textit{what is the impact on all variables if in the whole economy, we only charge GHG emissions coming from the production of companies in sector $1$?} The main difficulty here comes from the fact that the coefficients of the inverse of the elasticity matrix~$(\llambda^{ji})_{i,j\in\Jj}$ are not known. 
    \item The law of the logarithmic growth of macroeconomic variables that we obtain by using the Gaussianity and stationarity of logarithmic growth. 
    \item the evolution of the contribution of each sector in production/consumption as well as the evolution of GHG emissions.
\end{itemize}

\subsection{A Firm Valuation Model}
{The goal in this section is to describe the dynamics of the value of a firm operating in an economy that undergoes the climate transition described by both the carbon price and carbon intensities.}\\

Consider a firm $n\in\OneN$. For any time $t\in\RR_+$ and firm~$n$, we note~$F^n_{t}$ the free cash flows of~$n$ at $t$, and $r>0$ the discount rate, we introduce the following assumption:

\begin{assumption}\label{ct-ass:link}
The $\RR$-valued process on the instantaneous growth of the cash flows of firm~$n$ denoted by~$\dr \log{F^n_{t,\dd}}$  is linear in the economic factors (output growth by sector), 
 specifically we set for all $t\in\RR_+$,
\begin{equation} \label{ct-eq:CF_vs_GDPGrowth with conso}
    \dr \log{F^n_{t,\dd}} = \tilde{\af}^{n\cdot} \dr\log{\Yy_{t}}+ \sigma_{n} \dr \cW^n_t = \af^{n\cdot} (\dr \cA_t + \dr v(\dd_t)) + \sigma_{n} \dr \cW^n_t,
\end{equation}
for $\tilde{\af}^{n\cdot}\in\RR^{I}$ and $\af^{n\cdot} = \tilde{\af}^{n\cdot}(\Ir_I-\llambda)^{-1}$, where $(\cW^n_t)_{t\in\RR_+}$ is a $\RR^{N}$-Brownian motion with $\sigma_{n} > 0$. 
Moreover, $B^{\cZ}$ and $\cW^n$ are independent.
\end{assumption}
\noindent We call $\tilde{\mathfrak{a}}^{n\cdot} = (\tilde{\mathfrak{a}}^{ni})_{1\leq i\leq I}$ (respectively $\mathfrak{a}^{n\cdot} = (\mathfrak{a}^{ni})_{1\leq i\leq I}$)  \textit{factor loadings}, quantifying the extent to which~$\dr \log{F^n_{,\dd}}$ is related to~$\dr\log{\Yy}$ (respectively $\dr\cA$). We define the filtration~$\mathbb{F}=(\cF_t)_{t\geq 0}$ by $\cF_{t} = \sigma\left(\cG_t \cup  \sigma\left\{ \bb_{s}: s\in[0, t]\right\}\right)$ for $t\geq 0$, and we denote~$\EE_{t}[\cdot] := \EE[\cdot|\cF_{t}]$.

Recall that the economic motivation behind~\eqref{ct-eq:CF_vs_GDPGrowth with conso} comes from the fact that if firm~$n$ belongs to sector~$i$, then its production is proportional to the sectoral output and its cash flows are proportional to its production (as in the Dechow-Kothari-Watts model in~\cite{barth2001accruals}). Thus, we obtain a relation between the cash flows of firm~$n$ and the total output of sector~$i$. 
The assumption $\tilde{\mathfrak{a}}^{n\cdot}\in\RR^I$ stems from the fact that a company is not restricted to one sector only in general.
However, since we are considering the emission sector here, 
we expect that each firm~$n$ only belongs to one sector 
($i$ for example).
Therefore $\mathfrak{a}^{n j} = 0$ for all~$i\neq j$
and hence $|\mathfrak{a}^{n i}| = \max_{j\in\Jj} |\mathfrak{a}^{n j}|$.

Let $r\geq 0$ representing the interest rate, by the continuous form of the discounted cash flows valuation, the value~$V^{n}_{t,\dd}$ of the firm~$n$, at time~$t$, is 
\begin{equation}
    V^n_{t,\dd} := \EE_{t}\left[\int_{t}^{+\infty} e^{-r (s-t)} F^n_{s,\dd} \dr s \right].\label{ct-eq:DCF}
\end{equation}
For $n\in\OneN$ and $t\geq 0$, describing $V^n_{t,\dd}$ as a function of the underlying processes driving the economy
does not lead to an easily tractable formula. But this allows us to write it as a fixed-point problem which can be solved by numerical methods such as deep learning methods (see~\cite{hammad2022new, kobeissi2022temporal}) or Picard iteration (see~\cite{berinde2007iterative}). However,
to facilitate the forthcoming credit risk analysis, when $\varsigma$ (introduced in \cref{ct-sassump:OU}) is close to $0$, we approach~$V^n_{t,\dd}$ by the quantity
\begin{equation}\label{ct-eq de cV}
    \cV_{t,\dd}^n := F^n_{t,\dd} \int_{t}^{+\infty} e^{-r(s-t)}\EE_{t}\left[ \exp \left( (s-t)\af^{n\cdot} \mu
    +\af^{n\cdot}\left(v(\dd_{s})-v(\dd_t)\right)
    + \sigma_{n} (\cW^n_s - \cW^n_t)
    \right)\dr s\right],
\end{equation}
which is defined using the first order expansion of $\frac{V^n_{t,\dd} }{F^n_{t,\dd} }$ and that we describe as a proxy the firm $n$ value at time $t$. We will work directly with $\cV^n_{t,\dd}$ instead of $V^n_{t,\dd}$. We have the following proposition, whose proof is given in~\ref{ct-proof:lem:approx firm value}.

\begin{proposition}\label{ct-lem:approx firm value}
    For any $n \in \OneN$ and for all $t \in \RR_+$.
    \begin{enumerate} 
        \item Assume that $\varrho_n := \frac{1}{2}\sigma_{n}^2 + \af^{n\cdot}\mu -r<0$,
    then $\cV_{t,\dd}^n$ is well defined and
    \begin{align}\label{ct-eq expression of cV}
        \cV^n_{t,\dd} = F^n_0 \mathfrak{R}^n_t(\dd) \exp{\left(\af^{n\cdot}(\cA_t-v(\dd_{0}))\right)}\exp\left(\sigma_{n}\cW^n_t\right),
    \end{align}
where 
    \begin{equation}
     \mathfrak{R}^n_t(\dd) := \int_{0}^{\infty} e^{\varrho_n s}\exp{\left(\af^{n\cdot} v(\dd_{t+s})  \right)} \dr s
     \label{ct-eq:FV_second_term}.
    \end{equation}
        \item Moreover, with $t_\circ$ and $t_\star$ defined in \cref{ct-sassc:price}, we obtain the following explicit form,
    \begin{equation}
    \mathfrak{R}^n_t(\dd) = \left\{
    \begin{array}{ll}
    \displaystyle -\frac{e^{\af^{n\cdot} v(\dd_{t_\star})}}{\varrho_n}, & \mbox{if } t \geq t_\star,\\
    \displaystyle \int_{0}^{t_\star-t} e^{\varrho_n s}\exp{\left(\af^{n\cdot} v(\dd_{t+s})  \right)}\dr s - \frac{e^{\af^{n\cdot} v(\dd_{t_\star}) + \varrho_n (t_\star-t)}}{\varrho_n}, & \mbox{if } 0 \leq t < t_\star,.
    \end{array}
\right.
\end{equation}   
\item Assume that 
        {\begin{align} \label{ct_eq:main technical ass}
    \rho_n := \frac{1}{2}\sigma_{n}^2+\af^{n\cdot}\mu+\frac{1}{2}  \varsigma^2\frac{c_\Gamma^2}{\lambda_\Gamma^2}|\af^{n\cdot}|^2 |\Sigma|^2 < r,
\end{align}}
therefore $V_{t,\dd}^n$ is well defined and there exists a constant~$C$ such that $\EE \left[\left|\frac{V^n_{t,\dd}}{F^n_{t,\dd}} - \frac{\cV^n_{t,\dd}}{F^n_{t,\dd}}\right|\right] \le C \varsigma$, for all~$\varsigma > 0$.
    \end{enumerate}
\end{proposition}
{In $(iii)$ of \cref{ct-lem:approx firm value}, $\EE \left[\left|\frac{V^n_{t,\dd}}{F^n_{t,\dd}} - \frac{\cV^n_{t,\dd}}{F^n_{t,\dd}}\right|\right] \le C \varsigma$ means that when the noise term in the productivity process is small, the terms $\frac{\cV^n_t}{F^n_t}$ and $\frac{V^n_t}{F^n_t}$ become closer.} The following corollary gives (conditional) laws of the (proxy) of the firm value~$\cV_{\cdot,\dd}^n$.
\begin{corollary}\label{ct-cor:law Vt} 
For all $t, T\geq 0$.
\begin{enumerate}
    \item We note $\mathfrak{m}^n(\dd,t,\cA_{t}^\circ) :=  \log{(F^n_0 \mathfrak{R}^n_t(\dd))} +  \left(\af^{n\cdot}(\cA_t^\circ -v(\dd_{0}))\right)$ and we have
    \begin{equation}
        \log{\cV_{t,\dd}^n}|\cG_t \sim \cN\left(\mathfrak{m}^n(\dd,t,\cA_{t}^\circ), t\sigma_{n}^2 \right).
    \end{equation} 
    
    \item We note $\cK^n(\dd, t, T, \cA_t^\circ, \cZ_t) := \log{(F^n_0 \mathfrak{R}^n_{t+T}(\dd))} + \af^{n\cdot}(\mu T + \varsigma\Upsilon_{T}\cZ_t + \cA_t^\circ-v(\dd_{0}))$ and $\mathcal{L}^n(t,T) := \af^{n\cdot}\Sigma^{\cA, T}_{t}\af^{n\cdot} + (t+T) \sigma_{n}^2$, we have 
    \begin{equation}
       \log{\cV_{t+T,\dd}^n}|\cG_t \sim \cN\left(\cK^n(\dd, t, T, \cA_t^\circ, \cZ_t), \mathcal{L}^n(t,T)\right).
    \end{equation}
\end{enumerate}
\end{corollary}

\subsection{A Credit Risk Model without collateral}\label{subsec:risk model}
{In this section, we aim to calculate the probability of default of the previous firm (i.e. operating in an economy that undergoes the climate transition described by both the carbon price and carbon intensities).}\\

As~\cite{basel2017ead}, we introduce four credit risk parameters:
the probability of default (PD) measures the default risk associated with each borrower,
the exposure at default (EAD) measures the outstanding debt at the time of default,
the loss given default (LGD) denotes the expected percentage of EAD that is lost if the debtor defaults, 
and the effective maturity~$T$ represents the duration of the credit.
With these four parameters, we can compute the portfolio loss~$L$, with some assumptions:
\begin{assumption}\label{ct-ass:portfolio}
Consider a portfolio of $N\in\mathbb{N}^*$ credits. 
For $1 \leq n \leq N$, 
\begin{enumerate}[(1)]
    \item Firm~$n$ has issued two classes of securities: equity and debt.
        \item \label{item:EAD} $(\EAD_t^n)_{t\geq 0}$ is a $\RR_+^*$-valued continuous and  deterministic process, and for all~$t\geq 0$,
         the family $(\EAD_t^n)_{n=1,\ldots, N}$ is a sequence of positive constants such that
    \begin{enumerate}
        \item $\displaystyle\sum_{n\geq 1} \EAD_t^n = +\infty$;
        \item there exists $\upsilon > 0$ such that $\frac{\EAD_t^n}{\sum_{n=1}^{N} \EAD_t^n} = \mathcal{O}(N^{-(\frac{1}{2}+\upsilon)})$, as~$N$ tends to infinity.
    \end{enumerate}
    \item $(\LGD_{t}^n)_{t\geq 0}$ is a $(0, 1]$-valued continuous and deterministic process; \label{item:LGD assump}
    \item $(\cD^n_t)_{t\geq 0}$ is a $\RR_+$-valued continuous and deterministic process, representing the debt of firm~$n$ at time~$t$. We will also denote $D^n_t := \displaystyle \frac{\cD_t^n}{\EE[F_{t,0}^n]}$ representing the debt to cash flows ratio.
    \item The value of the firm~$n$ at time~$t$ is assumed to be a tradable asset given by~$V^n_{t,\dd}$ defined in~\eqref{ct-eq expression of cV}.
\end{enumerate}
\end{assumption}
According to~\cite{kruschwitz2020stochastic}, there are two ways to handle the default of a company: for a given financing policy, a levered firm is
\begin{itemize}
    \item \textit{in danger of illiquidity} if the cash flows do not suffice to fulfill the creditors’ payment claims (interest and net redemption) as contracted,
    \item \textit{over-indebted} if the market value of debt exceeds the firm’s market value.
\end{itemize}

We follow in the present work the second definition of default proposed: a firm default when it is \textit{over-indebted}, that is in fact the same approach used in the structural credit risk models. We retain this definition in this continuous setting in the same way as~\cite{bouveret2023propagation}. Therefore, the \textit{over-indebtedness} of entity~$n$ occurs at time~$t$ when the firm value~$\cV^n_{t,\dd}$ falls below a given barrier~$\cD_t^n$, related to the net debt, namely on the event~$\left\{\cV^n_{t,\dd} < \cD^n_t\right\}$. 

However, it should be noted that in a continuous time setting, it can be interesting  to work in the~\cite{black1976valuing} model. Here, the default event depends on the trajectory
of the firm value process~$\cV$. Therefore, at a given time $t$, the firm defaults if it  has been \textit{over-indebted} at least one time during the period $[0, t]$, that is~$\left\{\exists s\in[0,t]\text{  such that 
 } \cV^n_{s,\dd} < \cD^n_s\right\}$.
Thus, the default time is given by $\tau^n := \inf\left\{t\geq 0,\quad \cV^n_{t,\dd} < \cD^n_t \right\}$.

\noindent Then, if we are interested in the probability of the firm~$n$ defaulting before $t$ conditionally to~$\cG_t$ that is noted~$PD^n_{t,\dd}$, we have
\begin{align*}
    PD^n_{t,T,\dd} = \PP\left(\tau^n\leq t | \cG_t \right) = \PP\left(\inf_{0\leq s\leq t}  \cV^n_{s,\dd} < \cD^n_s\middle| \cG_t \right) = \PP\left(\inf_{0\leq s\leq t}  \log{\cV^n_{s,\dd}} < \log{\cD^n_s}\middle| \cG_t \right).
\end{align*} 
    But for $0\leq s\leq t$ and from~\eqref{ct-eq expression of cV},
    \begin{equation*}
        \log{\cV_{s,\dd}^n} = \log{\left(F_{0,\dd}^n \mathfrak{R}^n_s(\dd)\right)} + \af^{n\cdot} (\cA_s^\circ-v(\dd_0)) + \sigma_{n}\cW^n_s,
    \end{equation*}
    therefore $\log{\cV^n_{\cdot,\dd}}$ is a Gaussian process. However, as~\cite{azais2000distribution} summarizes, the computation of the distribution function of the random variable $\inf_{0\leq s\leq t}  \log{\cV^n_{s,\dd}}$ is by means of a closed formula is known only for a very restricted number of stochastic processes as the Brownian Motion, the Brownian Bridge, the Brownian Motion with a linear drift, and the stationary Gaussian processes with relatively simple with covariance. This is not the case here. {Consequently, we extend \cite{merton1974model} model instead of \cite{black1976valuing} model.}
    
    At each time~$t\geq 0$, we are interested in the probability that firm $n$ is \textit{over-indebted} at a certain date $t+T$, we note $\PD^n_{t,T,\dd}$ and we have
    \begin{equation}\label{ct-eq:def PD}
        \PD^n_{t,T,\dd} := \mathbb{P} \left(\cV^n_{t+T,\dd}\leq \cD^n_{t+T}\middle|\cG_t \right).
    \end{equation}
\noindent We have the following proposition whose proof is a direct application of Corollary~\ref{ct-cor:law Vt}.
\begin{proposition}[Probability of default]\label{ct-pr cond pd} 
For $t\geq 0$, $T\geq 0$, and $n \in \OneN$,
the (conditional) probability of default of the entity~$n$ at time~$t$ over the horizon~$T$ is
\begin{equation}
    \PD^n_{t,T,\dd}  = \Phi\left(\frac{\log(\cD_{t+T}^n) -\cK^n(\dd, t, T, \cA_t^\circ, \cZ_t)}{\sqrt{\mathcal{L}^n(t,T)}}\right),\label{ct-eq: PD t T invest}
\end{equation}
where $\cK^n(\dd, t, T, a,\theta)$ as well as $\mathcal{L}^n(t,T)$ are defined in Corollary~\ref{ct-cor:law Vt} {and $\Phi$ the Gaussian cumulative distribution function.}
\end{proposition}

\noindent The previous results tell us that the probability of the \textit{over-indebtedness} depends on parameters specific to:
\begin{enumerate}
    \item the climate transition
    \begin{itemize}
        \item the carbon price~$\delta$,
        \item the carbon intensities~$\tau, \zeta, \kappa$,
    \end{itemize}
    \item to the company (the contract),
    \begin{itemize}
        \item the factors loading $\af^{n\cdot}$ and the standard deviation of the cash flows~$\sigma_n$,
        \item the time~$t$ when it is computed,
        \item the potential date of the \textit{over-indebtedness}~$t+T$,
        \item the \textit{over-indebtedness}'s barrier~$\cD$,
    \end{itemize}
    \item the economy  to which the company belongs to:
    \begin{itemize}
        \item the productivity $\cZ$ and $\cA$ (and their parameters) of the economy,
        \item the interest rate~$r$.
    \end{itemize}
\end{enumerate}
 {\begin{remark}
     Motivated by the uncertainty of the climate transition, a natural extension is to introduce a stochastic carbon price (see, for example, \cite{le2025corporate}). Regardless of whether it is correlated with the productivity process, the expressions of the macroeconomic variables (for example, the output in \cref{ct-cor:output_consc}) remain unchanged; however, their dynamics do differ. Consequently, primarily due to the form of $v$ defined in \eqref{ct-eq:pricefuncc} -- whose law is a product and quotient of Gaussian random variables --, we could no longer obtain a closed-form expression for the firm value -- because filtration $\GG$ would not only be adapted to productivity $\Theta$, but also to the carbon price $\delta$ -- and then for the probability of default. Numerical methods are therefore required.
 \end{remark}}

{In the next section, we will define and express LGD as a function of some guarantees which are affected by the climate transition. Therefore, by assuming that EAD are deterministic and independent of the carbon price, we will also introduce the expected and unexpected losses. }

\section{LGD with stochastic collaterals in continuous time}\label{ct-sec:LGD}

We are in the same framework as in the previous section, but \cref{item:LGD assump} of \cref{ct-ass:portfolio} is not satisfied anymore, therefore the bank could require from each counterpart~$1\leq n\leq N$ a (single) collateral~$C^n$ to secure its debt. Collateral can  take the form of a \textit{physical asset} such as real estate, business equipment, or inventory, or it can be a \textit{financial asset} such as invoices, cash, or investments. If a firm is \textit{over-indebted} at time~$t$, we assume that the liquidation ends at~$t+a$ with $a\in\RR_+$ where $a$ is the \textit{liquidation delay}. Moreover, $k\in[0, 1)$ represents the fraction of the collateral used to cover liquidations auctions, as well as other legal and administrative procedures. \\

A firm is \textit{over-indebted} at time~$t\geq 0$ if the market value of its debt~$\cD^n_t$ exceed its market value~$\cV^n_{t,\dd}$, namely~$\{\cV^n_{t,\dd}<\cD^n_t\}$. At time $t$, if the company $n$ in the portfolio defaults i.e. $\cV^n_{t,\dd}<\cD^n_t~$, the bank recovers $(1-k)e^{-r a}C^n_{t+a}$ after the collateral liquidation. In general, the liquidations do not cover all the debt, i.e. $\EAD_{t}^n\geq (1-k)e^{-r a}C^n_{t+a}$, the bank deploys further actions to recover an additional fraction. We note that fraction~$\gamma\in[0, 1)$. Therefore, the bank recovers~$\gamma (\EAD_{t}^n-(1-k)e^{-r a}C^n_{t+a})_{+}$ by other tools.\\

The potential loss that would be recorded due to the firm default event is the difference between the debt amount $\EAD_t^n$ and the amount received after the recovery processes on the time horizon. Consequently, if there is not default ($\cV^n_{t,\dd}\geq \cD^n_t$) or there is default and the collateral liquidated exceed the exposure ($(1-k)e^{-r a}C^n_{t+a}\geq \EAD_t^n$), the loss is zero, and if there is default and if the exposure exceed the collateral liquidated, the loss is
\begin{equation}
\EAD_t^n - (1-k)e^{-r a}C^n_{t+a} - \gamma (\EAD_t^n - (1-k)e^{-r a}C^n_{t+a}) = (1-\gamma)\left(\EAD_t^n-(1-k)e^{-r a}C^n_{t+a}\right),
\end{equation}
where the constant~$r$ is the discount rate.  The loss, $L^N_{t}$, of the portfolio at time~$t$, is in fact, defined as
\begin{equation}
L^N_{t} := \sum_{n=1}^{N} (1-\gamma)(\EAD_{t}^n-(1-k)e^{-r a}C^n_{t+a})_{+}\cdot \bOne_{\{ \cV^n_{t,\dd}<\cD^n_t\}} \label{ct-def:ptloss}.
\end{equation}

The following result is similar with the one introduced in~\cite{bouveret2023propagation}[Theorem 3.5]. It gives a proxy of the loss of the portfolio.

\begin{theorem}[Definition of LGD] \label{ct-theo:gordy2003}
    For all~$t\in\RR_+$, define 
    \begin{equation}\label{ct-eq:def loss}
         \mathrm{L}^{\GG,N}_{t} := \EE\left[L^N_{t}\middle|\cG_t\right] =\sum_{n=1}^{N} \EAD_t^n \cdot \LGD_{t,\dd}^n \cdot \PD_{t,\dd}^n, 
    \end{equation}
    where
\begin{subequations}
\begin{align}
\displaystyle \PD_{t,\dd}^n &:= \PP\left( \cV^n_{t,\dd}<\cD^n_t|\cG_t \right) = \Phi\left( \frac{\log(\cD_{t}^n) -\mathfrak{m}^n(\dd,t,\cA_{t})}{\sigma^n\sqrt{t}}\right),\\
    \displaystyle \LGD_{t,\dd}^n &:= (1-\gamma) \EE\left[ \left(1 - (1-k)e^{-r a} \frac{C^n_{t+a}}{\EAD_{t}^n}\right)_{+}\middle|\cV^n_{t,\dd}<\cD^n_t,\cG_t\right].\label{ct-eq:def LGD}
\end{align}
\end{subequations}
Under Assumptions~\ref{ct-ass:portfolio}, we have $L^N_{t}- \mathrm{L}^{\GG,N}_{t}$ converges to zero almost surely as~$N$ tends to infinity, for all $t \in \RR_+$.
\end{theorem}
\begin{proof} Let~$t\in\RR_+$,
    \begin{equation*}
    \begin{split}
        \mathrm{L}^{\GG,N}_{t} &= \EE\left[L^N_{t}\middle|\cG_t\right]= \EE\left[\sum_{n=1}^{N}  (1-\gamma)(\EAD_{t}^n-(1-k)e^{-r a}C^n_{t+a})_{+}\cdot \bOne_{\{ \cV^n_{t,\dd}<\cD^n_t\}}\middle|\cG_t\right].
    \end{split}
    \end{equation*}
    $(\EAD_{t}^n)_{n\in\OneN}$ is deterministic, we have:
    \begin{equation*}
        \begin{split}
        \mathrm{L}^{\GG,N}_{t} &= \sum_{n=1}^{N} \EAD_{t}^n\cdot \EE\left[ (1-\gamma)\left(1-(1-k)e^{-ra}\frac{C^n_{t+a}}{\EAD_{t}^n}\right)_{+}\cdot \bOne_{\{ \cV^n_{t,\dd}<\cD^n_t\}}\middle|\cG_t\right]\\
        &= \sum_{n=1}^{N} \EAD_{t}^n\cdot \EE\left[ (1-\gamma)\left(1-(1-k)e^{-r a}\frac{C^n_{t+a}}{\EAD_{t}^n}\right)_{+}\middle|\cV^n_{t,\dd}<\cD^n_t,\cG_t\right] \cdot \PP\left[ \cV^n_{t,\dd}<\cD^n_t\middle|\cG_t\right].\\
        \end{split}
    \end{equation*}
    The rest of the proof requires a version of the strong law of large numbers (Appendix of~\cite[Propositions~1, 2]{gordy2003risk}), where the systematic risk factor is~$\cG_t$.
\end{proof}
Explicitly, in the above theorem, we assume that our portfolio is perfectly fine grained, so that we can approximate $L_t^N$ -- the portfolio loss -- by $L_t^{\GG, N}$ -- the conditional expectation of loss given the systemic factor. By construction, the loss given default noted $\LGD$ is the percentage of the total exposure that the bank loses when a \textit{over-indebtedness} occurs. The literature on LGD modeling is fairly extensive. We can distinguish namely, economic modeling~\cite{bastos2010forecasting, roncalli2020handbook} and stochastic modeling~\cite{roncalli2020handbook}[Page 193], \cite{chalupka2008modelling}. As the definition of $\PD$ does not change compare to what we did in \cref{ct_sec_cont_time} and as $\EAD$ is given, we will focus on $\LGD$ modeling. We can first remark that $0\leq \LGD_{t,\dd}^n\leq 1-\gamma$, then the presence of a collateral necessarily reduces $\LGD$.

Other key quantities for the bank to understand the (dynamics of the) risk in the portfolio are the (expected and unexpected) losses and probability of default conditionally to the (information generated by the) risk factors. Precisely, for a date $t$ and a horizon~$T$, a bank computes some risk measures at~$t$ of its portfolio maturing at horizon~$T$. 
    
\begin{definition}[Projected losses]
Let $t  \ge 0$ be the time when the risk measure is computed for a period $T\geq 0$. As classically done, the potential loss is separated into three components:
    \begin{itemize}
        \item The (conditional) expected loss (EL) -- the amount that a bank expects to lose
on a credit exposure  -- reads
        \begin{equation}
        \EL^{N,T}_t := \EE\left[\mathrm{L}^{\GG,N}_{t+T}\middle|\cG_t\right]\label{ct-eq:el}.
        \end{equation}
        \item The unexpected loss (UL) -- the amount by which potential credit losses
might exceed the EL -- reads
        for $\alpha \in (0, 1)$,
        \begin{equation} \UL^{N,T}_{t}\!(\alpha):= \VaR^{\alpha, N,T}_t - \EL^{N,T}_{t}\label{ct-eq:ul},
        \quad\text{where} \quad
            1-\alpha = \PP\left[L_{t+T}^{\GG,N} \leq \VaR^{\alpha, N,T}_t\middle|\cG_t \right],
        \end{equation}
        and the scalar $\VaR^{\alpha, N,T}_t$ called value-at-risk.
        \item The stressed loss (or expected shortfall or ES) -- the amount by which potential credit
losses might exceed the capital requirement 
 -- reads:
        \begin{equation} \ES^{N,T}_{t}(\alpha):= \EE\left[L^N_{t+T} \middle| L^N_{t+T} \geq \VaR^{\alpha, N,T}_t, \cG_t\right],
        \qquad\text{for }\alpha \in (0,1). \label{ct-eq:sl}\end{equation}
    \end{itemize}
    \end{definition}

\noindent{Knowing PD, LGD, and EAD is enough to determine the losses. We assume EAD is given (\cref{item:EAD} in \cref{ct-ass:portfolio}). PD is given by \eqref{ct-eq: PD t T invest} in \cref{subsec:risk model}. We will now determine LGD when a collateral exists.} We focus on two types: a financial asset and a property in housing market. We introduce the two following lemmas which will help later on to set up explicit formulas for LGD.
Let $n\in\{1,\hdots,N\}$.
\begin{lemma}\label{ct-lemma: LGD t}
Assume that a stochastic process $K^n$ satisfies for all~$t\in\RR_+$,
\begin{enumerate}
    \item $\log{K^n}_t|\cG_t \sim \cN(m_t^n, (\sigma^n_t)^2)$ with $m^n_t\in\RR$ and $\sigma^n_t>0$,
    \item and $K^n_t|\cG_t$ and $\cV^n_{t,\dd}|\cG_t$ are independent.
\end{enumerate}
Therefore, for $t,u\in\RR_+$, 
    \begin{small}
    \begin{equation}\label{ct-eq: LGD t lemma}
        \EE\left[ \left(u - (1-k) \frac{K^n_{t}}{\EAD_{t}^n}\right)_{+}\middle|\cV^n_{t,\dd}<\cD^n_t,\cG_t\right] = u\Phi\left(\frac{w_t^n}{\sigma^n_t} \right) - \exp{\left(-w_t^n + \frac{1}{2}(\sigma^n_t)^2 \right)}  \Phi\left(\frac{w_t^n}{\sigma^n_t} - \sigma^n_t \right),
    \end{equation}
    \end{small}
    where
    \begin{equation}\label{ct-eq: num LGD lemma}
        w_t^n := \log{\left(u\frac{\EAD_t^n}{1-k} \right)} - m_t^n.
    \end{equation}
\end{lemma}
\begin{proof} Let $t,u\in\RR_+$,
    we have
    \begin{align*}
         &\EE\left[ \left(u - (1-k) \frac{K^n_{t}}{\EAD_{t}^n}\right)_{+}\middle|\cV^n_{t,\dd}<\cD^n_t,\cG_t\right]\\
        &\qquad= \EE\left[ \left(u - (1-k) \frac{K^n_{t}}{\EAD_{t}^n}\right)_{+}\middle|\cG_t\right]\quad\text{because}\quad K^n_{t}|\cG_t\text{  and  } \cV_{t,\dd}^n|\cG_t\text{ are independent}\\
        &\qquad= \EE\left[ \left(u - (1-k) \frac{K^n_{t}}{\EAD_t^n}\right) \bOne_{\{u - (1-k) \frac{K^n_{t}}{\EAD_t^n} \geq 0\}} \middle|\cG_t\right]\\
        &\qquad= u\PP\left(u - (1-k) \frac{K^n_{t}}{\EAD_t^n} \geq 0\middle|\cG_t \right) -\frac{(1-k)}{\EAD_t^n}  \EE\left[K^n_{t} \bOne_{\{u - (1-k) \frac{K^n_{t}}{\EAD_t^n} \geq 0\}} \middle|\cG_t\right].
    \end{align*}
    However,
    $\log{K^n_t}|\cG_t \sim \cN(m_t^n, (\sigma^n_t)^2)$. We also consider $w_t^n$ defined in~\eqref{ct-eq: num LGD lemma}, therefore
    \begin{equation*}
        \PP\left(u - (1-k) \frac{K^n_{t}}{\EAD_t^n} \geq 0\middle|\cG_t \right) = \Phi\left(\frac{w_t^n}{\sigma^n_t} \right).
    \end{equation*}
    We also have
    \begin{equation*}
        \EE\left[ K^n_{t} \bOne_{\{u - (1-k) \frac{K^n_{t}}{\EAD_t^n} \geq 0\}} \middle|\cG_t\right] = \exp{\left(-w_t^n + \frac{1}{2} (\sigma^n_t)^2 \right)}  \Phi\left(\frac{w_t^n}{\sigma^n_t} - \sigma^n_t \right).
    \end{equation*}
    The conclusion follows.
\end{proof}

\begin{lemma}\label{ct-lemma cond loss and pd}
Assume that a stochastic process $K^n$ satisfies, for each $t,T \in\RR_+$, 
\begin{enumerate}
    \item $\log{K^n}_{t+T}|\cG_t \sim \cN(m_{t,T}^n, (\sigma^n_{t,T})^2)$ with $m^n_{t,T}\in\RR$ and $\sigma^n_{t,T}>0$,
    \item and $\begin{bmatrix}
\log{\cV_{t+T}^n}\\
\log{K_{t+T+a}^n}
\end{bmatrix} |\cG_t \sim \cN\left(\begin{bmatrix}
\cK^n(\dd, t, T, \cA_t^\circ, \cZ_t)\\
\overline{\cK}^n_{t, T+a}
\end{bmatrix},
\begin{bmatrix}
\mathcal{L}^n(t,T)&cv_{t,T,a}^n\\
cv_{t,T,a}^n&\overline{\mathcal{L}}_{t,T+a}^n
\end{bmatrix} \right)$, where $\overline{\mathcal{L}}_{t,T+a}^n > 0$, $cv_{t,T,a}^n, \overline{\cK}^n_{t, T+a}\in\RR$, and $\cK^n(\dd, t, T, \cA_t^\circ, \cZ_t)$ as well as $\mathcal{L}^n(t,T)$ are defined in Corollary~\ref{ct-cor:law Vt}.
\end{enumerate}
Therefore, for $t,T,u \in\RR_+$, we have
\begin{small}
    \begin{equation}\label{ct-eq: LGD t T lemma}
        \begin{split}
            &\EE\left[ \left(1-(1-k)e^{-r a}\frac{\cC^n_{t+T+a,\dd}}{\EAD_{t+T}^n}\right)_{+}\cdot \bOne_{\{ \cV^n_{t+T,\dd}<\cD^n_{t+T}\}}\middle|\cG_t\right]\\
            &\qquad= u \Phi_2\left(\overline\omega_{t, T, a}^n,\Phi^{-1}(\PD^n_{t,T,\dd});\rho_{t, T, a}^n\right)- \exp{\left(\frac{1}{2}\overline{\mathcal{L}}_{t,T+a}^n-\sqrt{\overline{\mathcal{L}}_{t,T+a}^n}\overline\omega_{t, T, a}^n\right)} \times\\
            &\qquad\qquad\qquad\Phi_2\left(\overline\omega_{t, T, a}^n-\sqrt{\overline{\mathcal{L}}_{t,T+a}^n},\Phi^{-1}(\PD^n_{t,T,\dd})-\rho_{t, T, a}^n\sqrt{\overline{\mathcal{L}}_{t,T+a}^n};\rho_{t, T, a}^n\right),
        \end{split}
    \end{equation}
    \end{small}
    where $\PD^n_{t,T,\dd}$ is defined in~\eqref{ct-eq: PD t T invest} and where
    \begin{equation*}
       \rho_{t, T, a}^n := \frac{cv_{t,T,a}^n}{\sqrt{\mathcal{L}^n(t,T)\overline{\mathcal{L}}^n_{t,T+a}}},\quad
    \text{and}\quad
        \overline\omega_{t, T, a}^n :=\frac{\log{\left(u\frac{\EAD_{t+T}^n}{(1-k)e^{-r a}}\right)}- \overline{\cK}^n_{\dd, t, T+a}}{\sqrt{\overline{\mathcal{L}}^n_{t,T+a}}}.
    \end{equation*}
    \end{lemma}
\begin{proof}
    Let $t,T,u\in\RR_+$, we have
\begin{equation*}
    \begin{split}
& \EE\left[ \left(u-(1-k)e^{-r a}\frac{K_{t+T+a}^n}{\EAD_{t+T}^n}\right)_{+}\cdot \bOne_{\{ \cV^n_{t+T,\dd}<\cD^n_{t+T}\}}\middle|\cG_t\right]\\
&\qquad=  \EE\left[ \left(u-(1-k)e^{-r a}\frac{K_{t+T+a}^n}{\EAD_{t+T}^n}\right)\cdot \bOne_{u-(1-k)e^{-r a}\frac{K_{t+T+a}^n}{\EAD_{t+T}^n} \geq 0} \bOne_{\{ \cV^n_{t+T,\dd}<\cD^n_{t+T}\}}\middle|\cG_t\right]\\
&\qquad= \EE\left[ \left(u-(1-k)e^{-r a}\frac{K_{t+T+a}^n}{\EAD_{t+T}^n}\right)\cdot \bOne_{\{ \cV^n_{t+T,\dd}<\cD^n_{t+T},\quad K_{t+T+a}^n\leq u\frac{\EAD_{t+T}^n}{(1-k)e^{-r a}} \}}\middle|\cG_t\right]\\
&\qquad= u\PP\left[\cV^n_{t+T,\dd}<\cD^n_{t+T}, K_{t+T+a}^n\leq u\frac{\EAD_{t+T}^n}{(1-k)e^{-r a}} \}\middle|\cG_t\right]\\
&\qquad\qquad- \frac{(1-k)e^{-r a}}{\EAD_{t+T}^n}\EE\left[K_{t+T+a}^n \cdot \bOne_{\{ \cV^n_{t+T,\dd}<\cD^n_{t+T}, K_{t+T+a}^n\leq u\frac{\EAD_{t+T}^n}{(1-k)e^{-r a}} \}}\middle|\cG_t\right].
\end{split}
\end{equation*}
However $\begin{bmatrix}
\log{\cV_{t+T}^n}\\
\log{K_{t+T+a}^n}
\end{bmatrix} |\cG_t \sim \cN\left(\begin{bmatrix}
\cK^n(\dd, t, T, \cA_t^\circ, \cZ_t)\\
\overline{\cK}^n_{t, T+a}
\end{bmatrix},
\begin{bmatrix}
\mathcal{L}^n(t,T)&cv_{t,T,a}^n\\
cv_{t,T,a}^n&\overline{\mathcal{L}}_{t,T+a}^n
\end{bmatrix} \right)$, therefore we have
\begin{small}
\begin{align*}
        &\PP\left[\cV^n_{t+T,\dd}<\cD^n_{t+T}, K^n_{t+T+a}\leq u\frac{\EAD_{t+T}^n}{(1-k)e^{-r a}} \}\middle|\cG_t\right]\\
        &\qquad= \Phi_2\left(\frac{\log{u\frac{\EAD_{t+T}^n}{(1-k)e^{-r a}}} - \overline{\cK}^n_{t, T+a}}{\sqrt{\overline{\mathcal{L}}^n_{t,T+a}}},\frac{\log{\cD^n_{t+T}} - \cK^n(\dd, t, T, \cA_t^\circ, \cZ_t)}{\sqrt{\mathcal{L}^n(t,T)}};\frac{cv_{t,T,a}}{\sqrt{\mathcal{L}^n(t,T)\overline{\mathcal{L}}^n_{t,T+a}}} \right),
\end{align*}
\end{small}
and
\begin{equation*}
    \begin{split}
    &\frac{(1-k)e^{-r a}}{\EAD_{t+T}^n}\EE\left[K^n_{t+T+a} \cdot \bOne_{\{ \cV^n_{t+T,\dd}<\cD^n_{t+T}, K^n_{t+T+a}\leq u\frac{\EAD_{t+T}^n}{(1-k)e^{-r a}} \}}\middle|\cG_t\right]\\
    &\qquad= \EE\left[e^{\log{K^n_{t+T+a}}} \cdot \bOne_{\left\{\frac{\log{K^n_{t+T+a}}- \overline{\cK}^n_{t, T+a}}{\sqrt{\overline{\mathcal{L}}^n_{t,T+a}}}  \leq \overline\omega_{t, T, a}^n, \frac{\log{\cV^n_{t+T,\dd}} - \cK^n(\dd, t, T, \cA_t^\circ, \cZ_t)}{\sqrt{\mathcal{L}^n(t,T)}}<\Phi^{-1}(\PD^n_{t,T,\dd})\right\}}\middle|\cG_t\right]
    \end{split} 
\end{equation*}
However, according to~\eqref{ct-INT:int2} in \cref{ct-app:Bivariate gaussian}, $\EE[e^{\sigma X} \bOne_{X\leq x, Y\leq y}] =  e^{\frac{1}{2}\sigma^2}\Phi_2\left(x-\sigma,y-\rho\sigma; \rho\right)$, therefore, 
\begin{scriptsize}
\begin{align*}
    &\frac{(1-k)e^{-r a}}{\EAD_{t+T}^n}\EE\left[e^{\log{K^n_{t+T+a}}} \cdot \bOne_{\{ \cV^n_{t+T,\dd}<\cD^n_{t+T}, K^n_{t+T+a}\leq u\frac{\EAD_{t+T}^n}{(1-k)e^{-r a}} \}}\middle|\cG_t\right]\\
    &\qquad= \exp{\left(\frac{1}{2}\overline{\mathcal{L}}^n_{t,T+a} -\sqrt{\overline{\mathcal{L}}^n_{t,T+a}}\overline\omega_{t, T, a}^n\right)}\Phi_2\left(\overline\omega_{t, T, a}^n-\sqrt{\overline{\mathcal{L}}^n_{t,T+a}}, \Phi^{-1}(\PD^n_{t,T,\dd})-\frac{cv_{t,T,a}}{\sqrt{\mathcal{L}^n(t,T)}} ; \frac{cv_{t,T,a}}{\sqrt{\mathcal{L}^n(t,T)\overline{\mathcal{L}}^n_{t,T+a}}}\right);
\end{align*}
\end{scriptsize}
Moreover, from \cref{ct-pr cond pd},
\begin{equation*}
    \PP\left[  \cV^n_{t+T,\dd}<\cD^n_{t+T}\middle|\cG_t\right] = \Phi\left(\frac{\log(\cD_{t+T}^n) -\cK^n(\dd, t, T, \cA_t^\circ, \cZ_t)}{\sqrt{\mathcal{L}^n(t,T)}}\right).
\end{equation*}
This concludes the proof.
\end{proof}

\subsection{When there is not collateral}\label{ct-sec:no col}
When firm~$n$ does not have a collateral, therefore $C^n = 0$ and from~\eqref{ct-eq:def LGD}, LGD is
\begin{equation}\label{ct-eq:no col}
    \LGD_{t,\dd}^n = 1-\gamma.
\end{equation}

\subsection{When collateral is a financial asset}\label{ct-sec:col fin asset}

Here we assume that the collateral of the firm~$n$ is an investment in a financial asset. Precisely, we assume that that investment is a proportion $\alpha^n\in(0, 1]$ of a given firm located in the economy described in \cref{ct_sec_cont_time}. Consequently, it is subjected to the same constraints in terms of productivity and of carbon transition scenarios as firm~$n$. As any investment, it should generate a stream of cash flows so that at each time, we can compute its value by using the discounted cash flows model introduced in~\eqref{ct-eq:DCF}.

Let note the collateral cash flows $(\overline F_t^{n})_{t\in\RR_+}$, its dynamics is similar to the firm cash flows introduced in \cref{ct-ass:link}. We have for all $t\in\RR_+$,
\begin{equation}\label{ct-eq:col cash flow dy}
    \dr \overline F^n_{t,\dd} = \mathfrak{\overline a}^{n\cdot} ((\mu + \varsigma \cZ_t)\dr t + \dr v(\dd_t)) + \overline\sigma_{n} \dr \overline{\cW}^n_t,
\end{equation}
where~$\mathfrak{\overline a}^{n\cdot}\in\RR^{I}$ and where $(\overline\cW_t)_{t\in\RR_+}$ is a $\RR^{N}$-Brownian motion with $\overline\sigma_{n} > 0$. 
Moreover, $B^{\cZ}$ (noise of productivity), $\overline\cW^n$ (noise of collateral), and $(\cW^n)_{n\in\OneN}$ (noise of debtors) are independent. We also note $\overline{\tilde{\af}}^{n\cdot} = \overline{\af}^{n\cdot}(\Ir_I-\llambda)$.

\begin{remark}
We have assumed that $\overline\cW^n$ and $\cW^n$ are not correlated, but this is not always the case. For example, if the depreciation of the firm value heading to its \textit{over-indebtedness} implies the depreciation of the collateral value, then we should have a positive correlation.
\end{remark}
Inspired by~\eqref{ct-eq:DCF}, the collateral value at time~$t$ is
\begin{equation*}
    C^n_{t,\dd} := \alpha^n \EE_{t}\left[\int_{s=t}^{+\infty} e^{-r s} \overline F^n_{s,\dd} \dr s \right],
\end{equation*}
and by~\eqref{ct-eq de cV}, a proxy collateral value as 
\begin{equation}\label{ct-eq:approx col fa value}
    \cC_{t,\dd}^n := \alpha^n \overline F^n_{t,\dd} \int_{t}^{+\infty} e^{-r(s-t)}\EE_{t}\left[ \exp \left( (s-t)\overline\af^{n\cdot} \mu
    +\overline\af^{n\cdot}\left(v(\dd_{s})-v(\dd_t)\right)
    + \sigma_{n} (\overline\cW^n_s - \overline\cW^n_t)
    \right)\dr s\right].
\end{equation}
Therefore, the following proposition (whose the proof is inspired by Lemma~\ref{ct-lem:approx firm value} and corollary~\ref{ct-cor:law Vt}) gives a proxy of the collateral value. 

\begin{proposition}\label{ct-prop: collateral value}
For any $n\in\OneN$ and for all $t\in\RR_+$
    \begin{enumerate}
        \item Assume that $\overline\varrho_n := \frac{1}{2}\overline\sigma_{\mathfrak{b}_n}^2 + \mathfrak{\overline a}^{n\cdot}\mu -r<0$.
    Given the \textit{carbon emissions costs} sequence~$\dd$, the proxy of collateral value defined in~\eqref{ct-eq:approx col fa value}, is well defined  and
    \begin{align}\label{ct-eq:fa expression}
        \cC^n_{t,\dd} = \alpha^n \overline{F}^n_0 \mathfrak{\overline R}^n_t(\dd) \exp{\left(\mathfrak{\overline{a}}^{n\cdot}(\cA_t^\circ-v(\dd_{0}))\right)}\exp\left(\overline\sigma_{n} \overline\cW_t^n\right),
    \end{align}
where 
    \begin{equation}
     \mathfrak{\overline R}^n_t(\dd) := \int_{0}^{\infty} e^{\overline{\varrho}_n s}\exp{\left(\mathfrak{\overline{a}}^{n\cdot} v(\dd_{t+s})  \right)} \dr s.
    \end{equation}
    \item Moreover, we note $\mathfrak{\overline{m}}^n(\dd,t,\cA_{t}^\circ) :=  \log{(\alpha^n\overline{F}^n_0)} +  \log{\mathfrak{\overline{R}}^n_t(\dd)} +  \left(\mathfrak{\overline{a}}^{n\cdot}(\cA_t^\circ-v(\dd_{0}))\right)$ and we have 
    \begin{equation}
        \log{\cC^n_{t,\dd}}|\cG_t \sim \cN\left(\mathfrak{\overline{m}}^n(\dd,t,\cA_{t}^\circ), t\overline{\sigma}_{n}^2 \right),
    \end{equation}
    and we note  $\overline{\cK}^n(\dd, t, T, \cA_t^\circ, \cZ_t) := \log{(\alpha^n\overline{F}^n_0 \overline{\mathfrak{R}}^n_{t+T}(\dd))} + \overline{\af}^{n\cdot}(\mu T + \varsigma\Upsilon_{T}\cZ_t + \cA_{t}^\circ-v(\dd_{0}))$ and $\overline{\mathcal{L}}^n(t,T) := \overline{\af}^{n\cdot}\Sigma^\cA_{t, t+T}\overline{\af}^{n\cdot} + (t+T) \overline{\sigma}_{n}^2$, and we have
    \begin{equation}\label{ct-eq:law col fin}
        \log{\cC_{t+T}^n}|\cG_t \sim \cN\left(\overline{\cK}^n(\dd, t, T, \cA_t^\circ, \cZ_t)  ,\overline{\mathcal{L}}^n(t,T)\right).
    \end{equation}
    \item Assume that 
        \begin{align} 
    \overline\rho_n := \frac{1}{2}\overline\sigma_{n}^2+\overline\af^{n\cdot}\mu+\frac{1}{2}  \varsigma^2\frac{c_\Gamma^2}{\lambda_\Gamma^2}\lVert\overline\af^{n\cdot}\rVert^2 \lVert\Sigma\rVert^2 < r,
\end{align}
therefore $C_{t,\dd}^n$ is well defined and there exists a constant~$\overline C$ such that $\EE \left[\left|\frac{C^n_{t,\dd}}{\overline F^n_{t,\dd}} - \frac{\cC^n_{t,\dd}}{\overline F^n_{t,\dd}}\right|\right] \le \overline C \varsigma$, for all~$\varsigma > 0$.
    \end{enumerate}
\end{proposition}

\begin{proof}
Let $n\in\OneN$ and $t\in\RR_+$, \eqref{ct-eq:fa expression} directly comes from \cref{ct-lem:approx firm value}. The proofs of the three points are equivalent to~\ref{ct-proof:lem:approx firm value}. Let us develop the conditional laws. From~\eqref{ct-eq:fa expression}, we have
    \begin{align*}
        \log{\cC^n_t} = \log{\alpha^n \overline F^n_0 \overline{\mathfrak{R}}^n_t(\dd)} + \overline\af^{n\cdot}(\cA_t^\circ-v(\dd_{0})) +\overline{\sigma}_n\overline\cW^n_t.
    \end{align*}
    Because $\overline\cW^n$ is a Brownian motion, $\overline\cW^n_t \sim \cN\left(0, t\right)$ and, $\overline\cW^n$ and $B^{\cZ}$ are independent, we obtain $\log{\cC_t^n}|\cG_t \sim \cN\left(\overline{\mathfrak{m}}^n(\dd,t,\cA_{t}^\circ), t\overline\sigma_{n}^2 \right)$. Let also $T\in\RR_+$, we have 
    \begin{equation*}
        \begin{split}
            \log{\cC^n_{t+T}} &= \log{\overline F^n_0 \overline{\mathfrak{R}}^n_{t+T}(\dd)} + \overline\af^{n\cdot}(\cA_{t+T}^\circ-v(\dd_{0})) +\overline\cW^n_{t+T}.
        \end{split}
    \end{equation*}
    From \cref{ct-rem:VAR1}, $\cA_{t+T}|\cG_t \sim \cN\left(M^{\cA, T}_{t} , \Sigma^{\cA, T}_{t}\right)$ and because $\overline\cW^n$ is a Brownian motion, $\overline\cW^n_{t+T} \sim \cN\left(0, t+T\right)$. Moreover, $\overline\cW^n$ and $B^{\cZ}$ are independent. We have
    \begin{equation*}
        \log{\cC^n_{t+T}}|\cG_t \sim \cN\left(\log{\alpha^n\overline F^n_0 \overline{\mathfrak{R}}^n_{t+T}(\dd)} + \overline\af^{n\cdot}(M^{\cA, T}_{t}-v(\dd_{0})) , \overline\af^{n\cdot}\Sigma^{\cA, T}_{t}\overline\af^{n\cdot} + (t+T) \overline\sigma_{n}^2\right).
    \end{equation*}
    The conclusion follows.
\end{proof}

From the (proxy of the) collateral value~$\cC$, we can then derive a precised expression of $\LGD$ based on \cref{ct-theo:gordy2003}. We have:
    
\begin{theorem}\label{ct-theorem: LGD t invest}
When $a = 0$ (no liquidation delay), the Loss Given Default of the obligor $n$ \textit{over-indebted} at time~$t\in\RR_+$, conditional on $\cG_t$ is
    \begin{equation}\label{ct-eq: LGD t invest}
        \LGD_{t,\dd}^n = (1-\gamma) \left[\Phi\left(\frac{w_t^n}{\overline{\sigma}_{\mathfrak{b}_n} \sqrt{t}} \right) - \exp{\left(-w_t^n + \frac{1}{2} t \overline{\sigma}_{n}^2 \right)}  \Phi\left(\frac{w_t^n}{\overline{\sigma}_{n}\sqrt{t}} - \overline{\sigma}_{n}\sqrt{t} \right) \right],
    \end{equation}
    where
    \begin{equation}\label{ct-eq: num LGD}
        w_t^n := \log{\left(\frac{\EAD_t^n}{1-k} \right)} - \mathfrak{\overline{m}}^n(\dd,t,\cA_{t}).
    \end{equation}
\end{theorem}
\begin{proof}
    Let $t\in\RR_+$. By remarking that 
    \begin{enumerate}
        \item $\log{\cC^n_{t,\dd}}|\cG_t \sim \cN\left(\mathfrak{\overline{m}}^n(\dd,t,\cA_{t}), t\overline{\sigma}_{n}^2 \right)$, 
        \item and $\cC_{t,\dd}^n|\cG_t$ and $\cV_{t,\dd}^n|\cG_t$ are independent,
    \end{enumerate}
    we can simply apply Lemma~\ref{ct-lemma: LGD t} with $u = 1$.
\end{proof}
We can also remark that the situation where there is no collateral corresponds to $\overline{F}^n_0 = 0$. We then have  
\begin{equation*}
    \begin{split}
        \overline{F}^n_0 \to 0 &\implies \log{(\overline{F}^n_0)} \to -\infty \implies \mathfrak{\overline{m}}^n(\dd,t,\cA_{t}^\circ) \to-\infty\implies w_t^n\to+\infty \implies\LGD_{t,\dd}^n \to 1-\gamma.
    \end{split}
\end{equation*}

For each $t,T \in\RR_+$, we introduce now the (conditional) $\LGD$ of the  entity~$n$ at time~$t$ on the horizon~$T$, namely
\begin{equation*}
    \LGD^n_{t,T,\dd } := (1-\gamma) \EE\left[ \left(1-(1-k)e^{-r a}\frac{\cC^n_{t+T+a,\dd}}{\EAD_{t+T}^n}\right)_{+}\middle|\cV^n_{t+T,\dd}<\cD^n_{t+T},\cG_t\right].
\end{equation*}
It is precisely about calculating at date $t$ the proportion of the exposure that the bank would lose if the counterpart~$n$ is \textit{over-indebted} at date $t+T$.

\begin{proposition}[Projected PD and LGD]\label{ct-pr cond loss and pd}
For each $t,T \in\RR_+$, the (conditional) $\LGD$ of the  entity~$n$ at time~$t$ on the horizon~$T$, reads
    \begin{small}
    \begin{equation}\label{ct-eq: LGD t T invest}
        \begin{split}
            \LGD^n_{t,T,\dd } &= \frac{1-\gamma}{\PD^n_{t,T,\dd}} \left[\Phi_2\left(\overline\omega_{t, T, a}^n,\Phi^{-1}(\PD^n_{t,T,\dd});\rho_{t, T, a}^n\right)- \exp{\left(\frac{1}{2}\overline{\mathcal{L}}^n(t,T+a)-\sqrt{\overline{\mathcal{L}}^n(t,T+a)}\overline\omega_{t, T, a}^n\right)}\times\right.\\
            &\qquad\qquad\qquad\left.\Phi_2\left(\overline\omega_{t, T, a}^n-\sqrt{\overline{\mathcal{L}}^n(t,T+a)},\Phi^{-1}(\PD^n_{t,T,\dd})-\rho_{t, T, a}^n\sqrt{\overline{\mathcal{L}}^n(t,T+a)};\rho_{t, T, a}^n\right) \right],
        \end{split}
    \end{equation}
    \end{small}
    where
    \begin{equation*}
       \rho_{t, T, a}^n := \frac{\varsigma^2 \af^{n\cdot}\Gamma^{-1} \left(\int_{0}^{T} \left(e^{-\Gamma u} - \Ir_I \right) \Sigma\Sigma^\top \left(e^{-\Gamma (u+a)} - \Ir_I \right) \dr u \right) (\overline{\af}^{n\cdot}\Gamma^{-1})^\top}{\sqrt{\mathcal{L}^n(t,T)\overline{\mathcal{L}}^n(t,T+a)}},
    \end{equation*}
    and
    \begin{equation*}
        \overline\omega_{t, T, a}^n :=\frac{\log{\frac{\EAD_{t+T}^n}{(1-k)e^{-r a}}}- \overline{\cK}^n(\dd, t, T+a, \cA_t^\circ, \cZ_t)}{\sqrt{\overline{\mathcal{L}}^n(t,T+a)}},
    \end{equation*}
    and where $\PD^n_{t,T,\dd}$ defined in \cref{ct-pr cond pd}.
    \end{proposition}
\begin{proof}
    Let $t,T\in\RR_+$, from \eqref{ct-eq:def loss} and \eqref{ct-eq:el},
    \begin{equation*}
        \EL^{N,T}_t := \EE\left[\mathrm{L}^{\GG,N}_{t+T}\middle|\cG_t\right] = \EE\left[\sum_{n=1}^{N} L_{n,t+T}^\GG\middle|\cG_t\right] = \sum_{n=1}^{N} \EE\left[L_{n,t+T}^\GG\middle|\cG_t\right].
        \end{equation*}
But for $1\leq n\leq N$, we have
\begin{equation*}
    \begin{split}
\EE\left[L_{n,t+T}^\GG\middle|\cG_t\right] &= \EE\left[\EAD_{t+T}^n \LGD^n_{t+T,\dd} \PD^n_{t+T,\dd}\middle|\cG_t\right]\\
&= \EAD_{t+T}^n \EE\left[\LGD^n_{t+T,\dd} \PD_{t+T,\dd}^n \middle|\cG_t\right]\quad\text{as } \EAD_{t+T}^n\text{ is deterministic}\\
&= \EAD_{t+T}^n \EE\left[\EE\left[  (1-\gamma)\left(1-(1-k)e^{-r a}\frac{\cC^n_{t+T+a,\dd}}{\EAD_{t+T}^n}\right)_{+}\cdot \bOne_{\{ \cV^n_{t+T,\dd}<\cD^n_{t+T}\}}\middle|\cG_{t+T}\right]\middle|\cG_t\right]\\
&= (1-\gamma)\EAD_{t+T}^n \EE\left[ \left(1-(1-k)e^{-r a}\frac{\cC^n_{t+T+a,\dd}}{\EAD_{t+T}^n}\right)_{+}\cdot \bOne_{\{ \cV^n_{t+T,\dd}<\cD^n_{t+T}\}}\middle|\cG_t\right].
\end{split}
\end{equation*}
However, from~\eqref{ct-eq expression of cV}, we have $ \log{\cV^n_{t+T,\dd}} = \log{(F^n_0 \mathfrak{R}^n_{t+T}(\dd))} + \af^{n\cdot}(\cA_{t+T}^\circ-v(\dd_{0})) +\sigma_n\cW^n_{t+T}$,
and from~\eqref{ct-eq:fa expression}, we have $\log{\cC^n_{t+T+a,\dd}} = \log{\alpha^n \overline{F}^n_0 \mathfrak{\overline R}^n_{t+T+a}(\dd)} +\mathfrak{\overline{a}}^{n\cdot}(\cA_{t+T+a}^\circ-v(\dd_{0}))+\overline\sigma_{n} \overline\cW_{t+T+a}^n$.
Therefore, 
$F_{t+T}^n|\cG_t \sim \cLN\left(\cK^n(\dd, t, T, \cA_t^\circ, \cZ_t), \mathcal{L}^n(t,T)\right)$, $\cC_{t+T+a}^n|\cG_t \sim \cLN\left(\overline{\cK}^n(\dd, t, T+a, \cA_t^\circ, \cZ_t) ,\overline{\mathcal{L}}^n(t,T)\right)$, and
\begin{equation*}
    cov(\log{F_{t+T}^n},\log{\cC_{t+T+a}^n}|\cG_t) = \EE[\log{F_{t+T}^n}\log{\cC_{t+T+a}^n}|\cG_t] - \EE[\log{F_{t+T}^n}|\cG_t]\EE[\log{\cC_{t+T+a}^n}|\cG_t].
\end{equation*}
However, 
\begin{small}
\begin{align*}
        &\EE[\log{F_{t+T}^n}\log{\cC_{t+T+a}^n}|\cG_t]\\
        & \qquad= \EE\left[ \left(\log{\left(F^n_0 \mathfrak{R}^n_{t+T}(\dd) e^{-\af^{n\cdot}v(\dd_{0})}\right)} + \af^{n\cdot}\cA_{t+T}^\circ +\sigma_n\cW^n_{t+T}\right) \right.\\
        &\qquad\qquad\left. \left(\log{\left(\alpha^n \overline{F}^n_0 \mathfrak{\overline R}^n_{t+T+a}(\dd)e^{-\overline{\af}^{n\cdot}v(\dd_{0})}\right)} +\overline{\af}^{n\cdot}\cA_{t+T+a}^\circ+\overline\sigma_{n} \overline\cW_{t+T+a}^n\right) \middle|\cG_t\right]\\
        &\qquad= \EE\left[ \log{\left(F^n_0 \mathfrak{R}^n_{t+T}(\dd) e^{-\af^{n\cdot}v(\dd_{0})}\right)} \log{\left(\alpha^n \overline{F}^n_0 \mathfrak{\overline R}^n_{t+T}(\dd)e^{-\overline{\af}^{n\cdot}v(\dd_{0})}\right)}+\af^{n\cdot}\cA_{t+T}^\circ\overline\sigma_{n} \overline\cW_{t+T+a}^n
        \right.\\
        &\qquad\qquad\left. + \log{\left(F^n_0 \mathfrak{R}^n_{t+T}(\dd) e^{-\af^{n\cdot}v(\dd_{0})}\right)} (\overline{\af}^{n\cdot}\cA_{t+T+a}^\circ +\overline\sigma_{n} \overline\cW_{t+T+a}^n)+\overline\sigma_{n} \overline\cW_{t+T+a}^n \sigma_n\cW^n_{t+T}\right.\\
        &\qquad\qquad\left. +\log{\left(\alpha^n \overline{F}^n_0 \mathfrak{\overline R}^n_{t+T+a}(\dd)e^{-\overline{\af}^{n\cdot}v(\dd_{0})}\right)} \af^{n\cdot}\cA_{t+T}^\circ +\af^{n\cdot}\cA_{t+T}^\circ\overline{\af}^{n\cdot}\cA_{t+T+a}^\circ\right.\\
        &\qquad\qquad\left.+ \log{\left(\alpha^n \overline{F}^n_0 \mathfrak{\overline R}^n_{t+T+a}(\dd)e^{-\overline{\af}^{n\cdot}v(\dd_{0})}\right)} \sigma_n\cW^n_{t+T}+\overline{\af}^{n\cdot}\cA_{t+T+a}^\circ \sigma_n\cW^n_{t+T}\middle|\cG_t
        \right]\\
        &\qquad= \log{\left(F^n_0 \mathfrak{R}^n_{t+T}(\dd) e^{-\af^{n\cdot}v(\dd_{0})}\right)} \log{\alpha^n \overline{F}^n_0 \mathfrak{\overline R}^n_{t+T+a}(\dd)e^{-\overline{\af}^{n\cdot}v(\dd_{0})}}\\
        &\qquad\qquad+ \log{\left(F^n_0 \mathfrak{R}^n_{t+T}(\dd) e^{-\af^{n\cdot}v(\dd_{0})}\right)} \overline{\af}^{n\cdot}\EE[\cA_{t+T+a}^\circ|\cG_t]\\
        &\qquad\qquad+\log{\left(\alpha^n \overline{F}^n_0 \mathfrak{\overline R}^n_{t+T+a}(\dd)e^{-\overline{\af}^{n\cdot}v(\dd_{0})}\right)} \af^{n\cdot}\EE[\cA_{t+T}^\circ|\cG_t] + \EE[\af^{n\cdot}\cA_{t+T}^\circ\overline{\af}^{n\cdot}\cA_{t+T+a}^\circ|\cG_t].
\end{align*}
\end{small}
By also developing $\EE[\log{F_{t+T}^n}|\cG_t]\EE[\log{\cC_{t+T}^n}|\cG_t]$, we obtain 
\begin{equation*}
    \begin{split}
        & cov(\log{F_{t+T}^n},\log{\cC_{t+T}^n}|\cG_t) = cov(\af^{n\cdot}\cA_{t+T}^\circ,\overline{\af}^{n\cdot}\cA_{t+T}^\circ|\cG_t)\\
        &\qquad= \varsigma^2 \af^{n\cdot}\Gamma^{-1} \left(\int_{0}^{T} \left(e^{-\Gamma u} - \Ir_I \right) \Sigma\Sigma^\top \left(e^{-\Gamma (u+a)} - \Ir_I \right) \dr u \right) (\overline{\af}^{n\cdot}\Gamma^{-1})^\top := cv_{t,T,a}.
    \end{split}
\end{equation*}
We obtain
\begin{small}
    \begin{equation}
\begin{bmatrix}
\log{\cV_{t+T}^n}\\
\log{\cC_{t+T}^n}
\end{bmatrix} |\cG_t \sim \cN\left(\begin{bmatrix}
\cK^n(\dd, t, T, \cA_t^\circ, \cZ_t)\\
\overline{\cK}^n(\dd, t, T+a, \cA_t^\circ, \cZ_t)
\end{bmatrix},
\begin{bmatrix}
\mathcal{L}^n(t,T)&cv_{t,T,a}\\
cv_{t,T,a}&\overline{\mathcal{L}}^n(t,T+a)
\end{bmatrix} \right).
\end{equation}
\end{small}
Then, we use the Lemma~\ref{ct-lemma cond loss and pd} with $u=1$ to conclude the proof.
\end{proof}
We remark that the carbon price introduced in our economy affect both PD through the obligor cash flows and LGD through the collateral cash flows. See more remarks in \cref{sec:remark_on_LGD}.

\subsection{When collateral is commercial or residential property}\label{ct-sec:col proporty}

In this section, we assume that loans are backed
by either residential or commercial building. The problem here is then to model the real estate market in the presence of the climate transition risk. The latter is represented by energy efficiency as well as the carbon price. We would like to compute the value of a dwelling at time time~$t$. We use exactly as in~\cite{sopgoui2024realestate} the actualized sum of the cash flows before the renovation date (taking into account the additional energy costs due to inefficiency of the building), at the renovation date, and after the renovation date (when the building becomes efficient). Moreover, the agent chooses rationally the date of renovation which maximizes the value of his property. Therefore, according to \cite{sopgoui2024realestate}[Theorem 2.4], we have the following proposition.
\begin{theorem}\label{ct-prop:housing}
    Assume that the following conditions are satisfied:
    \begin{enumerate}
        \item the carbon price function $\delta: t\mapsto \delta_t$ is non decreasing on $\RR_+$ and deterministic;
        \item the energy price $\ff(\cdot, \pf)$ is non decreasing on $\RR_+$ for each {type of
energy~$\pf$}.
    \end{enumerate}
    Then, the market value of the building serving as the collateral to firm~$n$ at $t\geq 0$, given the carbon price sequence~$\delta$, is given by
    \begin{align}\label{ct-eq:coldynamics1}
        \cC^n_{t,\delta} = C^n_{t} - R_n X_{t,\delta}^n,
    \end{align}
    where
    \begin{align}\label{ct-eq:coldynamics4}
        X_{t,\delta}^n := \cc(\alpha^n, \alpha^\star) e^{-\bar r(\ft_n-t)} + (\alpha^n-\alpha^\star)\int_{t}^{\ft_n} \ff(\delta_u,\pf) e^{-\bar r(u-t)} \dr u,
    \end{align}
    and where the optimal date of renovations~$\ft_n\in[t,+\infty]$ is given by 
\begin{numcases}{\ft_n=}
$t$ & if $\ff(\delta_{\theta},\pf) - \bar r \frac{\cc(\alpha^n, \alpha^\star)}{\alpha^n-\alpha^\star} > 0$ for all $\theta\in[t,\infty)$\label{ct-eq:optimal t_n t} \\
+\infty & if $\ff(\delta_{\theta},\pf) - \bar r \frac{\cc(\alpha^n, \alpha^\star)}{\alpha^n-\alpha^\star} < 0$ for all $\theta\in[t,\infty)$ \label{ct-eq:optimal t_n infty}\\
\theta^\star & the unique solution of $\ff(\delta_{\theta},\pf) =\bar r \frac{\cc(\alpha^n, \alpha^\star)}{\alpha^n-\alpha^\star}$ on $\theta\in[t,\infty)$.\label{ct-eq:optimal t_n}
\end{numcases}
Moreover, 
\begin{align}\label{ct-eq:housing EOU}
        C^n_{t} &:= R_n C^n_{0} e^{K_t} ,
    \end{align}
    where
\begin{subequations}
\begin{align}
\displaystyle\dr K_t &= \left(\dot{\chi}_t+ \nu(\chi_t - K_t) \right) \dr t + \overline{\sigma} \dr\overline{B}_t, \label{ct-eq:coldynamics2}\\
\dr \overline{B}_t &= \rho^\top \dr B^{\cZ}_t + \sqrt{1-\lVert\rho\rVert^2} \dr\overline\cW_t,
\label{ct-eq:coldynamics3}
\end{align}
\end{subequations}
with $(\overline\cW_t)_{t\in\RR_+}$ is a standard Brownian motion independent to~$B^{\cZ}$ introduced in \cref{ct-sassump:OU} and driving the productivity of the economy. Moreover, $C^n_{0}$, $r$, $R_n, \overline{\sigma} > 0$, $\rho\in\RR_+^I$, and $\chi\in\mathfrak{D}(\RR_+, \RR_+)$. 
    \noindent We introduce the following filtration {$\mathbb{U}:=(\cU_t)_{t\in\RR_+}$} with respect to $t\geq 0$, $\cU_t := \sigma\left(\left\{ \overline{W}_s, B^{\cZ}_s: s\leq t\right\}\right)$.  
\end{theorem}
The following corollary gives the conditional distribution of the collateral. Its proof is straightforward and is detailed in \cite{sopgoui2024realestate}.
\begin{corollary}\label{ct-cor: cond law coll 2}
    For $0\leq t \leq T$, the law of
    $C^n_{t+T} = R_n C_0^n\exp{(K_{t+T})}$ conditional on $\cG_t$ is log-Normal~$\cLN(m_{t, T}^{n}, v_{t, T}^{n})$ with
    \begin{equation}\label{ct-eq:mean col}
        m_{t, T}^{n} := \log{(R_n C_0^n)} + \chi_{t+T} - \left(\chi_0-K_0 \right)e^{-\nu (t+T)} + \overline{\sigma}\rho^\top \int_{0}^{t} e^{-\nu(t+T-s)} \dr B^{\cZ}_s,
    \end{equation}
    and
    \begin{equation}\label{ct-eq:var col}
        v_{t, T}^{n} :=  \frac{(\overline{\sigma}\lVert\rho\rVert)^2)}{2\nu}\left(1-e^{-2\nu T}\right) + \frac{(\overline{\sigma})^2 (1-\lVert\rho\rVert^2)}{2\nu}\left(1-e^{-2\nu (t+T)}\right).
    \end{equation}
\end{corollary}

\paragraph{An example of the energy price function}
We can assume that the price of each type of energy~$\pf$ is a linear function of the carbon price, therefore
\begin{equation}\label{ct-eq:f}
    \ff: (\delta_t,\pf) \mapsto \ff^\pf_1 \delta_t + \ff^\pf_0\qquad t\geq 0,
\end{equation}
with $\ff^\pf_1, \ff^\pf_0>0$ and $\delta$ is the carbon price defined in the \cref{ct-sassc:price} or an example given in~\eqref{ct-eq:carbon price}.

\paragraph{An example of the renovation costs function}
We can consider that the costs of renovation of a dwelling~$\cc$, to move its energy efficiency from $x$ to $y$, is
\begin{equation}\label{ct-eq:c}
    \cc: (x,y) \mapsto c_0 |x-y|^{1+c_1},
\end{equation}
with $c_0>0$ and $c_1\geq -1$. This choice of $\cc$ allows us to model that when a building has a bad energy efficiency, its renovation is costly. {This is similar to that of \cite{sopgoui2024realestate} and is inspired on the data given in \cite{caumont2016renovation}, which summarizes the results of the survey on energy renovation
of single-family homes in France.}
\paragraph{An example of the optimal renovation time}
 With the example of the carbon price in~\eqref{ct-eq:carbon price}, the example of the energy price in~\eqref{ct-eq:f}, and the example of the renovation costs in~\eqref{ct-eq:c}, the optimal renovation time, solution of~\eqref{ct-eq:optimal t_n} is given by
\begin{align}\label{ct-eq:ex optimal t_n}
    \ft_n = t_\circ + \frac{1}{\eta_\delta} \log{\left(\frac{c_0 r |\alpha^n-\alpha^\star|^{c_1} -  \ff^\pf_0}{\ff^\pf_1 P_{carbon}}\right)}.
\end{align}
We can clearly remark that the optimal renovation date depends on the climate transition policy ($P_{carbon}$ and $\eta_\delta$), on the energy prices ($\ff^\pf_0$ and $\ff^\pf_1$), on the renovation costs ($c_0$ and $c_1$), and on the energy efficiencies ($\alpha^n$ and $\alpha^\star$).

By using the housing price under the climate transition as given in \cref{ct-prop:housing}, we can then derive a precised expression of $\LGD$ when the collateral exists and is a building. We have:
\begin{theorem}\label{ct-theorem: LGD t real estate}
Let~$1\leq n\leq N$. When $a = 0$ (no liquidation delay), the Loss Given Default of the obligor $n$ is \textit{over-indebted} at time~$t\in\RR_+$, conditional on $\cG_t$, is
    \begin{equation}\label{ct-eq: LGD t real estate}
        \LGD_{t,\delta}^n = (1-\gamma) \left[\left(1+ (1-k)\frac{R_n X_{t,\delta}^n}{\EAD_t^n}\right)\Phi\left(\frac{w_t^n}{\sqrt{v_{t, 0}^{n}}} \right) - \exp{\left(-w_t^n + \frac{1}{2} v_{t, t}^{n} \right)}  \Phi\left(\frac{w_t^n}{\sqrt{v_{t, 0}^{n}}} - \sqrt{v_{t, 0}^{n}} \right) \right],
    \end{equation}
    where
    \begin{equation}\label{ct-eq: num LGD re}
        w_t^n :=\log{\left(\frac{\EAD_t^n}{(1-k) R_n}+X_{t,\delta}^n\right)} - m_{t, 0}^{n},
    \end{equation}
    and with $ m_{t, 0}^{n}$ and $v_{t, t}^{n}$ defined in Corollary~\ref{ct-cor: cond law coll 2}, and $X_{t,\delta}^n$ defined in~\eqref{ct-eq:coldynamics4}.
\end{theorem}
\noindent We can also verify that when there is not collateral corresponding to $C^n_0 = 0$. We then have  
\begin{equation*}
    \begin{split}
        C^n_0 \to 0 &\implies \log{(R_n C^n_0)} \to -\infty\implies m_{t, 0}^{n} \to-\infty \implies w_t^n\to+\infty \implies\LGD_{t,\delta}^n \to 1-\gamma.
    \end{split}
\end{equation*}
It is even worse when the costs associated with the transition explode, LGD also explodes as 
\begin{equation*} 
X_{t,\delta}^n \to +\infty \implies w_t^n\to+\infty\implies\LGD_{t,\delta}^n \to +\infty.
\end{equation*}

\begin{proof}
    Let~$t\geq 0$ and $1\leq n\leq N$. By remarking that 
    \begin{enumerate}
        \item $\log{C^n_{t,\dd}}|\cG_t \sim \cN\left(m_{t, 0}^{n}, v_{t, 0}^{n} \right)$, and
        \item $C_{t,\dd}^n|\cG_t$ and $\cV_{t,\dd}^n|\cG_t$ are independent,
        \item and from~\eqref{ct-eq:def LGD} when $a=0$, we have
        \begin{align*}
            \LGD_{t,\delta}^n 
            &= (1-\gamma) \EE\left[ \left(1 - (1-k) \frac{\cC^n_{t,\delta}}{\EAD_{t}^n}\right)_{+}\middle|\cV^n_{t,\dd}<\cD^n_t,\cG_t\right]\\
        &= (1-\gamma) \EE\left[ \left(1 + (1-k) \frac{R_n X^n_{t,\delta}}{\EAD_{t}^n} - (1-k) \frac{C^n_{t,\delta}}{\EAD_{t}^n}\right)_{+}\middle|\cV^n_{t,\dd}<\cD^n_t,\cG_t\right],
        \end{align*}
    \end{enumerate}
    we can simply apply Lemma~\ref{ct-lemma: LGD t} with $u = 1 + (1-k) \frac{R_n X^n_{t,\delta}}{\EAD_{t}^n}$.
\end{proof}

Once again, we want to compute the (conditional) Loss Given Default of the entity $n$ at time $t$ on the horizon $T$. We can formalize that in the following proposition:
\begin{proposition}[Projected LGD]\label{ct-pr cond loss re}
For each $t,T \geq 0$ and $1\leq n \leq N$, the (conditional) Loss Given Default of the  entity~$n$ at time~$t$ on the horizon~$T$, reads
\begin{footnotesize}
\begin{equation}\label{ct-eq:stressedLGD re}
\begin{split}
    \LGD^n_{t,T,\delta} &= \frac{1-\gamma}{\PD^n_{t,T,\dd}} \left[\left(1+ (1-k)e^{-r a}\frac{R_n X_{t+T+a,\delta}^n}{\EAD_{t+T}^n}\right)\Phi_2\left(\overline\omega_{t, T, a}^n,\Phi^{-1}(\PD^n_{t,T,\dd});\rho_{t, T, a}^n\right)\right.\\
        &\qquad\qquad\quad\left.- \exp{\left(\frac{1}{2}v_{t, T+a}^{n} -\sqrt{v_{t, T+a}^{n}}\overline\omega_{t, T, a}^n\right)}\Phi_2\left(\overline\omega_{t, T, a}^n-\sqrt{v_{t, T+a}^{n}},\Phi^{-1}(\PD^n_{t,T,\dd})-\rho_{t, T, a}^n\sqrt{v_{t, T+a}^{n}};\rho_{t, T, a}^n\right) \right],
\end{split}
    \end{equation}
\end{footnotesize}
where
\begin{equation*}
    \rho_{t, T, a}^n := \overline{\sigma}\varsigma \frac{ \af^{n\cdot}\Gamma^{-1} \left(\int_{0}^{T} e^{-\nu(u+a)}\left(\Ir_I - e^{-\Gamma u} \right)\Sigma\dr u \right)\rho}{\sqrt{\mathcal{L}^n(t,T) \times v_{t, T+a}^{n}}},
\end{equation*}
and
\begin{equation*}
    \overline w_{t, T, a}^n := \frac{\log{\left(\frac{\EAD_{t+T}^n}{(1-k)R_ne^{-r a}}+X_{t+T+a,\delta}^n \right)} - m_{t, T+a}^{n}}{\sqrt{v_{t, T+a}^{n}}},
    \end{equation*}
and with $m_{t, t+T+a}^{n}$ and $v_{t, t+T+a}^{n}$ defined in Corollary~\ref{ct-cor: cond law coll 2} and $\PD^n_{t,T,\dd}$ defined in \cref{ct-pr cond pd}.
\end{proposition}
    
\begin{proof}
    Let~$t, T\geq 0$ and $1\leq n\leq N$.
    From the beginning of the proof of \cref{ct-pr cond loss and pd}, we have
\begin{align*}
\EE\left[L_{n,t+T}^\GG\middle|\cG_t\right]
&= (1-\gamma)\EAD_{t+T}^n \EE\left[ \left(1-(1-k)e^{-r a}\frac{\cC^n_{t+T+a,\dd}}{\EAD_{t+T}^n}\right)_{+}\cdot \bOne_{\{ \cV^n_{t+T,\dd}<\cD^n_{t+T}\}}\middle|\cG_t\right],
\end{align*}
then
\begin{footnotesize}
\begin{align*}
\EE\left[L_{n,t+T}^\GG\middle|\cG_t\right]
&= (1-\gamma)\EAD_{t+T}^n \EE\left[ \left(1+(1-k)e^{-r a} \frac{R_n X_{t+T+a,\delta}^n}{\EAD_{t+T}^n}- (1-k)e^{-r a} \frac{\cC_{t+T+a,\delta}^n}{\EAD_{t+T}^n}\right)_{+}\cdot \bOne_{\{ \cV^n_{t+T,\dd}<\cD^n_{t+T}\}}\middle|\cG_t\right].
\end{align*}
\end{footnotesize}
Remark that 
\begin{equation*}
\begin{split}
    \log{C_0^n e^{K_{t+T+a}}} &= \log{(C_0^n)} + \chi_{t+T+a} - (\chi_0 - K_0)e^{-\nu (t+T+a)}\\
    &\qquad+ \overline{\sigma}\rho^\top \int_{0}^{t+T+a} e^{-\nu(t+T+a-s)} \dr B_s^{\cZ}  + \overline{\sigma}\sqrt{1-\lVert\rho\rVert^2}  \int_{0}^{t+T+a} e^{-\nu(t+T+a-s)} \dr \overline \cW_s,
\end{split}    
\end{equation*}
and recall that $\log{(\cV^n_{t+T,\dd})} = \log{(F^n_0 \mathfrak{R}^n_{t+T}(\dd))} + \af^{n\cdot}(\cA_{t+T}^\circ-v(\dd_{0})) +\sigma_n\cW^n_{t+T}$. Therefore,
\begin{equation*}
\begin{split}
    cov\left(\log{C_0^n e^{K_{t+T+a}}}, \log{F^n_{t+T}}\middle|\cG_t\right) &= cov\left(\af^{n\cdot}\cA_{t+T}^\circ, \overline{\sigma}\rho^\top \int_{0}^{t+T+a} e^{-\nu(t+T+a-s)} \dr B_s^{\cZ} \middle|\cG_t\right)\\
    &= \overline{\sigma} \af^{n\cdot} cov\left(\cA_{t+T}^\circ, \int_{0}^{t+T+a} e^{-\nu(t+T+a-s)} \dr B_s^{\cZ} \middle|\cG_t\right) \rho\\
    &= \overline{\sigma}\varsigma \af^{n\cdot}\Gamma^{-1} \left(\int_{0}^{T} e^{-\nu(u+a)}\left(\Ir_I - e^{-\Gamma u} \right)\Sigma\dr u \right)\rho := cv_{t,T,a}.
    \end{split}
\end{equation*}
Consequently, we can write
\begin{small}
    \begin{equation}
\begin{bmatrix}
\log{F_{t+T}^n}\\
\log{C_0^n e^{K_{t+T+a}}}
\end{bmatrix} |\cG_t \sim \cN\left(\begin{bmatrix}
\cK^n(\dd, t, T, \cA_t^\circ, \cZ_t)\\
m_{t, T+a}^{n}
\end{bmatrix},
\begin{bmatrix}
\mathcal{L}^n(t,T)&cv_{t,T,a}\\
cv_{t,T,a}& v_{t, T+a}^{n}
\end{bmatrix} \right).
\end{equation}
\end{small}
Then, we use the Lemma~\ref{ct-lemma cond loss and pd} with $u=1+ (1-k)e^{-r a}\frac{R_n X_{t+T+a,\delta}^n}{\EAD_{t+T}^n}$ to conclude the proof.
\end{proof}
We can remark that~$\LGD^n_{t,\delta}$ as well as $\LGD^n_{t,T,\delta}$ are also functions of the optimal renovation time~$\ft_n$. Furthermore, if both the financial asset and the housing price are affected by the climate transition through their dependence on the carbon price sequence~$\delta$, the financial asset depends also on the carbon price intensities (of firms production/consumption and of households consumption)~$(\tau, \zeta, \kappa)$ which are not specific to a given company but to the economy as a whole. The housing price is clearly affected by specific climate factors, namely the energy efficiency~$\alpha^n$ and the renovation date~$\ft_n$.

\subsection{Expected and Unexpected losses}\label{ct-sec:EL_UL}

Let us recall that we have a portfolio with N loans. We assume that loans from 1 to $N_1$ are unsecured, loans from $N_1+1$ to $N_2$ are secured by a financial asset as collateral, and loans from $N_2+1$ to $N$ are secured by a commercial or residential property as collateral.

We write the expression of the portfolio EL and UL as functions of the parameters and of the processes introduced above, and introduce the entity's probability of default.\\

We can therefore give expressions of EL and UL. Let~$t, T \geq 0$, the (conditional) Expected Loss of the  portfolio at time~$t$ on the horizon~$T$ defined in~\eqref{ct-eq:el}, reads
  \begin{equation}
        \begin{split}
            \EL^{N,T}_{t} &= \EE\left[\mathrm{L}^{\GG,N}_{t+T}\middle|\cG_t\right] = \sum_{n=1}^{N} \EAD_{t+T}^n \cdot \LGD^n_{t,T,\dd } \cdot \PD^n_{t,T,\dd}\\
            &= \sum_{n=1}^{N_1} \EAD_{t+T}^n \cdot \LGD^n_{t,T,\dd } \cdot \PD^n_{t,T,\dd}+ \sum_{n=N_1+1}^{N_2} \EAD_{t+T}^n \cdot \LGD^n_{t,T,\dd } \cdot \PD^n_{t,T,\dd}\\
            &\qquad\qquad+ \sum_{n=N_2+1}^{N} \EAD_{t+T}^n \cdot \LGD^n_{t,T,\delta } \cdot \PD^n_{t,T,\dd}.\label{ct-eq:stressedEL}
        \end{split}
    \end{equation}
We can then compute each term conditionally to~$\cG_t$.
\begin{enumerate}
    \item Given that $\LGD^n_{t,T,\dd } = 1-\gamma$ for $1\leq n\leq N_1$, to compute {$\sum_{n=1}^{N_1} \EAD_{t+T}^n \cdot \LGD^n_{t,T,\dd } \cdot \PD^n_{t,T,\dd}$}, all you have to do is calculate $\PD^n_{t,T,\dd}$.
    \item Given that $N_1+1\leq n\leq N_2$, the collaterals are financial assets, therefore, to compute $\sum_{n=N_1+1}^{N_2} \EAD_{t+T}^n \cdot \LGD^n_{t,T,\dd} \cdot \PD^n_{t,T,\dd}$, we compute first $\PD^n_{t,T,\dd}$. Then we compute~$\LGD^n_{t,T,\dd}$ directly through~\eqref{ct-eq: LGD t T invest}.
     \item Given that {$N_2+1\leq n\leq N$}, the collaterals are properties, therefore, to compute $\sum_{n=N_2+1}^{N}\EAD_{t+T}^n \cdot \LGD^n_{t,T,\delta} \cdot \PD^n_{t,T,\dd}$, we compute first $\PD^n_{t,T,\dd}$. Then we compute~$\LGD^n_{t,T,\delta}$ directly through~\eqref{ct-eq:stressedLGD re}.
\end{enumerate}

For $\alpha\in(0,1)$, the (conditional) Unexpected Loss of the  portfolio at time~$t$ on the horizon~$T$,  cannot be obtained in closed-form as EL. Precisely, there is not a closed-form expression neither of $\UL^{\alpha, N,T}_t$ nor of $\VaR^{\alpha, N,T}_t$. But we can describe how to compute $\VaR^{\alpha, N,T}_t$ given that $\PP\left[L_{t+T}^{\GG,N} \leq \VaR^{\alpha, N,T}_t\middle|\cG_t \right]$ as introduced in~\eqref{ct-eq:ul}. First, let us note that from \cref{ct-theo:gordy2003}, we have
\begin{equation*}
    \begin{split}
        L_{t+T}^{\GG,N}&= \sum_{n=1}^{N} \EAD_{t+T}^n \cdot \LGD_{t+T,\dd}^n\cdot \PD_{t+T,\dd}^n \\
        &= \sum_{n=1}^{N_1} \EAD_{t+T}^n \cdot \LGD_{t+T,\dd}^n\cdot \PD_{t+T,\dd}^n + \sum_{n=N_1+1}^{N_2}\EAD_{t+T}^n \cdot \LGD_{t+T,\dd}^n\cdot \PD_{t+T,\dd}^n\\
            &\qquad\qquad+ \sum_{n=N_2+1}^{N} \EAD_{t+T}^n \cdot \LGD_{t+T,\delta}^n\cdot \PD_{t+T,\dd}^n.
    \end{split}
\end{equation*}
We can then describe each term's law conditionally to~$\cG_t$. 
\begin{enumerate}
    \item $\LGD_{t+T,\dd}^n = 1-\gamma$ and from~\eqref{ct-eq:def PD}, we have $\PD_{t+T,\dd}^n$ which depends on $\cA_{t+T}$. Then to simulate law of $\sum_{n=1}^{N_1} \EAD_{t+T}^n \cdot \LGD_{t+T,\dd}^n\cdot \PD_{t+T,\dd}^n$ conditional on $\cG_t$, just simulate $\cA_{t+T}|\cG_t$.
    \item From~\eqref{ct-eq: LGD t invest}, we have $\LGD_{t+T,\dd}^n$ which depends on $\cA_{t+T}$ through $w_{t+T}^n$ defined in~\eqref{ct-eq: num LGD}. We said in the previous item that $\PD_{t+T,\dd}^n$ depends on $\cA_{t+T}$. Therefore, to simulate law of $\sum_{n=N_1+1}^{N_2} \EAD_{t+T}^n \cdot \LGD_{t+T,\dd}^n\cdot \PD_{t+T,\dd}^n$ conditional on $\cG_t$, just simulate $\cA_{t+T}|\cG_t$.
    \item From~\eqref{ct-eq: LGD t real estate}, we have $\LGD_{t+T,\dd}^n$ which depends on $\int_{0}^{t+T} e^{-\nu(t+T-s)} \dr B^{\cZ}_s$ through $w_{t+T}^n$ defined in~\eqref{ct-eq: num LGD re}. We said in the previous item that $\PD_{t+T,\dd}^n$ depends on $\cA_{t+T}$. Therefore, to simulate law of $\sum_{n=N_2+1}^{N} \EAD_{t+T}^n \cdot \LGD_{t+T,\delta}^n\cdot \PD_{t+T,\dd}^n$ conditional on $\cG_t$, just simulate $\cA_{t+T}|\cG_t$ and $\int_{0}^{t+T} e^{-\nu(t+T-s)} \dr B^{\cZ}_s|\cG_t$ (which are in fact the same because both $\cA_{t+T}$ and $\int_{0}^{t+T} e^{-\nu(t+T-s)} \dr B^{\cZ}_s$ are $\cG_{t+T}$-measurable).
\end{enumerate}

\subsection{Remarks on the determinants of LGD}\label{sec:remark_on_LGD}
The results ~\eqref{ct-eq: LGD t invest} and~\eqref{ct-eq:stressedLGD re} tell us that, in the case the collateral is an investment, Loss Given Default depends on:
\begin{enumerate}
    \item the carbon price~$\delta$ for both~\eqref{ct-eq: LGD t invest} and~\eqref{ct-eq:stressedLGD re},
    \item parameters specific to the company (the contract),
    \begin{itemize}
        \item the time~$t$ when it is computed,
        \item the date of default~$t+T$,
        \item the Exposure at Default~$EAD$,
    \end{itemize}
    \item parameters specific to the collateral,
    \begin{itemize}
        \item its liquidation time $t+T+a$,
        \item the liquidation costs~$k$,
        \item the correlation of its cash flows with the environment~$\overline{\af}$,
        \item the standard deviation of its cash flows~$\overline{\sigma}_{\bb}$,
        \item the fraction of recovery from other means~$\gamma$,
    \end{itemize}
    \item the nature of the collateral:
    \begin{itemize}
        \item if it is a financial asset, then parameters related to the carbon intensities~$\tau, \zeta, \kappa$,
        \item if it is a building, then parameters related to the energy efficiency~$\alpha$, type of energy $\pf$, and renovation costs~$\cc$,
    \end{itemize}
    \item parameters specific to the economy  to which the colateral belongs to:
    \begin{itemize}
        \item the (cumulative) productivity $\cA$ (and its parameters) of the economy,
        \item the interest rates~$r$ and $\bar r$.
    \end{itemize}
\end{enumerate}
Some of these typical risk drivers are reported by~\cite{chalupka2008modelling}.\\

We could also look at the sensitivities of the LGD to each of these variables and parameters. However, the expressions of LGD we obtained are not very tractable so that it would be difficult to get detailed expressions of theses sensitivities. If necessary, they can be calculated using numerical methods {(see \cite{caprioli2023quantifying, fermanian2005sensitivity})}.

\section{Numerical experiments, estimation and calibration}\label{sec:numerical exp}
In this section, we describe how the parameters of multisectoral model, of the firm valuation model, and of the credit risk model are estimated given the historical macroeconomic variables (consumption, labour, output, GHG emissions, housing prices, etc.) as well as the historical credit portfolio data (firms rated and defaulted, collateral, etc.) In a second step, we give the expression of the risk measures (PD, LGD, EL, and UL) introduced in the previous sections, that we compute using Monte Carlo simulations.

\subsection{Calibration and estimation}\label{ct-sec3}
We will calibrate the model parameters on a set of data ranging from time $\ft_0$ to $\ft_1$. In practice, $\ft_0=1978$ and $\ft_1=t_\circ = 2021$. From now on, we will discretize the observation interval into $M\in\NN^*$ steps $t_m = \ft_0 + \frac{\ft_1 -\ft_0}{M} m$ for $0\leq m\leq M$. We note $\Ttt^M := \{t_0, t_1,\hdots,t_{M}\}$. We will not be interested in convergence results here. {For the sake of clarity, we will omit the dependence of each estimated parameter/function on~$\mathfrak{t}_0$ and~$\mathfrak{t}_1$. We then note $\widehat x$ the estimation of $x$.}

\subsubsection{Estimation of carbon intensities}\label{ct-sec:calibintensities}
For each sector~$i\in\Jj$ and for $0\leq m\leq M$, 
we observe the output~$Y_{t_m}^i$, the labor~$N_{t_m}^i$, the intermediary input~$(Z_{t_m}^{ji})_{j\in\Jj}$, and the consumption~$C_{t_m}^i$ (recall that the transition starts at year~$t_\circ$). For the sake of clarity, we will omit the dependence of each estimated parameter on~$M$.

To calibrate each carbon intensity~$\mathfrak{y} \in\{\tau^1, \hdots, \tau^I, \zeta^{11}, \zeta^{12},\hdots, \zeta^{I I-1}, \zeta^{II}, \kappa^1, \hdots, \kappa^I\}$, we follow exactly the same process already presented in~\cite{bouveret2023propagation}. The main difference is that after calibration, we can compute~$\mathfrak{y}$ for each~$t\in\RR_+$.
Afterwards, if we consider the example of the carbon price introduces in~\eqref{ct-eq:carbon price}, we can compute the \textit{emissions cost rate}~$\widehat\dd_t$.

\subsubsection{Estimation of economic parameters} \label{ct-subsec:calibva}
As in~\cite{gali2015monetary}, we assume a unitary Frisch elasticity of labor supply so~$\varphi = 1$ and the utility
of consumption is logarithmic so~$\sigma = 1$, while we calibrate~$({\llambda}^{ij})_{i,j\in\Jj}$ and {~$({\psi}^i)_{i\in\Jj}$} in the same way as in~\cite{bouveret2023propagation}.
We can then compute the functions~{$\Psi$} and~$\Lambda$ defined in \cref{ct-cor:output_consc}, followed by the function~$\widehat v^i$ as defined in~\eqref{ct-eq:pricefuncc}.
We can also compute the output growth
$\left(\Delta^{\Yy}_{t_m} = (\log(\Yy_{t_{m}}^i) -\log(\Yy_{t_{m-1}}^i))_{j\in\Jj}\right)_{1\leq m\leq M}$ directly from data.\\

Without carbon tax in any sector,
it follows from~\eqref{ct-eq:consumptionc} in Corollary~\ref{ct-cor:output_consc} that, 
for each $1\leq m\leq M$, the computed consumption growth $\Delta^{\Yy}_{t_m}$ is equal to $\Delta^{\Yy}_{t_m} = \frac{\ft_1-\ft_0}{M}(\Ir_I-\widehat{\llambda})^{-1}\widehat{\Theta}_{t_m} $ when $\Ir_I-\widehat{\llambda}$ is not singular;
hence $\widehat{\Theta}_{t_m} = \frac{M}{\ft_1-\ft_0} (\Ir_I-\widehat{\llambda}) \Delta^{\Yy}_{t_m}$. We can then compute the estimations~$\widehat{\mu}$, $\widehat{\Gamma}$, $\widehat{\Sigma}$, and $\widehat{\varsigma}$, of parameters $\mu$, $\Gamma$, $\Sigma$, and~$\varsigma$ (all defined in \cref{ct-sassump:OU}), as detailed in \cite{sopgoui2024realestate}[Section 3.1.1.].

\subsubsection{Estimation of firm and of the credit model parameters} \label{ct-subsec:calibfirmp}

Recall that we have a portfolio with $N\in\NN^{*}$ firms (or credit) at time~$t_\circ$. For each firm~$n\in\{1,\hdots,N\}$, we have its historical cash flows~$(F_{t_m}^n)_{1\leq m\leq M}$, 
hence its log-cash flow growths. For any~$t\in\Ttt^M$ and $1\leq i\leq I$, we denote by $r_{t}^i$ (resp.~$d_{t}^i$) the number of firms in $g_i$ rated at the beginning of the year~$t$ (resp. defaulted during the year~$t$). In particular, $r_{\mathfrak{t}_0} = \#g_i$. 
Within each group~$g_i$, all the firms behave in the same way as there is only one risk class. Since each sub-portfolio constitutes a single risk class, we have for each $n\in g_m$, $\af^n = \af^{n_i}$, $\sigma_{\bb^n} = \sigma_{\bb^{n_i}}$, and $B^n = B^{n_i}$. We then proceed as follows:

\begin{enumerate}
    \item Knowing the output growth~$\left(\Delta^{\Yy}_{t}\right)_{t\in\Ttt^M}$, we calibrate the factor loading~$\af_{n_i}$ and the standard deviation~$\sigma_{n_i}$, according to \cref{ct-ass:link}, appealing to the regression. For all $1\leq m\leq M$,
    \begin{equation}
        \sum_{n\in g_i} \log{F_{t_{m}}^n} - \log{F_{t_{m-1}}^n} = (\#g_i)\af^{n_i}\Delta^{\Yy}_{t_{k}} + \sqrt{\frac{\ft_1-\ft_0}{M} \#g_i} \sigma_{\bb^{n_i}} \mathfrak{u}_{t_{m}}\quad\text{where}\quad \mathfrak{u}_{t_{m}}\sim\cN(0,1).
    \end{equation}
    \item {We then obtain the barrier $B^{n_i}$ by maximum likelihood estimation (MLE)} as detailed in~\cite[Section~3]{gordy2002corr}. We have
\begin{equation*}
\widehat{B}^{n_i} := \argmax_{B^{n_i} \in \RR^+} \mathcal{L}(B^{n_i}),
\end{equation*}
where $\mathcal{L}(B^{n_i})$ is the log-likelihood function defined by
\begin{equation*}
    \mathcal{L}(B^{n_i}) := \sum_{m=1}^{M} \log\left(\int_{\mathbb{R}^{2I}} \mathbb{P}[D^{n_i} = d_{t_{m}}^{i}|(a,\theta)] d\PP[(\cA_{t_{m}}^\circ,\cZ_{t_{m}}) \leq (a,z)]\right),
\end{equation*}
and where
\begin{equation*}
    \mathbb{P}[D^{n_i} = d_{t_{m}}^{i} | (\cA_{t_{m}}^\circ,\cZ_{t_{m}})] = {\binom{r_{t_m}^{i}}{d_{t_m}^{i}}} (\PD^{n_i}_{t_{m},1,0})^{d_{t_{m}}^{i}} \Big(1 - \PD^{n_i}_{t_{m},1,0}\Big)^{r_{t_{m}}^{i}-d_{t_{m}}^{i}},
\end{equation*}
with $D^{n_i}$ the Binomial random variable standing for the conditional number of defaults, and~$\PD^{n_i}_{t_{m},1,0}$ in \cref{ct-eq: PD t T invest},
depending on $\sigma_{\bb^{n_i}} = \widehat{\sigma}_{\bb^{n_i}}$, $\af^{n_i} = \widehat{\af}^{n_i}$, for $1\leq m\leq M$, $\delta_{t_{m}} = 0$ and on~$B^{n_i}$.

\end{enumerate}

 \subsubsection{Calibration of collateral}\label{ct-sec:calibcoll}

Recall that we have a portfolio with $N\in\NN^{*}$ firms (or credit) at time~$t_\circ$. For each firm~$n\in\{1,\hdots,N\}$, if the collateral is
 
 \paragraph{A financial asset}\label{ct-sec:calibcoll_fin}
 
 We have its historical cash flows~$(\overline F_{t_m}^n)_{0\leq m\leq M}$, 
hence its log-cash flow growths. Recall that, even if two firms belong to the same sub-portfolio, there is no reason that their collaterals behave in the same way. We also know the output growth~$\left(\Delta^{\Yy}_{t_m}\right)_{1\leq m\leq M}$. We then have,

\begin{enumerate}
    \item the proportion~$\alpha^n$ of the investment representing the collateral is known.
    \item We calibrate the factor loading~$\widehat{\overline{\af}}_{n}$ and the standard deviation~$\widehat{\overline{\sigma}}_{n}$, according to~\eqref{ct-eq:col cash flow dy}, appealing to the regression. For all $1\leq m\leq M$
    \begin{equation}
         \log{\overline F_{t_{m}}^n} - \log{\overline F_{t_{m-1}}^n} = {\overline\af}^{n}\Delta^{\Yy}_{t_{m}} + \sqrt{\frac{\ft_1-\ft_0}{M}} \overline\sigma_{n} \mathfrak{u}_{t_{m}}\quad\text{where}\quad \mathfrak{u}_{t_{m}}\sim\cN(0,1).
    \end{equation}
\end{enumerate}

\paragraph{A commercial or residential property}\label{ct-sec:calibcoll_house}

We assume that in the past, carbon price did not have impact on the dwelling price so that for all~$t\in\Ttt^M$, $X_{t,\delta}^n$ defined in~\eqref{ct-eq:coldynamics4} is zero. Moreover,  $C^n_{0}$ and $R_n$ defined in~\eqref{ct-eq:housing EOU}, the value of the collateral at $0$ and the surface, are known. All that remains is to calibrate the parameters of the process $K$ defined in~\eqref{ct-eq:coldynamics2} and~\eqref{ct-eq:coldynamics3}. Let us consider a real estate index~$(REI_{t_m})_{0\leq m\leq M}$. We assume that the long-term average of the real estate index~$\chi$, introduced in~\eqref{ct-eq:coldynamics2} is linear and for $t\in\RR_+$, $\chi_t = \varrho t + \vartheta$. We then estimate $\varrho$, $\vartheta$, $\nu$, $\overline{\sigma}$, and $\rho$ as~\cite{sopgoui2024realestate}[Section 3.1.2]. 

\subsection{Simulations}
In this section as well, the idea here is not to (re)demonstrate or improve convergence results. {Recall the climate transition should start at $t_\circ$ and end at $t_\star$. Therefore, we would like to compute the risk measures for $t\in[t_\circ,t_\star]$.}
\subsubsection{Of the productivities~$\cZ$ and $\cA$}\label{ct-sec:simul Theta and A}
Let {$K~\in\NN^*$}, for $0\leq k\leq K$, we note $u_k = t_\circ + \frac{t_\star -t_\circ}{K} k$ for $0\leq k\leq K$. We would like to simulate $\cZ_{u_k}$ and $\cA_{u_k}$. 
For $\cZ$, we adopt the Euler-Maruyana (see~\cite{maruyama1955continuous, kanagawa1988rate}) scheme as in~\cite{sopgoui2024realestate}[Section 3.2.1]. 

\subsubsection{Of the probability of \textit{over-indebtedness} PD and of LGD}

For $n\in\OneN$ and $t_\circ\leq t\leq t_\star$, We would like to compute~$\PD^n_{t,T,\dd}$ as defined in~\eqref{ct-eq: PD t T invest} as well as $\LGD_{t,T,\dd}^n$ defined in~\eqref{ct-eq: LGD t T invest} and $\LGD_{t,T,\delta}^n$ in~\eqref{ct-eq:stressedLGD re}. After simulating $\cZ_t$ and $\cA_t$ as described in~\ref{ct-sec:simul Theta and A}, we get $\widehat{\cZ}_t$ and $\widehat{\cA}_t$. Then, for each~$1\leq i\leq I$ and for each~$n\in g_i$, we have

\begin{enumerate}
    \item from~\eqref{ct-eq: PD t T invest}, the estimated probability of default of firm~$n$ is 
    \begin{equation}\label{ct-eq: approx PD t T}
        \widehat{\PD}^{n}_{t,T,\dd} = \Phi\left( \frac{\log(\cD_{t+T}^n) - \widehat\cK^n(\widehat\dd, t, T, \widehat\cA_t, \widehat\cZ_t)}{\widehat{\mathcal{L}}^n(t,T)}\right),
    \end{equation}
    with
    \begin{equation}\label{ct-eq:approx K}
        \widehat\cK^n(\widehat\dd, t, T, \widehat\cA_t, \widehat\cZ_t) = \log{(F^n_0 \widehat{\mathfrak{R}}^n_{t+T}(\widehat\dd))} + \widehat\af^{n\cdot}(\widehat\mu T + \widehat\varsigma\Upsilon_{T}\widehat\cZ_t +  \widehat\cA_t-\widehat v(\widehat\dd_{0})),
    \end{equation}
     and
    \begin{equation}\label{ct-eq:approx L}
        \widehat{\mathcal{L}}^n(t,T) := \widehat\varsigma^2\widehat\af^{n\cdot} \widehat\Gamma^{-1} \left(\frac{T}{L} \sum_{l=0}^{L} \left(e^{-\widehat\Gamma u_l} - \Ir_I \right) \widehat\Sigma\widehat\Sigma^\top \left(e^{-\widehat\Gamma u_l} - \Ir_I \right) \right) (\widehat\af^{n\cdot} \widehat\Gamma^{-1})^\top + (t+T) \widehat\sigma_{n}^2,
    \end{equation}
    where 
    $F^n_0$ and $\cD_{t+T}^n$ are know, 
    $\widehat\dd$ defined in \cref{ct-sec:calibintensities}, $\widehat{\mathfrak{R}}^n_{t+T}(\widehat\dd))$ in~\eqref{ct-eq:FV_second_term} in \cref{ct-lem:approx firm value},
    $\widehat\Gamma,\widehat\varsigma, \widehat v$ in \cref{ct-subsec:calibva},
    $\widehat\af^{n\cdot}, \widehat\sigma^n$ in \cref{ct-subsec:calibfirmp}, 
    $\widehat\Upsilon_{T} := \widehat\Gamma^{-1}(\Ir_I-e^{-\widehat\Gamma T})$ and with $u_l:= \frac{T l}{L},l=0,\hdots,L$.

    \item If the collateral of loan~$n$ is a financial asset, from~\eqref{ct-eq: LGD t T invest},
    \begin{small}
    \begin{equation}\label{ct-eq: approx LGD t T invest}
        \begin{split}
            \widehat\LGD^n_{t,T,\dd} &= \frac{1-\gamma}{\widehat{\PD}^n_{t,T,\dd}} \left[\Phi_2\left(\widehat{\overline\omega}_{t, T, a}^n,\Phi^{-1}\widehat{\PD}^n_{t,T,\dd});\widehat{\rho}_{t, T, a}^n\right)- \exp{\left(\frac{1}{2}\widehat{\overline{\mathcal{L}}}^n(t,T+a) \right)}\times\right.\\
            &\qquad\qquad\qquad\left.\Phi_2\left(\widehat{\overline\omega}_{t, T, a}^n-\sqrt{\widehat{\overline{\mathcal{L}}}^n(t,T+a)},\Phi^{-1}(\widehat{\PD}^n_{t,T,\dd})-\widehat{\rho}_{t, T, a}^n\sqrt{\widehat{\overline{\mathcal{L}}}^n(t,T+a)};\widehat{\rho}_{t, T, a}^n\right) \right],
        \end{split}
    \end{equation}
    \end{small}
    where $\widehat{\overline{\cK}}$ and $\widehat{\overline{\mathcal{L}}}$ are computed in the same way $\widehat{{\cK}}$ and $\widehat{{\mathcal{L}}}$ were in~\eqref{ct-eq:approx K} and \eqref{ct-eq:approx L}. Moreover,
    \begin{equation*}
       \widehat{\rho}_{t, T, a}^n := \frac{\widehat{\varsigma}^2 \widehat{\af}^{n\cdot}\widehat{\Gamma}^{-1} \left(\frac{T}{L} \sum_{l=0}^{L}  \left(e^{-\widehat{\Gamma} u_l} - \Ir_I \right) \widehat{\Sigma}\widehat{\Sigma}^\top \left(e^{-\widehat{\Gamma} (u_l+a)} - \Ir_I \right)  \right) (\widehat{\overline{\af}}^{n\cdot}\widehat{\Gamma}^{-1})^\top}{\sqrt{\widehat{\mathcal{L}}^n(t,T)\widehat{\overline{\mathcal{L}}}^n(t,T+a)}},
    \end{equation*}
    and
    \begin{equation*}
        \widehat{\overline\omega}_{t, T, a}^n :=\frac{\log{\frac{\EAD_{t+T}^n}{(1-k)e^{-r a}}}- \widehat{\overline{\cK}}^n(\widehat{\dd}, t, T+a, \widehat{\cA}_t^\circ, \widehat{\cZ}_t)}{\sqrt{\widehat{\overline{\mathcal{L}}}^n(t,T+a)}},
    \end{equation*}
    with $a, k$, $\gamma$, $\EAD_t^n$, $\alpha^n$, and $\overline{F}^n_0$ are known, $\widehat\dd$ defined in \cref{ct-sec:calibintensities}, $\widehat\mu, \widehat\varsigma,\widehat v$ in \cref{ct-subsec:calibva}, and $\mathfrak{\widehat{\overline{a}}}^{n\cdot}, \widehat{\overline{\sigma}}_{n}$ in \cref{ct-sec:calibcoll_fin} and $\mathfrak{\overline{R}}^n_{t+T}(\widehat\dd)$ in~\eqref{ct-eq:FV_second_term}. Finally, $u_l:= \frac{T l}{L},l=0,\hdots,L$.

    \item If the collateral of loan~$n$ is a commercial or residential property, we compute in order~\eqref{ct-eq:mean col},~\eqref{ct-eq:var col},~\eqref{ct-eq: num LGD re}, and~\eqref{ct-eq:stressedLGD re}. Since $\chi_t = \varrho t + \vartheta$ and $C_0^n$ are known
    \begin{equation*}
        \widehat m_{t, T+a}^{n} := \log{(R_n C_0^n)} + (\widehat\varrho t + \widehat\vartheta) - (\widehat\vartheta-K_0)e^{-\widehat\nu (t+T+a)} + \widehat{\overline{\sigma}}\widehat\rho^\top \sum_{k=0}^{L} e^{-\widehat\nu((t+T+a)-\frac{k t}{L})} \eta_{u_\frac{k T}{L}},
    \end{equation*}
    where $\eta_{\frac{k t}{L}}\sim\cN\left(0, \frac{t}{L}\Ir_I\right), k=0,\hdots,L$ with $L\in\NN^*$, and
    \begin{equation*}
        \widehat v_{t, T+a}^{n} := \frac{(\widehat{\overline{\sigma}}\widehat\rho)^2)}{\widehat\nu}\left(1-e^{-2\widehat\nu (T+a)}\right)+   \frac{(\widehat{\overline{\sigma}})^2 (1-(\widehat\rho)^2)}{\widehat\nu}\left(1-e^{-2\widehat\nu (t+T+a)}\right).
    \end{equation*}
Therefore, we have
\begin{scriptsize}
\begin{equation}\label{ct-eq: approx LGD t real estate}
\begin{split}
    \widehat\LGD^n_{t,T,\delta} &= \frac{1-\gamma}{\widehat\PD^n_{t,T,\dd}} \left[\left(1+ (1-k)e^{-r a}\frac{R_n\widehat X_{t+T+a,\delta}^n}{\EAD_{t+T}^n}\right)\Phi_2\left(\widehat{\overline w}_{t, T, a}^n,\Phi^{-1}(\widehat\PD^n_{t,T,\dd});\widehat\rho_{t, T, a}^n\right)\right.\\
    &\qquad\qquad\quad\left.- \exp{\left(\frac{1}{2} \widehat v_{t, T+a}^{n}  -\sqrt{\widehat v_{t, T+a}^{n}}\widehat{\overline w}_{t, T, a}^n\right)}\Phi_2\left(\widehat{\overline w}_{t, T, a}^n-\sqrt{\widehat v_{t, T+a}^{n}},\Phi^{-1}(\widehat\PD^n_{t,T,\dd})-\widehat\rho_{t, T, a}^n\sqrt{\widehat v_{t, T+a}^{n}};\widehat\rho_{t, T, a}^n\right) \right],
\end{split}
\end{equation}
\end{scriptsize}
where
\begin{equation*}
    \widehat{\rho}_{t, T, a}^n := \widehat{\overline{\sigma}}\widehat{\varsigma} \frac{ \widehat{\af}^{n\cdot}\widehat{\Gamma}^{-1} \left(\frac{T}{L} \sum_{l=0}^{L} e^{-\widehat{\nu}(u_l+a)}\left(\Ir_I - e^{-\widehat{\Gamma} u_l} \right)\widehat{\Sigma} \right)\widehat{\rho}}{\sqrt{\widehat{\mathcal{L}}^n(t,T) \times \widehat v_{t, t+T+a}^{n}}},
\end{equation*}
and
\begin{equation*}
    \widehat{\overline w}_{t, T, a}^n := \log{\left(\frac{\EAD_{t+T}^n}{(1-k)R_n e^{-r a}} + \widehat X_{t+T+a,\delta}^n  \right)} - \widehat m_{t, T+a}^{n},
    \end{equation*}
and $\widehat X$ is obtained by considering that from~\eqref{ct-eq:coldynamics4},
\begin{equation*}
 \widehat X_{t,\delta}^n = \cc(\alpha^n, \alpha^\star) e^{-r(\ft_n-t)} + (\alpha^n-\alpha^\star)\frac{(\ft_n-t)}{P}\sum_{p=1}^{P} \ff(\delta_{v_p},\pf) e^{-r(v_p-t)},
\end{equation*}
 and where $\gamma$, $k$, $r$, $R_n$, and $\EAD_t^n$ are known, $\ft_n$ given by~\eqref{ct-eq:ex optimal t_n}, $u_l:= \frac{(t_\star-t) l}{L},l=0,\hdots,L$, and $v_p:= \frac{(\ft_n-t) p}{P},l=0,\hdots,P$.  
\end{enumerate}

\subsubsection{Of the (un)expected losses EL and UL}
For EL, the result is direct by using $\widehat{\PD}^{n}_{t,T,\dd}$ in~\eqref{ct-eq: approx PD t T}, $\widehat\LGD^n_{t,T,\dd }$ in~\eqref{ct-eq: approx LGD t T invest}, and $\widehat\LGD^n_{t,T,\delta}$ in~\eqref{ct-eq: approx LGD t real estate},  we have from~\eqref{ct-eq:stressedEL},
\begin{equation}\label{ct-eq:approx stressedEL}
    \begin{split}
        \widehat\EL^{N,T}_{t,\dd} 
        &:= \sum_{n=1}^{N_1} (1-\gamma) \EAD_{t+T}^n \cdot \widehat\PD^n_{t,T,\dd}+ \sum_{n=N_1+1}^{N_2} \EAD_{t+T}^n \cdot \widehat\LGD^n_{t,T,\dd } \cdot \widehat\PD^n_{t,T,\dd}\\
        &\qquad\qquad+ \sum_{n=N_2+1}^{N} \EAD_{t+T}^n \cdot \widehat\LGD^n_{t,T,\delta} \cdot \widehat\PD^n_{t,T,\dd}.
        \end{split}
    \end{equation}
For UL, we use 
\begin{equation}\label{ct-eq:approx proxy L}
    \begin{split}
        \widehat L_{t+T}^{\GG,N}
        &= \sum_{n=1}^{N_1} (1-\gamma) \EAD_{t+T}^n \cdot \widehat\PD_{t+T,\dd}^n + \sum_{n=N_1+1}^{N_2}\EAD_{t+T}^n \cdot \widehat\LGD_{t+T,\dd}^n\cdot \widehat\PD_{t+T,\dd}^n\\
        &\qquad\qquad+ \sum_{n=N_2+1}^{N} \EAD_{t+T}^n \cdot \widehat\LGD_{t+T,\delta}^n\cdot \widehat\PD_{t+T,\dd}^n, 
    \end{split}
\end{equation}
by noting that $\PD_{t+T,\dd}^n = \PD^{n}_{t+T,0,\dd}$, $\LGD_{t+T,\dd}^n = \LGD^n_{t+T, 0,\dd }$, and $\LGD_{t+T,\delta}^n = \LGD^n_{t+T,0,\delta}$. Therefore, as $\widehat L_{t+T}^{\GG,N}$ depends on 
$(\widehat\PD_{t+T,\dd}^n, \widehat\LGD_{t+T,\dd}^n, \widehat\LGD_{t+T,\delta}^n)$ which depends on $(\widehat\cA_{t+T}, \widehat\cZ_{t+T})$. However, we want to compute $\VaR^{\alpha, N,T}_t$ so that~$\PP\left[L_{t+T}^{\GG,N} \leq \VaR^{\alpha, N,T}_t\middle|\cG_t \right]$. Then, we simulate $D\in\NN^*$ couples noted~$(\widehat\cA_{t+T|t}^p, \widehat\cZ_{t+T|t}^p)_{1\leq p\leq D}$ so that $\widehat\cZ_{t+T|t}^p =^d \cZ_{t+T}|\cG_t$ and $\widehat\cA_{t+T|t}^p =^d \cA_{t+T}|\cG_t$. That is straightforward and 
\begin{equation}\label{eq:cond Z}
    \widehat\cZ_{t+T|t}^p|\cG_t \sim \cN\left( e^{-\widehat\Gamma T} \widehat\cZ_t, \frac{T}{L} \sum_{l=0}^{L} e^{-\widehat\Gamma u_l} \widehat\Sigma \widehat\Sigma^\top e^{-\widehat\Gamma^\top u_l} \right),
\end{equation} 
and
\begin{equation}\label{eq:cond A}
    \widehat\cA_{t+T|t}^p|\cG_t \sim \cN\left(\widehat\mu T + \widehat\varsigma\widehat\Upsilon_{T}\widehat\cZ_t + \widehat\cA_t, \widehat\varsigma^2 \widehat\Gamma^{-1} \left[\frac{T}{L} \sum_{l=0}^{L} \left(e^{-\widehat\Gamma u_l} - \Ir_I \right) \widehat\Sigma\widehat\Sigma^\top \left(e^{-\widehat\Gamma u_l} - \Ir_I \right)\right] (\widehat\Gamma^{-1})^\top \right),
\end{equation}
with $u_l:= \frac{T l}{L},l=0,\hdots,L$. We also need to simulate~$\mathfrak{h}_{t+T}|\cG_t := \int_{0}^{t+T} e^{-\nu(t+T-s)} \dr B^{\cZ}_s|\cG_t$ (which comes from~$m_{t+T,0}^{n}$ in~\eqref{ct-eq:mean col}). As $\mathfrak{h}_{t+T}|\cG_t \sim\cN\left(\int_{0}^{t} e^{-\nu(t+T-s)} \dr B^{\cZ}_s ,\frac{1-e^{-2\nu T}}{2\nu} \Ir_I \right)$, we have
\begin{equation*}
    \widehat{\mathfrak{h}}_{t+T}^p|\cG_t \sim\cN\left(\sum_{k=0}^{L} e^{-\widehat\nu((t+T)-\frac{k t}{L})} \eta_{u_\frac{k T}{L}} ,\frac{1-e^{-2\widehat\nu T}}{2\widehat\nu} \Ir_I \right),\quad \eta_{\frac{k t}{L}}\sim\cN\left(0, \frac{t}{L}\Ir_I\right).
\end{equation*}
Then, the unexpected loss is
    \begin{equation}\label{ct-eq:approx stressedUL}
    \widehat{\UL}^{N,T}_{t,\delta, \alpha} := q_{\alpha, D}\left(\left\{(\widehat L_{t+T}^{\GG,N})^1, (\widehat L_{t+T}^{\GG,N})^2,\hdots,(\widehat L_{t+T}^{\GG,N})^D\right\}\right) - \widehat{\EL}^{N,T}_{t,\dd},
\end{equation}
where $(\widehat L_{t+T}^{\GG,N})^p$ is obtained by replacing $(\widehat\cA_{t+T}, \widehat\cZ_{t+T})$ in \eqref{ct-eq:approx proxy L} by $(\widehat\cA_{t+T|t}^p, \widehat\cZ_{t+T|t}^p)$, and where $q_{\alpha, M}(\{Y^1,\hdots,Y^D\})$ denotes the empirical $\alpha$-quantile of the distribution of~$Y$.

{\color{blue}
\subsection{Summary of the process}
More concretely, the goal is to calculate the expected and unexpected losses at each $t\in[\ft_0,\ft_1]$. To achieve that, we use (1) the number of firms rated $r_{t}$ and defaulted $d_{t}$ in each $t\in\Ttt^M$, 
(2) all the firms' cash flows~$(F_{t}^n)_{1\leq n\leq N}$ in each $t\in\Ttt^M$, (3) for each borrower $n$, if a collateral exists and is a\textbf{ financial asset}, its cash flows~$(\overline{F}_{t}^n)_{t\in\Ttt^M}$ and if it is a \textbf{building}, its surface $R_n$ and its price at $\ft_1$, $C_0^n$, (4) the macroeconomic variables as well as the carbon intensities by sector, and the real estate index observed $REI_t$ in each $t\in\Ttt^M$, and (5) the carbon price dynamics~$(\delta_t)_{t\in\{t_\circ,\hdots, t_\star\}}$ given by the regulator. We proceed as follows.
\begin{enumerate}
    \item From the macroeconomic historical data, we estimate the productivity parameters~$\widehat{\Gamma}$, $\widehat{\mu}$,~$\widehat{\Sigma}$ and $\widehat{\varsigma}$, as well as the elasticities~$\widehat{\psi}$ and~$\widehat{\llambda}$ as described in~\cref{ct-subsec:calibva}. We also estimate the real estate market parameters $\widehat\varrho$, $\widehat\vartheta$, $\widehat\nu$, $\widehat{\overline{\sigma}}$, and $\widehat\rho$ as described in~\cref{ct-sec:calibcoll_house}.

    \item  For each~$i\in\{1, \ldots, I\}$, we estimate the parameters~$B^{n_i}$, $\sigma_{\bb^{n_i}}$, $\mathbf{a}^{n_i}$ using~\cref{ct-subsec:calibfirmp}, 
    yielding $\widehat{B}^{n_i}$, $\widehat{\sigma}_{\bb^{n_i}}$, $\widehat{\mathbf{a}}^{n_i}$.

    \item We compute the carbon price dynamics~$(\delta_t)_{t_\circ\leq t\leq t_\star}$ and the carbon intensities~$(\tau_t)_{t_\circ\leq t\leq t_\star}$, $(\zeta_t)_{t_\circ\leq t\leq t_\star}$, and $(\kappa_t)_{t_\circ\leq t\leq t_\star}$ as defined in~\cref{ct-sassc:price} as well as  the emissions cost rate~$(\dd_t)_{t_\circ\leq t\leq t_\star}$ defined in~\eqref{ct-eq:emiss cost rate} and the output carbon cost function~$v$ defined in~\eqref{ct-eq:pricefuncc}.
    
    \item For $t_\circ\leq t\leq t_\star$, we fix a large enough integer~$D$, and simulate $D$ paths of the productivity process~$(\widehat\cZ_t^d)_{1\leq d\leq D}$, then we derive~$((\cA_t^\circ)^d)_{1\leq d\leq D}$ as defined in~\cref{ct-sec:simul Theta and A}. Then, for each $p$, we simulate $D$ paths $(\widehat\cZ_{t+T|t}^{d'}|\cG_t^d)_{1\leq d\leq D}$ and $(\widehat\cZ_{t+T|t}^{d'}|\cG_t^d)_{1\leq d'\leq D}$ using \eqref{eq:cond Z} and \eqref{eq:cond A} respectively. 
    \item For each~$n\in\{1, \ldots, N\}$ and for each $1\leq d,d'\leq D$, 
    \begin{itemize}
        \item we compute the probability of default $\widehat\PD_{t+1,\dd}^n$ using~\eqref{ct-eq: approx PD t T}. 
    \item Depending on the nature of the collateral, we compute $\widehat\LGD_{t+1,\delta}^n$ , using either \eqref{ct-eq:no col}, \eqref{ct-eq: approx LGD t T invest}, or \eqref{ct-eq: approx LGD t real estate}
    \end{itemize}

    \item We finally compute the expected (resp. unexpected) losses $\widehat{\EL}^{N}_{t,\dd}$ (resp. $\widehat{\UL}^{N,T}_{t,\delta, \alpha}$), for each~$t_\circ\leq t\leq t_\star$, using~\eqref{ct-eq:approx stressedEL} (resp.~\eqref{ct-eq:approx stressedUL}).
\end{enumerate}

}

\section{Discussion}\label{sec:discussion}
In this section, we describe the data used to calibrate the different parameters, we perform some simulations, and we comment the results.
\subsection{Data}\label{result:data} As in~\cite{bouveret2023propagation} and \cite{sopgoui2024realestate}, we work on data related to the French economy.
\begin{enumerate}
    \item Due to data availability (precisely, we do not find public monthly/quaterly data for the intermediary inputs), we consider an annual frequency.
    \item\label{histo-macro-data} Annual consumption, labor, output, and intermediary inputs come from INSEE\footnote{The French National Institute of Statistics and Economic Studies} from 1978 to 2021 (see~\cite{insee2023sut} for details) and are expressed in billion euros, therefore $\ft_0=1978$, $\ft_1 = 2021$, and $M = 44$. 
    \item For the climate transition, we consider a time horizon of ten years with $t_\circ = 2021$ as starting point, a time step of one year and $t_\star = 2030$ as ending point.  In addition, we will be extending the curves to 2034 to see what happens after the transition, even though the results will be calculated and analyzed during the transition.
    \item The 38 INSEE sectors are grouped into four categories: \textit{Very High Emitting}, 
    \textit{Very Low Emitting}, 
    \textit{Low Emitting}, and
    \textit{High Emitting}, based on their carbon intensities.
    \item The carbon intensities are calibrated on the realized emissions from~\cite{emissions2020GHG}  (expressed in tonnes of CO2-equivalent) between 2008 and 2021.
    \item Metropolitan France housing price index comes from \textit{OECD\footnote{{The Organisation for Economic Co-operation and Development}} data} and are from 1980 to 2021 (see~\cite{oecd2024housingpindex} for details) in \textit{Base 2015}. We renormalize in \textit{Base 2021} {(i.e. in 2021, the housing price index is equal to $100$)}.
\end{enumerate}

\subsection{Definition of the climate transition} \label{ct-subsec:calibtaxes} 
We consider four deterministic transition scenarios giving four deterministic carbon price trajectories. The scenarios used come from the NGFS\footnote{{Network for Greening the Financial System.}} simulations, whose descriptions are given by~\cite{ngfs2020scenario} as follows:
\textit{\begin{itemize}
    \item \textbf{Net Zero 2050} is an ambitious scenario that limits global warming to $1.5^\circ C$ through stringent climate policies and innovation, reaching net zero $\mathrm{CO}_2$\footnote{{Carbon dioxide.}} emissions around 2050. Some jurisdictions such as {the United States of America, the European Union, and Japan} reach net zero for all GHG by this point.
   \item \textbf{Divergent Net Zero} reaches net-zero by 2050 but with higher costs due to divergent policies introduced across sectors and a quicker phase out of fossil fuels.
   \item \textbf{Nationally Determined Contributions (NDCs)} includes all pledged policies even if not yet implemented.
   \item \textbf{Current Policies} assumes that only currently implemented policies are preserved, leading to high physical risks.
\end{itemize}}
For each scenario, the regulator decides the carbon price $P_{carbon, 0}$ ({in euro per ton of CO$_2$-equivalent}) at $t_0$ and the speed rate of the transition~$\eta_\delta$. We summarize them in \cref{ct-tab:carbon_price_params} below.
\begin{table}[!ht]
\small \centering
\begin{tabular}{|r|r|r|r|r|}
\hline
\textit{}& \textbf{Current Policies} & \textbf{NDCs} & \textbf{Divergent Net Zero} & \textbf{Net Zero 2050} \\  \hline
\textbf{$P_{carbon, 0}$ (in euro/ton)} & 30.957&	33.321&	32.963&	34.315 \\ \hline
\textbf{$\eta_\delta$ (in \%)}& 1.693&	7.994&	12.893&	17.935\\ \hline
\end{tabular}
\caption{Carbon price parameters}
\label{ct-tab:carbon_price_params}
\end{table}
We can then compute the carbon price, whose evolution is plotted in Figure~\ref{fig:carbon_price_per_scenario}, at each date using~\eqref{ct-eq:carbon price}. 

For the energy price, we consider electricity as the unique source of energy. Then, we assume a linear relation between the electricity and the carbon price inspired by~\cite{abrell2023rising}, where a variation of the carbon price is linked withe the variation of the electricity by a the pass-through rate noted~$k$ \footnote{{We can definitively consider other types of energy. Moreover, instead of working with a linear relation (with a pass-through rate) between the energy price and the carbon price, we could consider a more complex model such as NiGEM (the National Institute Global Econometric Model in \cite{hantzsche2018using}).}}. This means that~$\ff^{elec}_1$ and~$\ff^{elec}_2$  define in~\eqref{ct-eq:f} are respectively~$k$ and $P_{elec, 0} - k \times P_{carbon, 0}$. For France, we take the electricity price $P_{elec, 0} = 0.2161$ euro per Kilowatt-hour and $k= 0.55$ (see~\cite{abrell2023rising}) ton per Kilowatt-hour. Its evolution is plotted in Figure~\ref{fig:Energy_price_per_scenario}.
\begin{figure}[!ht]
    \centering
    \begin{subfigure}[b]{0.4\textwidth}
        \includegraphics[width=0.95\textwidth]{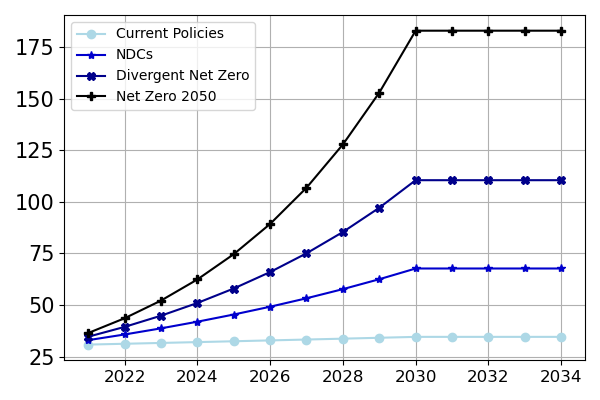}
        \caption{Carbon price}
        \label{fig:carbon_price_per_scenario}
    \end{subfigure}
    \begin{subfigure}[b]{0.4\textwidth}
        \includegraphics[width=0.95\textwidth]{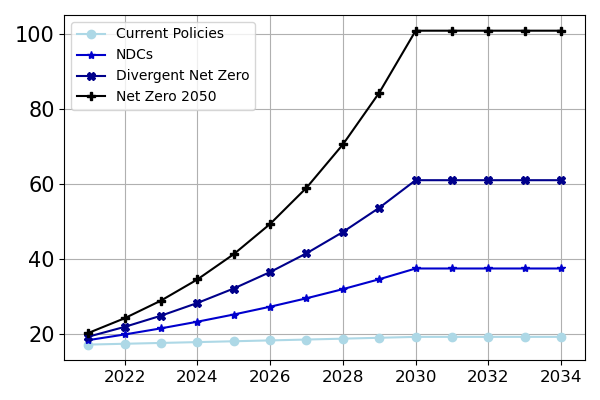}
        \caption{Energy price}
        \label{fig:Energy_price_per_scenario}
    \end{subfigure}
    \caption{per scenario and per year}
\end{figure}
For the renovation costs to improve a building for the energy efficiency~$\alpha$ to~$\alpha^\star$ as defined in~\eqref{ct-eq:c}, we take $c_0 = 0.01$ euro per kilowatt-hour and per square meter (\euro/kWh.m$^2$) and $c_1 = 0.1$.  

\subsection{Estimations}
\subsubsection{The carbon intensities}\label{ct-sec:resultintensities}
We use the realized GHG emissions as well as the macroeconomic variables and their frequency being the same as in~\cite{bouveret2023propagation}, we use the same estimations. But after that, we can compute the carbon intensities at each date in~$\RR_+$ using~\eqref{ct-eq:ghg_intensity}.

\subsubsection{Economic and housing pricing index (HPI) parameters} \label{ct-subsec:result_va}
We keep the values of $\phi$, $\sigma$, $(\chi^i)_{i\in\Jj}$, and $(\llambda^{ji})_{i,j\in\Jj}$ already estimated. For the productivity process, we switched from a vector autoregressive process to an Ornstein-Uhlenbeck process. We therefore calibrate ${\mu}$, ${\varsigma}$, ${\Sigma}$, and ${\Gamma}$ as detailed in \cref{ct-subsec:calibva} and we obtain the same results as in \cite{sopgoui2024realestate}[Section 4.3.1.].

 We write the housing price index~$K$ in \textit{Base 2021} and we apply the logarithm function. This means that $K_{t_0} = 0$. 
 We can therefore calibrate~$\varrho, \vartheta, \nu, \overline\sigma$, and $\rho$. The values are presented in \cite{sopgoui2024realestate}[Table 5]. 

\subsection{Simulations and discussions}
In the previous work in discrete time, we simulate for different climate transition scenario between $t_\circ = 2021$ and $t_\star=2030$, the annual evolution of (1) the output growth per sector (2) the output share per sector in the total output, (3) the firms direct GHG emissions per sector, (4) a given firm value and distribution, (5) the probabilities of default of fictive sub-portfolio of 4 firms each and of the resulting portfolio, (6) the expected and the unexpected losses of the previous (sub-)portfolios when the LGD are constant and deterministic, (7) the sensitivities of the losses to the carbon price.

In the current simulations, since we are keeping the same data at the same frequency (annual), the main change is then to replace the {Vector Autoregressive} process by the O.-U. process. Therefore, the comments already made for (1) to (5) concerning the trends, the impact of the carbon price, the difference of scenarios, the relation between sectors, etc. do not change. We will focus here on the LGD and on the losses, with different type of collateral.\\
{We will assume $D = 500$. This may seem low but it is actually $500 \times 500$ simulations. This is  computationally-intensive. In fact, for each year $t$, we firstly simulate $500$ times $(\cZ_t,\cA_t)$. However, we are interested in one year risk measures \footnote{Recall for example that we are in 2025, we would like to compute the probability that a firm defaults between 1 January 2027 and 31 December 2027.}, therefore we also simulate $500$ times $(\cZ_{t+1},\cA_{t+1})$ given $(\cZ_t,\cA_t)$. Moreover, we considered 4 sectors and for each step, there are numerical integrals, matrices products, etc. to compute. }

\subsubsection{Impact of the carbon price on Loss Given Default}
When there is no guarantee, we assume as in the previous work that $\LGD$ is equal to $45\%$ so that $\gamma = 0.55$. To illustrate the case where there is guarantee, we consider, both if the collateral is a \textit{financial asset} and a \textit{building}, $\EAD$ starts at $200$ and growths annually as the economic total output growth in the \textit{Current Policies} scenario (see Table~\ref{ct-tab:EAD_per_year} below).

\begin{table}[ht!]
\footnotesize \centering
\begin{tabular}{|r|r|r|r|r|r|r|r|r|r|r|r|r|r|r|r|r|}
\hline
\textit{\textbf{Year}} & \textbf{2021} & \textbf{2022} & \textbf{2023} &  \textbf{2024} & \textbf{2025} & \textbf{2026} & \textbf{2027} & \textbf{2028} &  \textbf{2029} & \textbf{2030}& \textbf{2031} & \textbf{2032} &  \textbf{2033} & \textbf{2034} \\ \hline
\textit{\textbf{EAD}}    & 200. & 202.8& 206.6& 209.9& 213.1& 216.2& 219.5& 222.6& 226.2& 229.6& 233.1& 236.3& 239.9& 243.9  \\ \hline
\end{tabular}
\caption{EAD per year}
\label{ct-tab:EAD_per_year}
\end{table}

\paragraph{If the collateral is a financial asset}
We consider 4 firms so that firm $1$, $2$, $3$, and $4$ respectively belong to the \textit{Very High Emitting}, \textit{High Emitting}, \textit{Low Emitting}, and \textit{Very Low Emitting} groups. Each firm is characterized by its cash flows~$F_{t_\circ-1}$ at $t_\circ-1$, the standard deviation of its cash flows~$\sigma_{\bb}$, and the contribution~$\mathfrak{a}$ of sectoral output growth to its cash flows growth as detailed in table~\ref{ct-tab:impact_on_fv}. The chosen interest rate $r = 5\%$. 
\begin{table}[!ht]
\small \centering
\begin{tabular}{|l|r|r|r|r|r|}
\hline
{ \textbf{Firm }}          & \multicolumn{1}{l|}{\textbf{$1$}} & \multicolumn{1}{l|}{{ \textbf{$2$}}} & \multicolumn{1}{l|}{{ \textbf{$3$}}} & \multicolumn{1}{l|}{{ \textbf{$4$}}} \\ \hline
{ \textbf{$\sigma_{\bb^n}$}} & 0.05& 0.05& 0.05& 0.05\\ \hline\hline
{ \textbf{$F_0^n$}}          & 1.0  & 1.0  & 1.0  & 1.0 \\ \hline\hline
{ \textbf{$\mathfrak{a}^n$(\textbf{Very High})}} & 1.0 & 0.0 & 0.0 & 0.0 \\ \hline
{ \textbf{$\mathfrak{a}^n$(\textbf{High})}}  & 0.0 & 1.0    & 0.0  & 0.0  \\ \hline
{ \textbf{$\mathfrak{a}^n$(\textbf{Low})}}                 & 0.0 & 0.0    & 1.0    & 0.0 \\ \hline
{ \textbf{$\mathfrak{a}^n$(\textbf{Very Low})}}                 & 0.0 & 0.0    & 0.0    & 1.0 \\ \hline
\end{tabular}
\caption{Characteristics of the firms}
\label{ct-tab:impact_on_fv}
\end{table}
We compute here for $D=500$ simulations of the productivity processes~$(\cZ, \cA)$, the loss given default of 4 loans with the same exposure but with 4 different financial assets collateral described in Table~\ref{ct-tab:impact_on_fv}.

\begin{table}[ht!]
\small \centering
\begin{tabular}{|r|r|r|r|r|r|}
\hline
\textit{\textbf{Emissions level}} & \textbf{No collateral} & \textbf{Firm 1} & \textbf{Firm 2} &  \textbf{Firm 3} & \textbf{Firm 4} \\ \hline
\textit{\textbf{Current Policies}}          & 45.&	32.934&	31.960&	34.561&	29.281 \\ \hline
\textit{\textbf{NDCs}}          & 45.&	33.177&	32.184&	34.609&	29.357 \\ \hline
\textit{\textbf{Divergent Net Zero}}  & 45.	&33.485&	32.471&	34.673&	29.459 \\ \hline
\textit{\textbf{Net Zero 2050}}     & 45.&	33.995&	32.940&	34.784&	29.640 \\ \hline
\end{tabular}
\caption{Average annual LGD per scenario between 2021 and 2030 (in \%)}
\label{ct-tab:LGD_with_fa}
\end{table}

Both in Table~\ref{ct-tab:LGD_with_fa} and in Figure~\ref{ct-fig:LGD_with_fa}, we can first see that the presence of guarantees reduce $\LGD$. Without collateral, we assume 45\%, and with collateral, for all scenarios and for different characteristics of firms, $\LGD$ is less than 45\%.

\begin{figure}[!ht]
    \centering
    \includegraphics[width=0.95\textwidth]{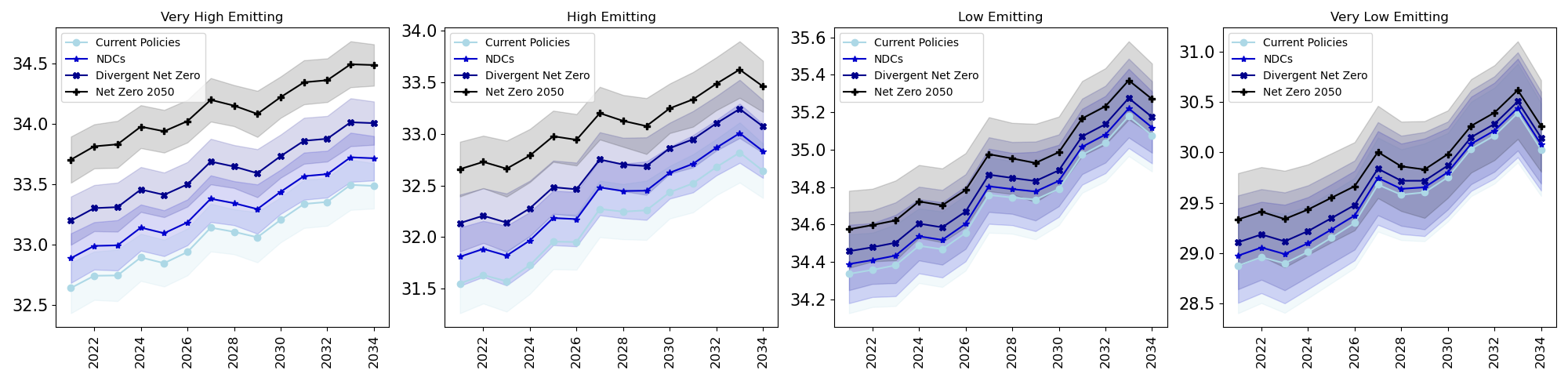}
    \caption{LGD with a financial asset as collateral}
    \label{ct-fig:LGD_with_fa}
\end{figure}
However, the decreasing of LGD depends on the scenarios. When the scenario becomes tougher, the impact of the presence of the collateral on LGD is lessened. This is logical and due to the fact that the value of the liquidated asset loses value when the price of carbon rises. The decreasing of the LGD also depends on the distinctive characteristics of the guarantees. Precisely, each firm in Table~\ref{ct-tab:impact_on_fv}, serving as collateral, belongs to a unique and distinct sector (through~$\af$), which go from the more to the less polluting. Therefore the more the collateral is in a polluting sector, the less it reduces $\LGD$.

\paragraph{If the collateral is a building}
We consider 5 apartments of 25 square meters whose price of the square meter fixed to 4000 euros in $t_\circ=2021$ is the same for all, but whose the energy efficiency are different.
Moreover, we assume that the optimal energy efficiency equals to $\alpha^\star = 70$ kilowatt hour per square meter per year (see~\cite{TotalEnergies2024dpe}) is reached.
\begin{table}[ht!]
\small\centering
\begin{tabular}{|l|r|r|r|r|r|}
\hline
{ \textbf{Building }}          & \textbf{$1$} &  \textbf{$2$} &  \textbf{$3$} & \textbf{$4$}&  \textbf{$5$} \\ \hline
{ \textbf{$C_n^0$}} & 4000& 4000& 4000& 4000& 4000\\ \hline\hline
{ \textbf{$\alpha^n$}}          & 320.& 253.& 187.& 120.&70.\\ \hline\hline
{ \textbf{$R_n$}} & 25.0 &  25.0 & 25.0 & 25.0 & 25.\\ \hline
\end{tabular}
\caption{Characteristics of the building}
\label{ct-tab:impact_on_hp}
\end{table}

We use the $D=500$ trajectories of the productivity processes~$(\cZ, \cA)$ simulated above. We compute the loss given default of 4 loans with the same exposure but with the 4  buildings described in Table~\ref{ct-tab:impact_on_hp} as collateral.

\begin{table}[ht!]
\small \centering
\begin{tabular}{|r|r|r|r|r|r|}
\hline
\textit{\textbf{Emissions level}} & \textbf{No collateral} & \textbf{Building 1} & \textbf{Building 2} &  \textbf{Building 3} & \textbf{Building 4} \\ \hline
\textit{\textbf{Current Policies}}          & 45.&	36.020&	36.152&	35.928&	36.095 \\ \hline
\textit{\textbf{NDCs}}          & 45.&	38.383&	37.752&	36.922&	36.499 \\ \hline
\textit{\textbf{Divergent Net Zero}}  & 45.	&38.939	&38.303&	37.377&	36.751 \\ \hline
\textit{\textbf{Net Zero 2050}}     & 45.&	39.102&	38.524&	37.615&	36.908 \\ \hline
\end{tabular}
\caption{Average annual LGD per scenario between 2021 and 2030 (in \%)}
\label{ct-tab:LGD_with_hou}
\end{table}

As we see in \cref{ct-tab:LGD_with_hou,ct-fig:LGD_with_hou}, all the comments made for a financial asset as collateral are valid for a building: the presence of a collateral reduces LGD, but this reduction is smaller when the climate transition scenario becomes tougher.

\begin{figure}[!ht]
    \centering
    \includegraphics[width=0.95\textwidth]{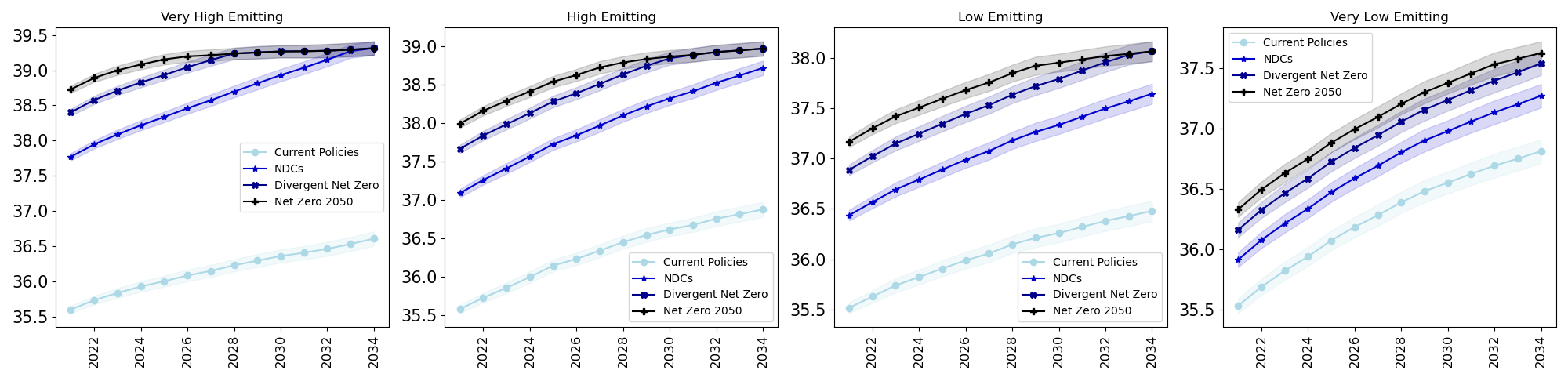}
    \caption{LGD with a building as collateral}
    \label{ct-fig:LGD_with_hou}
\end{figure}
There are two main differences. First, the more the building is energetically inefficient, the more LGD increases (it is the same above when the financial asset belongs to a very polluting sector). Secondly, LGD decreases when time increases. This is a consequence of the dynamics of the impact of the carbon price of the housing market (as described in \cite{sopgoui2024realestate}): as we approach the optimal renovation date, the prices of energy-inefficient buildings rise and converge progressively towards the prices of energy-efficient buildings (we can see on \cite{sopgoui2024realestate}[Figure 3]). LGD follows the same behaviour logically but with an inverse monotony.

\subsubsection{Expected and unexpected loss}
To this aim, to keep things simple, we will consider a credit portfolio of $N=12$ loans contracted by the firms described in Table~\ref{ct-tab:portfolio} below. 
\begin{table}[ht!]
\small
\begin{tabular}{|l||llll||llll||llll|}
\hline
\textbf{Loans} &
  {\textbf{1}} &
  {\textbf{2}} &
  {\textbf{3}} &
  \textbf{4} &
  {\textbf{5}} &
  {\textbf{6}} &
  {\textbf{7}} &
  \textbf{8} &
  {\textbf{9}} &
  {\textbf{10}} &
  {\textbf{11}} &
  \textbf{12} \\ \hline
\textit{\textbf{$\EAD_n$}} &
  \multicolumn{1}{l|}{200.} &
  \multicolumn{1}{l|}{200.} &
  \multicolumn{1}{l|}{200.} &
  200. &
  \multicolumn{1}{l|}{200.} &
  \multicolumn{1}{l|}{200.} &
  \multicolumn{1}{l|}{200.} &
  200. &
  \multicolumn{1}{l|}{200.} &
  \multicolumn{1}{l|}{200.} &
  \multicolumn{1}{l|}{200.} &
  200. \\ \hline
\textit{\textbf{$F_0^n$}} &
  \multicolumn{1}{l|}{1.} &
  \multicolumn{1}{l|}{1.} &
  \multicolumn{1}{l|}{1.} &
  1. &
  \multicolumn{1}{l|}{1.} &
  \multicolumn{1}{l|}{1.} &
  \multicolumn{1}{l|}{1.} &
  1. &
  \multicolumn{1}{l|}{1.} &
  \multicolumn{1}{l|}{1.} &
  \multicolumn{1}{l|}{1.} &
  1. \\ \hline
\textit{\textbf{$B^n$}} &
  \multicolumn{1}{l|}{3.76} &
  \multicolumn{1}{l|}{3.98} &
  \multicolumn{1}{l|}{3.75} &
  4.41 &
  \multicolumn{1}{l|}{3.76} &
  \multicolumn{1}{l|}{3.98} &
  \multicolumn{1}{l|}{3.75} &
  4.41 &
  \multicolumn{1}{l|}{3.76} &
  \multicolumn{1}{l|}{3.98} &
  \multicolumn{1}{l|}{3.75} &
  4.41 \\ \hline
\textit{\textbf{$\sigma_{\bb^n}$}} &
  \multicolumn{1}{l|}{0.05} &
  \multicolumn{1}{l|}{0.05} &
  \multicolumn{1}{l|}{0.05} &
  0.05 &
  \multicolumn{1}{l|}{0.05} &
  \multicolumn{1}{l|}{0.05} &
  \multicolumn{1}{l|}{0.05} &
  0.05 &
  \multicolumn{1}{l|}{0.05} &
  \multicolumn{1}{l|}{0.05} &
  \multicolumn{1}{l|}{0.05} &
  0.05 \\ \hline
\textit{\textbf{$\af^n$(Very High)}} &
  \multicolumn{1}{l|}{1.} &
  \multicolumn{1}{l|}{0.} &
  \multicolumn{1}{l|}{0.} &
  0. &
  \multicolumn{1}{l|}{1.} &
  \multicolumn{1}{l|}{0.} &
  \multicolumn{1}{l|}{0.} &
  0. &
  \multicolumn{1}{l|}{1.} &
  \multicolumn{1}{l|}{0.} &
  \multicolumn{1}{l|}{0.} &
  0. \\ \hline
\textit{\textbf{$\af^n$(High)}} &
  \multicolumn{1}{l|}{0.} &
  \multicolumn{1}{l|}{1.} &
  \multicolumn{1}{l|}{0.} &
  0. &
  \multicolumn{1}{l|}{0.} &
  \multicolumn{1}{l|}{1.} &
  \multicolumn{1}{l|}{0.} &
  0. &
  \multicolumn{1}{l|}{0.} &
  \multicolumn{1}{l|}{1.} &
  \multicolumn{1}{l|}{0.} &
  0. \\ \hline
\textit{\textbf{$\af^n$(Low)}} &
  \multicolumn{1}{l|}{0.} &
  \multicolumn{1}{l|}{0.} &
  \multicolumn{1}{l|}{1.} &
  0. &
  \multicolumn{1}{l|}{0.} &
  \multicolumn{1}{l|}{0.} &
  \multicolumn{1}{l|}{1.} &
  0. &
  \multicolumn{1}{l|}{0.} &
  \multicolumn{1}{l|}{0.} &
  \multicolumn{1}{l|}{0.} &
  0. \\ \hline
\textit{\textbf{$\af^n$(Very Low)}} &
  \multicolumn{1}{l|}{0.} &
  \multicolumn{1}{l|}{0.} &
  \multicolumn{1}{l|}{0.} &
  1. &
  \multicolumn{1}{l|}{0.} &
  \multicolumn{1}{l|}{0.} &
  \multicolumn{1}{l|}{0.} &
  1. &
  \multicolumn{1}{l|}{0.} &
  \multicolumn{1}{l|}{0.} &
  \multicolumn{1}{l|}{0.} &
  0. \\ \hline
\textit{\textbf{Collateral type}} &
  \multicolumn{1}{l|}{No} &
  \multicolumn{1}{l|}{No} &
  \multicolumn{1}{l|}{No} &
  No &
  \multicolumn{1}{l|}{Fa} &
  \multicolumn{1}{l|}{Fa} &
  \multicolumn{1}{l|}{Fa} &
  Fa &
  \multicolumn{1}{l|}{Ho} &
  \multicolumn{1}{l|}{Ho} &
  \multicolumn{1}{l|}{Ho} &
  Ho \\ \hline
\textit{\textbf{$\overline{F}_0^n$}} &
  \multicolumn{4}{l|}{\multirow{10}{*}{}} &
  \multicolumn{1}{l|}{1.} &
  \multicolumn{1}{l|}{1.} &
  \multicolumn{1}{l|}{1.} &
  1. &
  \multicolumn{4}{l|}{\multirow{6}{*}{}} \\ \cline{1-1} \cline{6-9}
\textit{\textbf{$\overline{\sigma}_{\bb^n}$}} &
  \multicolumn{4}{l|}{} &
  \multicolumn{1}{l|}{0.05} &
  \multicolumn{1}{l|}{0.05} &
  \multicolumn{1}{l|}{0.05} &
  0.05 &
  \multicolumn{4}{l|}{} \\ \cline{1-1} \cline{6-9}
\textit{\textbf{$\overline{\af}^n$(Very High)}} &
  \multicolumn{4}{l|}{} &
  \multicolumn{1}{l|}{1.} &
  \multicolumn{1}{l|}{0.} &
  \multicolumn{1}{l|}{0.} &
  0. &
  \multicolumn{4}{l|}{} \\ \cline{1-1} \cline{6-9}
\textit{\textbf{$\overline{\af}^n$(High)}} &
  \multicolumn{4}{l|}{} &
  \multicolumn{1}{l|}{0.} &
  \multicolumn{1}{l|}{1.} &
  \multicolumn{1}{l|}{0.} &
  0. &
  \multicolumn{4}{l|}{} \\ \cline{1-1} \cline{6-9}
\textit{\textbf{$\overline{\af}^n$(Low)}} &
  \multicolumn{4}{l|}{} &
  \multicolumn{1}{l|}{0.} &
  \multicolumn{1}{l|}{0.} &
  \multicolumn{1}{l|}{1.} &
  0. &
  \multicolumn{4}{l|}{} \\ \cline{1-1} \cline{6-9}
\textit{\textbf{$\overline{\af}^n$(Very Low)}} &
  \multicolumn{4}{l|}{} &
  \multicolumn{1}{l|}{0.} &
  \multicolumn{1}{l|}{0.} &
  \multicolumn{1}{l|}{0.} &
  1. &
  \multicolumn{4}{l|}{} \\ \cline{1-1} \cline{6-13} 
\textit{\textbf{$C_n^0$}} &
  \multicolumn{4}{l|}{} &
  \multicolumn{4}{l|}{\multirow{4}{*}{}} &
  \multicolumn{1}{l|}{4000.} &
  \multicolumn{1}{l|}{4000.} &
  \multicolumn{1}{l|}{4000.} &
  4000. \\ \cline{1-1} \cline{10-13} 
\textit{\textbf{$R_n$}} &
  \multicolumn{4}{l|}{} &
  \multicolumn{4}{l|}{} &
  \multicolumn{1}{l|}{25.} &
  \multicolumn{1}{l|}{25.} &
  \multicolumn{1}{l|}{25.} &
  25. \\ \cline{1-1} \cline{10-13} 
\textit{\textbf{$\alpha^n$}} &
  \multicolumn{4}{l|}{} &
  \multicolumn{4}{l|}{} &
  \multicolumn{1}{l|}{320.} &
  \multicolumn{1}{l|}{253.} &
  \multicolumn{1}{l|}{187.} &
  120. \\ \cline{1-1} \cline{1-13} 
\end{tabular}
\caption{Characteristics of the portfolio (No = no collateral, Fa = Financial asset collateral, Ho = housing collateral)}
\label{ct-tab:portfolio}
\end{table}

We can remark that, {for each $k=0,1,2$}, firms $4k+1$, $4k+2$, $4k+3$, and $4k+4$ respectively belong to the \textit{Very High Emitting}, \textit{High Emitting}, \textit{Low Emitting}, and \textit{Very Low Emitting} groups. Moreover, we assume that
\begin{itemize}
    \item the loans of the firms $1$, $2$, $3$, and $4$ are not collateralized;
    \item the loans of the firms $5$, $6$, $7$, and $8$ are collateralized by financial assets described in Table~\ref{ct-tab:impact_on_fv};
    \item the loans of the firms $9$, $10$, $11$, and $12$ are collateralized by a building described in Table~\ref{ct-tab:impact_on_hp}.
\end{itemize}
We want to calculate the expected (respectively unexpected) loss noted~$\EL$ (respectively~$\UL$) for each loan~$n=1,\hdots,12$, by using~\eqref{ct-eq:approx stressedEL} (respectively~\eqref{ct-eq:approx stressedUL}).
\begin{table}[ht!]
\small \centering
\begin{tabular}{|r||r|r|r|r||r|r|r|r||r|r|r|r|r|}
\hline
\textit{\textbf{Emissions level}} & 
  {\textbf{1}} &
  {\textbf{2}} &
  {\textbf{3}} &
  \textbf{4} &
  {\textbf{5}} &
  {\textbf{6}} &
  {\textbf{7}} &
  \textbf{8} &
  {\textbf{9}} &
  {\textbf{10}} &
  {\textbf{11}} &
  \textbf{12} \\ \hline
\textit{\textbf{Current Policies}} & 1.00&	1.00&	1.00&	1.00& 0.79&	0.77&	0.81&	0.71& 0.58&	0.59&	0.58&	0.58 \\ \hline
\textit{\textbf{NDCs}} & 1.28&	1.17&	1.04&	1.02& 1.01&	0.90&	0.85&	0.72& 0.91&	0.79&	0.66&	0.62\\ \hline
\textit{\textbf{Divergent Net Zero}}  & 1.74&	1.41&	1.10&	1.05& 1.38&	1.10&	0.90&	0.75& 1.29&	0.99&	0.72&	0.65\\ \hline
\textit{\textbf{Net Zero 2050}} & 2.85&	1.91&	1.21&	1.11& 2.27&	1.47&	0.98&	0.79& 2.13	&1.37&	0.81&	0.70\\ \hline
\end{tabular}
\caption{Average annual EL per scenario between 2021 and 2030 (in \%)}
\label{ct-tab:EL}
\end{table}

Table~\ref{ct-tab:EL} (respectively Table~\ref{ct-tab:UL}) shows average annual EL (respectively UL) normalized to the EL without collateral observed in the scenario \textit{Current Policies}. We can make two key observations that were to be expected from the PD and LGD calculations:
\begin{enumerate}
    \item Whether collateral is involved or not, we can see that EL and UL increase as the transition hardens. This is to be expected, since PD and LGD behave in the same way.
    \item When a loan is collateralized, it significantly reduces the bank’s expected and unexpected losses. And for collateralized loans, these losses increase if the collateral has a high carbon footprint: in particular, if the collateral is a financial asset whose value growth is driven by a polluting sector or if it is a building that is not energy efficient. 
\end{enumerate}

\begin{table}[ht!]
\small \centering
\begin{tabular}{|r||r|r|r|r||r|r|r|r||r|r|r|r|r|}
\hline
\textit{\textbf{Emissions level}} & 
  {\textbf{1}} &
  {\textbf{2}} &
  {\textbf{3}} &
  \textbf{4} &
  {\textbf{5}} &
  {\textbf{6}} &
  {\textbf{7}} &
  \textbf{8} &
  {\textbf{9}} &
  {\textbf{10}} &
  {\textbf{11}} &
  \textbf{12} \\ \hline
\textit{\textbf{Current Policies}} & 1.00&	1.00&	1.00&	1.00& 0.81	&0.81&	0.84&	0.82 &0.59&	0.57&	0.56&	0.62\\ \hline
\textit{\textbf{NDCs}} & 1.18&	1.01&	1.00&	1.00 &0.95&	0.82	&0.84&	0.82& 0.85&	0.67&	0.62&	0.64\\ \hline
\textit{\textbf{Divergent Net Zero}}  & 1.42	&1.02&	1.00&	1.00& 1.16	&0.83&	0.84&	0.82& 1.05&	0.71&	0.64	&0.65\\ \hline
\textit{\textbf{Net Zero 2050}} & 1.80	&1.02&	1.0&	0.99 &	1.49&	0.85&	0.85&	0.82& 1.34&	0.73&	0.65&	0.66\\ \hline
\end{tabular}
\caption{Average annual UL per scenario between 2021 and 2030 (in \%)}
\label{ct-tab:UL}
\end{table}
 We can finally say that introducing and increasing the carbon price in the economy will
\begin{enumerate}
    \item increase bank charges (materialized by the level of provisions calculated from the expected loss~$EL$) charged to clients or by the operating income of banks; 
    \item and reduce the solvency and profitability of banks (translated by the economic capital calculated from the unexpected loss~$UL$).
\end{enumerate}

{For example, consider the \textbf{Net Zero 2050} scenario, for a borrower belonging to a very high GHG-emitting sector and its collateral too, the bank will see its associated EL increases by 187\% and its associated UL by 83\%. Conversely, if the borrower is in a very low GHG emitting sector and its collateral too, the bank will obtain its associated EL increases by 11\% and its UL unchanged.}

\section*{Conclusion}
Following~\cite{bouveret2023propagation}, we developed here  a framework to quantify the impacts of the carbon price on a credit portfolio (expected and unexpected) losses, when the obligor companies as well as their guarantees belong to an economy subject to the climate transition declined by carbon price. We start by describing a closed economy, driven by a productivity following a multidimensional Ornstein-Uhlenbeck and subject to a climate transition modeled through a dynamic and deterministic carbon price, by a dynamic stochastic multisectoral. Then, by using the discounted cash flow methodology with the cash flows, following a stochastic differential equation, depending on the productivity as well as the carbon price, we evaluate the obligor value that helps us later on to compute its probability of \textit{over-indebtedness}. We then turn to the bank's loss in the event of a borrower's \textit{over-indebtedness} and if its loan is collateralized. When that is the case, the potential loss of the bank is written as the difference between the debt amount (EAD) and the collateral liquidated. We finally distinguish two types of collateral: either a financial asset or a building, both belonging to the economy so affected by the productivity and the carbon price. This work opens the door to many extensions as a finer modeling of the real estate market, taking into account other types of guarantees, modeling the unsecured loans that we assumed constant, modelling the impact of the carbon price on the exposure. This work can be extended to include physical risk. In order to implement that, we could for example introduce the damage multiplier defined through the temperature, and which will affected the output of each sector. Therefore, for each sector $i\in\Jj$, we will have
\begin{align}
    \exp{\left(-\frac{\gamma_i}{2} T_t^2\right)},
\end{align}
where $\gamma_i\geq 0$ is the damage function coefficient and $T_t$ is the mean temperature of the earth relative to its
pre-industrial level.

\newpage
\small
\bibliographystyle{apalike}
\bibliography{main}
\newpage

\appendix

\normalsize
\section{Proofs}
\subsection{Hurwitz matrix}\label{ct-app:Hurwitz-matrix}
Assume that $-\Gamma$ is a Hurwitz matrix, then
\begin{enumerate}
    \item if we note $\lambda_\Gamma :=\max_{\lambda\in\lambda(\Gamma)} Re(\lambda) \geq 0$, there exists $c_\Gamma>0$ so that $\lVert e^{-\Gamma t}\rVert < c_\Gamma e^{-\lambda_\Gamma t}$ for all~$t\geq 0$.
    \item Moreover, for $t\geq 0$n $\Upsilon_{t}$ defined in~\eqref{ct-eq:Upsilon} is such that
    \begin{align}\label{ct-eq:Upsilon bound}
        \lVert\Upsilon_{t}\rVert = \left\lVert\int_{0}^{t} e^{-\Gamma s} \dr s\right\rVert \leq \int_{0}^{t} \left\lVert e^{-\Gamma s}\right\rVert \dr s \leq c_\Gamma \int_{0}^{t} e^{-\lambda_\Gamma s}\dr s \leq c_\Gamma \min{\left\{\frac{1}{\lambda_\Gamma}, t\right\}}.
    \end{align}
\end{enumerate}

\subsection{Bivariate Gaussian}\label{ct-app:Bivariate gaussian}

Assume that $X$ and $Y$ are two standard Gaussian with correlation coefficient $\rho$. We then have for $(x,y)\in\RR^2$, the cdf,
\begin{equation}
    \Phi_2(x,y) := \PP[X\leq x, Y\leq y] = \frac{1}{2\pi(1-\rho^2)} \int_{-\infty}^{x} \int_{-\infty}^{y} \exp{\left(-\frac{1}{2(1-\rho^2)} \left(u^2+v^2 - 2\rho u v\right)\right)} \dr u\dr v.
\end{equation}
Let $\sigma > 0$, we want to compute $\EE[e^{\sigma X} \bOne_{X\leq x, Y\leq y}]$. We have
\begin{equation*}
\begin{split}
     \EE[e^{\sigma X} \bOne_{X\leq x, Y\leq y}] &= \frac{1}{2\pi\sqrt{1-\rho^2}} \int_{-\infty}^{x} \int_{-\infty}^{y}e^{\sigma u} \exp{\left(-\frac{1}{2(1-\rho^2)} \left(u^2+v^2 - 2\rho u v\right)\right)} \dr u\dr v\\
     &= \frac{1}{2\pi\sqrt{1-\rho^2}} \int_{-\infty}^{x} e^{\sigma u -\frac{1}{2(1-\rho^2)} u^2} \int_{-\infty}^{y} \exp{\left(-\frac{1}{2(1-\rho^2)} \left(v^2 - 2\rho u v\right)\right)}\dr v\quad\dr u\\
     &= \frac{1}{2\pi\sqrt{1-\rho^2}} \int_{-\infty}^{x} e^{\sigma u -\frac{1}{2(1-\rho^2)} u^2} \int_{-\infty}^{y} \exp{\left(-\frac{1}{2(1-\rho^2)} \left((v - \rho u)^2- \rho^2 u^2\right)\right)}\dr v\quad\dr u\\
     &= \frac{1}{2\pi\sqrt{1-\rho^2}} \int_{-\infty}^{x} e^{\sigma u - \frac{1}{2}u^2} \int_{-\infty}^{y} \exp{\left(-\frac{1}{2(1-\rho^2)} \left((v - \rho u)^2\right)\right)}\dr v\quad\dr u\\
\end{split}
\end{equation*}
But \begin{equation*}
    \int_{-\infty}^{y} \exp{\left(-\frac{1}{2(1-\rho^2)} \left((v - \rho u)^2\right)\right)}\dr v = \sqrt{2\pi(1-\rho^2)} \Phi\left(\frac{y - \rho u}{\sqrt{1-\rho^2}}\right),
\end{equation*}
therefore,
\begin{equation*}
\begin{split}
     \EE[e^{\sigma X} \bOne_{X\leq x, Y\leq y}] &= \frac{1}{\sqrt{2\pi}} \int_{-\infty}^{x} e^{\sigma u-\frac{1}{2}u^2}\Phi\left(\frac{y - \rho u}{\sqrt{1-\rho^2}}\right)\dr u\\
    &= \frac{e^{\frac{1}{2}\sigma^2}}{\sqrt{2\pi}} \int_{-\infty}^{x} e^{-\frac{1}{2}(u-\sigma)^2}\Phi\left(\frac{y - \rho u}{\sqrt{1-\rho^2}}\right)\dr u\\
     &= e^{\frac{1}{2}\sigma^2} \int_{-\infty}^{x-\sigma}\phi(u)\Phi\left(\frac{y-\rho\sigma}{\sqrt{1-\rho^2}} + \frac{-\rho}{\sqrt{1-\rho^2}}u\right)\dr u.
\end{split}
\end{equation*}
However, \begin{equation*}
    \int_{-\infty}^{c} \Phi(a + b x) \phi(x) \mathrm{d} x = \Phi_2\left(c,\frac{a}{\sqrt{1+b^2}}; \frac{-b}{\sqrt{1+b^2}}\right),
\end{equation*}
we can then conclude that 
\begin{equation}\label{ct-INT:int2}
    \EE[e^{\sigma X} \bOne_{X\leq x, Y\leq y}] =  e^{\frac{1}{2}\sigma^2}\Phi_2\left(x-\sigma,y-\rho\sigma; \rho\right).
\end{equation}

\section{The multisectoral model in continuous time}\label{ct-app-sec:multisec}
For all~$i\in\Jj$, let us consider the following $\GG$-measurable and positive processes: $Y^i$ the production of sector~$i$, $N^i$ the labor demand in sector~$i$, and for all~$j\in\Jj$, $Z^{ji}$ the consumption by sector~$i$ of intermediate inputs produced by sector~$j$.
\subsection{The firm's point of view} 

Aiming to work with a simple model, we follow~\cite[Chapter 2]{gali2015monetary}. It then appears that the firm's problem corresponds to an optimization performed at each period, depending on the state of the world. This problem will depend, in particular, on the productivity and the price processes introduced above. Moreover, it will also depend on $P^i$ and $W^i$, two $\mathbb{G}$-adapted positive stochastic processes representing respectively the price of good $i$ and the wage paid in sector $i\in\Jj$. 
We start by considering the associated deterministic problem below, when time and randomness are fixed. 

\paragraph{Solution for the deterministic problem} 
We denote $\overline{a} \in (0,+\infty)^I$ the level of technology in each sector, $\overline{p} \in (0,\infty)^I$ the price of the goods produced by each sector, $\overline{w} \in (0,\infty)^I$ the nominal wage in each sector, $\overline{\tau} \in [0,1)^I$ and $\overline{\zeta} \in [0,1)^{I\times I}$ the price on production and consumption of goods.
For $i \in\Jj$,
we consider a representative firm of sector~$i$, with technology described by the production function 
\begin{align}
   \RR_+\times \RR_+^{I} \ni (n,z) \mapsto F^i_{\overline a}(n,z) = \overline{a}^i n^{\psi^i} \prod_{j\in\Jj} (z^{j})^{\llambda^{ji}} \in \RR_+,
\end{align}
where $n$ represents the number of hours of work in the sector, 
and~$z^j$ the firm's consumption of intermediary input produced by sector~$j$. The coefficients $\psi \in (\RR_+^*)^I$ and $\llambda \in (\RR_+^*)^{I\times I}$ are elasticities satisfying~\eqref{ct-eq const ret to scale}.
The management of firm~$i$
then solves the classical problem of profit maximization
\begin{align}\label{ct-eq firm problem deterministic}
    \widehat{\Pi}^i_{(\overline{a},\overline{w},\overline{p},\overline{\tau},\overline{\zeta}, \overline{\delta})} := \sup_{(n,z) \in \RR_+\times \RR_+^{I}  } \Pi^i(n,z),
\end{align}
where, omitting the dependency in $(\overline{a},\overline{w},\overline{p},\overline{\tau},\overline{\zeta})$,
\begin{equation}
   \Pi^i(n,z) := F^i_{\overline a}(n,z) \overline{p}^i - \overline{\tau}^i F^i_{\overline a}(n,z)\overline{p}^i\overline{\delta}  - \overline{w}^i n - \sum_{j \in \mathcal{I} }z^{j}\overline{p}^j + z^{j}\overline{\zeta}^{ji}\overline{p}^j\overline{\delta}.
\end{equation}

Note that $F^i_{\overline a}(n,z)(1-\overline{\tau} ^i)\overline{p}^i$ represents the firm’s revenues after carbon price, that~$\overline{w}^i n$ stands for the
firm’s total compensations, and that $\sum_{j \in \mathcal{I} }z^{j}(1+\overline{\zeta}^{ji})\overline{p}^j$ is the firm’s total intermediary inputs. Now, we would like to solve the optimization problem for the firms, namely determine the optimal demands~$\mathfrak{n}$ and~$\mathfrak{z}$ as functions of $(\overline{a},\overline{w},\overline{p},\overline{\tau},\overline{\zeta})$. Because we will lift these optimal quantities in a dynamical stochastic setting, we impose that they are expressed as measurable functions.
We thus introduce:
\begin{definition}\label{ct-de admissible solution}
    An \emph{admissible solution} to problem~\eqref{ct-eq firm problem deterministic} is a pair of measurable functions 
    \begin{equation*}
        (\mathfrak{n},\mathfrak{z}):(0,+\infty)^I\times(0,+\infty)^I\times (0,+\infty)^I \times [0,1)^I \times [0,1)^{I\times I} \rightarrow [0,+\infty)^I \times [0,+\infty)^{I\times I},
    \end{equation*}
    such that, for each sector $i$, denoting $\overline{n}:=\mathfrak{n}^i(\overline{a},\overline{w},\overline{p},\overline{\tau},\overline{\zeta})$ and $\overline{z}:=\mathfrak{z}^{\cdot i}(\overline{a},\overline{w},\overline{p},\overline{\tau},\overline{\zeta})$,
    \begin{align*}
       F^i_{\overline a}(\overline{n},\overline{z})(1-\overline{\tau}^i\overline{\delta})\overline{p}^i - \overline{w}^i \overline n - \sum_{j \in \mathcal{I} }\overline z^{j}(1+\overline{\zeta}^{ji}\overline{\delta})\overline{p}^j = \widehat{\Pi}^i_{(\overline{a},\overline{w},\overline{p},\overline{\tau},\overline{\zeta}, \overline{\delta})},
    \end{align*}
 and $F^i_{\overline a}(\overline{n},\overline{z})>0$ (non-zero production), according to~\eqref{ct-eq firm problem deterministic}.
\end{definition}

\begin{remark}
    The solution obviously depends also on the coefficients $\psi$ and $\llambda$. But these are fixed once and we will not study the dependence of the solution with respect to them. 
\end{remark}

\begin{proposition}\label{ct-pr firm foc}
There exists  admissible solutions in the sense of Definition~\ref{ct-de admissible solution}. Any admissible solution is given by for all $i \in \Jj$, $\mathfrak{n}^i>0$ and for all $(i,j) \in \Jj^2$,
\begin{align}\label{ct-eq optimal demand}
   \mathfrak{z}^{ji} = 
   \frac{\llambda^{ji}}{\psi^i} \frac{\overline w^i}{(1+\overline \zeta^{ji}\overline{\delta})\overline p^j}\mathfrak{n}^i > 0.
\end{align}
Moreover, it holds that $\widehat{\Pi}^i_{(\overline{a},\overline{w},\overline{p},\overline{\tau},\overline{\zeta},\overline{\delta})} = 0$ (according to~\eqref{ct-eq firm problem deterministic}) and 
\begin{subequations}
\begin{align}
\displaystyle \mathfrak{n}^i &= \psi^i F^i_{\overline a}(\mathfrak{n}^i, \mathfrak{z}^{\cdot i}) \frac{(1-\overline \tau^i\overline{\delta})\overline p^i}{\overline w^i} \label{ct-eq:Nopti-firm-inputs},\\
\displaystyle\mathfrak{z}^{ji} &= \llambda^{ji} F^i_{\overline a}(\mathfrak{n}^i, \mathfrak{z}^{\cdot i}) \frac{(1-\overline \tau^i\overline{\delta})\overline p^i}{(1+\overline \zeta^{ji}\overline{\delta})\overline p^j}.\label{ct-eq:Zopti-firm-inputs}
\end{align}
\end{subequations}

\end{proposition}

\begin{proof}
We study the optimization problem for the representative firm $i \in \Jj$.
Since $\psi^i>0$ and $\llambda^{ji} > 0$, for all $j \in \Jj$, as soon as $n=0$ or  $z^j = 0$, for some $j \in \Jj$, the production is equal to $0$. From problem~\eqref{ct-eq firm problem deterministic}, we obtain that necessarily $n \neq 0$ and $z^j\neq 0$ for all $j$ in this case. So an admissible solution, which has non-zero production, has positive components.
\\
Setting $\overline{n}=\mathfrak{n}^i(\overline{a},\overline{w},\overline{p},\overline{\tau},\overline{\zeta})>0$ and $\overline{z}=\mathfrak{z}^{\cdot i}(\overline{a},\overline{w},\overline{p},\overline{\tau},\overline{\zeta})>0$, the optimality of $(\overline n, \overline z)$ yields
\begin{equation*}
    \partial_{n} \Pi^i(\overline n, \overline z) = 0 \text{ and for any }j\in\Jj,\quad \partial_{z^j} \Pi^i(\overline n, \overline z) = 0.
\end{equation*}
We then compute
\begin{equation*}
    \psi^i \frac{F^i_{\overline a}(\overline n, \overline z)}{\overline n}(1-\overline \tau^i\overline{\delta})\overline p^i-\overline w^i = 0
    \text{ and for any }j\in\Jj,\quad 
    \llambda^{ji} \frac{F^i_{\overline a}(\overline n, \overline z)}{\overline z^j}(1-\overline \tau^i\overline{\delta})\overline p^i- (1+\overline \zeta^{ji}\overline{\delta})\overline p^j = 0,
\end{equation*}
which leads to~\eqref{ct-eq optimal demand}, \eqref{ct-eq:Nopti-firm-inputs}, and~\eqref{ct-eq:Zopti-firm-inputs}.
\end{proof}

\paragraph{Dynamic setting} In~\ref{ct-subse equilibrium} below, we  characterize the dynamics of the output and consumption processes using market equilibrium arguments. There, the optimal demand by the firm for intermediary inputs and labor is lifted to the stochastic setting where the admissible solutions then write as functions of the productivity, carbon price, price of goods/services; and wage processes, see Definition~\ref{ct-def:market equilibrium}. For all~$i\in\Jj$, $Y^i$ representing the production of sector~$i$, $N^i$ representing the labor demand in sector~$i$, and for all~$j\in\Jj$, $Z^{ji}$ representing the consumption by sector~$i$ of intermediate inputs produced by sector~$j$ are therefore positive and $\GG$-adapted processes.

\subsection{The household's point of view}
 Let $(r_t)_{t\geq 0}$ be the (exogenous) deterministic interest rate, valued in~$\RR_{+}$.
At each time $t\geq 0$ and for each sector~$i \in\Jj$, we denote
\begin{itemize}
    \item $C_{t}^i$ the quantity consumed of the single good in the sector~$i$, valued in~$\RR_{+}^{*}$;
    \item $H_{t}^{i}$ the number of hours of work in sector~$i$, valued in~$\RR_{+}^{*}$.
\end{itemize}

We also introduce a time preference parameter 
$\beta \in [0,1)$
and a utility function $U:(0,\infty)^2 \to \RR$ given, for $\varphi\geq 0$, by~$U(x, y) := \frac{x^{1-\sigma}}{1-\sigma} - \frac{y^{1+\varphi}}{1+\varphi}$
if  $\sigma \in [0, 1)\cup(1, +\infty)$ and by $U(x, y) := \log(x) - \frac{y^{1+\varphi}}{1+\varphi}$, if $\sigma = 1$. 
We also suppose that 
\begin{align}\label{ct-eq hyp integrability P,W}
\mathfrak{P}:=\sup_{t\geq 0,i \in\Jj} \EE\left[\left(\frac{P_{t}^i}{W_{t}^i}\right)^{1+\varphi}\right]<+\infty.
\end{align}

\noindent For any
$\Cc, \Hh \in \mathscr{L}^1_{+}(\mathbb{G},(0,\infty)^I)$, we introduce the wealth process
\begin{equation}\label{ct-eq:WealthProcess}
\dr Q_{t} = r_{t} Q_{t} \dr t + \sum_{i\in\Jj} W_{t}^i H_{t}^{i}- \sum_{i\in\Jj}P_{t}^i C_{t}^i - \sum_{i\in\Jj}\kappa_{t}^{i}P_{t}^iC_{t}^i\delta_t,
\qquad\text{for any }t\geq 0,
\end{equation}
with the convention $Q_{0}:=0$ and $r_{0}:=0$. 
Note that we do not indicate the dependence of~$Q$ upon~$C$ and~$H$ to alleviate the notations.\\

For~$t \geq 0$ and~$i \in\Jj$, $P_{t}^i C_{t}^i$ represents the household's consumption in the sector~$i$ and  $\kappa_{t}^{i}P_{t}^iC_{t}^i\delta_t$ is the cost paid by households due to their emissions when they consume goods~$i$, so $\sum_{i\in\Jj}P_{t}^i  C_{t}^i(1 + \kappa_{t}^{i}\delta_t)$ is the household's total expenses. Moreover, $W_{t}^i H_{t}^{i}$ is the household's labor income in the sector~$i$, $(1 + r_{t-1}) Q_{t-1}$ the household's capital income, and $(1 + r_{t-1}) Q_{t-1} + \sum_{i\in\Jj} W_{t}^i H_{t}^{i}$ the household's total revenue. 

We define $\mathscr{A}$ as the set of all  couples $(\Cc, \Hh)$ with~$\Cc, \Hh \in \mathscr{L}^1_{+}(\mathbb{G},(0,\infty)^I)$ such that
\begin{equation*}
\left\{
\begin{array}{ll}
 & \displaystyle \EE \left [ \sum_{i\in\Jj} \int_{t=0}^{\infty} \beta^t |U(C_{t}^i, H_{t}^i)|\dr t\right] < \infty,\\
 &\lim_{T\uparrow\infty}\EE[Q_T|\cG_t]\geq 0,
\qquad \text{for all } t \geq 0.
\end{array}
\right.
\end{equation*} 

The representative household consumes the~$I$ goods of the economy and provides labor to all the sectors. For any $(\Cc, \Hh)\in\mathscr{A}$, let
\begin{equation*}
    \mathcal{J}(\Cc, \Hh) := \sum_{i\in\Jj} \mathcal{J}_i(\Cc^{i}, \Hh^{i}),
    \qquad\textrm{with}\qquad
    \mathcal{J}_i( \Cc^{i}, \Hh^{i}) := \EE \left[\int_{t=0}^{\infty} \beta^t U(C_{t}^i, H_{t}^{i})\dr t\right], 
    \quad\textrm{for all } i\in\Jj.
\end{equation*} 
The representative household seeks to maximize its objective function by solving
\begin{equation}\label{ct-eq:HouseholdMax}
    \max_{(\Cc, \Hh) \in \mathscr{A}}\quad 
    \mathcal{J}(\Cc, \Hh).
\end{equation}
We choose above a separable utility function as~\cite{miranda2019comparing} does, meaning that the representative household optimizes its consumption and hours of work for each sector independently but under a global budget constraint.
\noindent The following proposition provides an explicit solution to~\eqref{ct-eq:HouseholdMax}.

\begin{proposition}\label{ct-prop:HouseholdMax}
Assume that~\eqref{ct-eq:HouseholdMax} has a solution $(C,H)\in \mathscr{A}$. Then,
for all $i,j\in\Jj$, 
the household’s optimality condition reads,
for any $t\geq 0$,
\begin{subequations}
\begin{align}
\displaystyle\frac{P_{t}^i}{W_{t}^i} & = \displaystyle \frac{1}{1 + \kappa_{t}^{i}\delta_t}  (H_{t}^i)^{-\varphi} (C_{t}^i)^{-\sigma}\label{ct-eq:fH_first_order_1},\\
\displaystyle\frac{P_{t}^i}{P_{t}^j} & = \displaystyle \frac{1 + \kappa_{t}^{j}\delta_t}{1 + \kappa_{t}^{i}\delta_t} \left( \frac{C_{t}^i}{C_{t}^j}\right)^{-\sigma}.\label{ct-eq:sH_first_order_1}
\end{align}
\end{subequations}
\end{proposition}
Note that the discrete-time processes~$C$ and~$H$ cannot hit zero by definition of~$\mathscr{A}$,
so that the quantities above are well defined.

\begin{proof}
Suppose that $\sigma \neq 1$. We first check that $\mathscr{A}$ is non empty.
Assume that, for all~$t\geq 0$ and~$i\in\Jj$, $\tilde{C}_t^i = 1$ and $\tilde{H}_t^i = \frac{P_t^i(1+\kappa_t^i)}{W_t^i}$, then
\begin{align*}
    \EE \left [ \sum_{i\in\Jj} \int_{t=0}^{\infty} \beta^t |U(\tilde{C}_{t}^i, \tilde{H}_{t}^i)| \dr t\right] & \le \sum_{i\in\Jj} \int_{t=0}^{\infty} \beta^t \left( \frac{1}{1-\sigma} + \frac{1}{1+\varphi} \EE \left [ \left(\frac{P_t^i(1+\kappa_t^i\delta_t)}{W_t^i}\right)^{1+\varphi}\right] \right) \dr t.
    \\
    &\le \sum_{i\in\Jj} 
    \int_{t=0}^{\infty} \beta^t \left( \frac{1}{1-\sigma} + \frac{\mathfrak{P}(1+\kappa_t^i\delta_t)^{1+\varphi}}{1+\varphi}  \right) \dr t<+\infty,
\end{align*}
using~\eqref{ct-eq hyp integrability P,W}. 
We also observe that $Q$ built from $\tilde{H}, \tilde{C}$  satisfies $Q_t = 0$, for $t\geq 0$. Thus $(\tilde{H}, \tilde{C}) \in \mathscr{A}$.

Let now $(\widehat{\Cc}, \widehat{\Hh}) \in\mathscr{A}$ be such that $\displaystyle\mathcal{J}(\widehat{\Cc}, \widehat{\Hh}) = \max_{(\Cc, \Hh) \in \mathscr{A}} \mathcal{J}(\Cc, \Hh)$.\\
\noindent We fix $s\geq 0$ and~$i\in\Jj$. Let~$\eta = \pm1$, $0<h<1$, $A^s \in \mathcal{G}_s$, 
$\Delta^{(i,s)}:= (\bOne_{\{i=k,s=t\}})_{k\in\Jj,t\geq 0}$ and
$\theta^{(i,s)} := \frac12(1\wedge \frac{W^i_s}{P^i_s(1+\kappa^i_s)})\hat{C}_s^i \wedge \hat{H}^i_s\wedge 1>0$.
Set 
\begin{align}
    \overline{\Cc} := \widehat{\Cc} + \eta h \theta^{(i,s)} \bOne_{A^s} \Delta^{(i,s)}
    \text{ and } 
    \overline{\Hh} := \widehat{\Hh} + \eta h \theta^{(i,s)} \bOne_{A^s} \Delta^{(i,s)} \frac{\Pp^i(1+\kappa^i\delta_s)}{\Ww^i}.
\end{align}
We observe that for $(j,t)\neq (i,s)$, $\overline{\Cc}^j_t = \widehat{\Cc}^j_t$ and $\overline{H}^j_t = \widehat{H}^j_t$ and we compute 
\begin{align*}
    \overline{\Cc}^i_s \ge \widehat{\Cc}^i_s - \theta^{(i,s)} \ge \frac12\widehat{\Cc}^i_s > 0.
\end{align*}
Similarly, we obtain $\overline{H}^i_s>0$. We also observe that $\overline{C} \le \frac32 \widehat{C}$ and $\overline{H} \le \frac32 \widehat{H}$. Finally, we have that 
$$\sum_{j\in\Jj} W_{t}^j \overline{H}_{t}^{j}- \sum_{j\in\Jj}P_{t}^j (1 + \kappa_{t}^{j}\delta_t) \overline{C} _{t}^j
=
\sum_{j\in\Jj} W_{t}^j \widehat{H}_{t}^{j}- \sum_{j\in\Jj}P_{t}^j (1 + \kappa_{t}^{j}\delta_t) \widehat{C} _{t}^j.
$$
This allows us to conclude that $(\overline{\Cc}, \overline{\Hh})\in\mathscr{A}$.\\

We have, by optimality of $(\widehat{\Cc}, \widehat{\Hh})$,
    \begin{equation*}
        \mathcal{J}(\widehat{\Cc}, \widehat{\Hh}) - \mathcal{J}(\overline{\Cc}, \overline{\Hh}) = \sum_{j\in\Jj} \mathcal{J}_j(\widehat{\Cc}^{j}, \widehat{\Hh}^{j}) - \sum_{j\in\Jj} \mathcal{J}_j(\overline{\Cc}^{j}, \overline{\Hh}^{j}) \geq 0.
    \end{equation*}
However, for all~$(t,j) \neq (s,i)$, $\overline{\Cc}^{j}_t = \widehat{\Cc}^{j}_t$ and $\overline{\Hh}^{j}_t = \widehat{\Hh}^{j}_t$, then
    \begin{equation*}
          \EE \left[\beta^s U(\widehat{\Cc}_{s}^{i}, \widehat{\Hh}_{s}^{i})\right] - \EE  \left[\beta^s U\left(\widehat{\Cc}_{s}^{i} + \eta h \theta^{(i,s)} \bOne_{A^s}, \widehat{\Hh}_{s}^{i} + \eta h \theta^{(i,s)} \bOne_{A^s}\frac{P_{s}^{i}(1+\kappa_{s}^{i}\delta_s)}{W_{s}^{i}}\right)\right] \geq 0,
    \end{equation*}
    i.e.
    \begin{equation*}
          \frac{1}{h} \EE \left[U(\widehat{\Cc}_{s}^{i}, \widehat{\Hh}_{s}^{i}) - U\left(\widehat{\Cc}_{s}^{i} + \eta h  \theta^{(i,s)} \bOne_{A^s}, \widehat{\Hh}_{s}^{i} + \eta h \theta^{(i,s)} \bOne_{A^s} \frac{P_{s}^{i}(1+\kappa_{s}^{i}\delta_s)}{W_{s}^{i}}\right)\right] \geq 0.
    \end{equation*}
    Letting $h$ tend to $0$, we obtain
    \begin{equation*}
          \EE \left[\eta \theta^{(i,s)} \bOne_{A^s} \frac{\partial U}{\partial x}(\widehat{\Cc}_{s}^{i}, \widehat{\Hh}_{s}^{i}) + \eta \theta^{(i,s)} \bOne_{A^s} \frac{P_{s}^{i}(1+\kappa_{s}^{i}\delta_s)}{W_{s}^{i}} \frac{\partial U}{\partial y}(\widehat{\Cc}_{s}^{i}, \widehat{\Hh}_{s}^{i})\right] \geq 0.
    \end{equation*}
    Since the above holds for all $A^s \in \mathcal{G}_s$, $\eta = \pm1$ and since $\theta^{(i,s)}>0$, then
    \begin{equation*}
        \frac{\partial U}{\partial x}(\widehat{\Cc}_{s}^{i}, \widehat{\Hh}_{s}^{i}) + \frac{P_{s}^{i}(1+\kappa_{s}^{i}\delta_s)}{W_{s}^{i}} \frac{\partial U}{\partial y}(\widehat{\Cc}_{s}^{i}, \widehat{\Hh}_{s}^{i}) = 0,
    \end{equation*}
    leading to~\eqref{ct-eq:fH_first_order_1}.\\
    For~$j\in\Jj \setminus \{i\}$ and
    $
    \theta^{(i,j,s)} := \frac12 \left(1 \wedge \frac{\Pp^j_s(1+\kappa_s^j\delta_s)}{\Pp^i_s(1+\kappa_s^i\delta_s)}\right)(1 \wedge \widehat{C}^i_s \wedge \widehat{C}^j_s)>0
    $, 
    setting now
     $$\overline{\Cc} := \widehat{\Cc} + \eta h \bOne_{A^s} \theta^{(i,j,s)}\left(\Delta^{(i,s)} -  \Delta^{(j,s)} \frac{\Pp^i(1+\kappa^i\delta_s)}{\Pp^j(1+\kappa^j\delta_s)}\right) 
     \quad\text{and}\quad \overline{\Hh} := \widehat{\Hh},$$ and using similar arguments as above, we obtain~\eqref{ct-eq:sH_first_order_1}.

\noindent When $\sigma = 1$, we carry out an analogous proof.
\end{proof}

\subsection{Markets equilibrium}
\label{ct-subse equilibrium}
\noindent We now consider that firms and households interact on the labor and goods markets.

\begin{definition} \label{ct-def:market equilibrium} A \emph{market equilibrium} is a $\mathbb{G}$-adapted positive random process $(\overline{W},\overline{P})$ such that
    \begin{enumerate}
        \item Condition~\eqref{ct-eq hyp integrability P,W} holds true for $(\overline{W},\overline{P})$.
        \item The goods' and labor's market clearing conditions are met, namely, for each sector~$i\in\Jj$, and for all $t \geq 0$,
        \begin{align}\label{ct-eq market clearing}
            Y_{t}^i = C_{t}^i + \sum_{j\in\Jj} Z_{t}^{ij}
            \qquad\text{and}\qquad  H_{t}^{i}= N_{t}^i, 
        \end{align}
        where $N_t = \overline{n}(A_t,\overline{W}_t,\overline{P}_t,\kappa_t,\zeta_t)$, $Z_t = \overline{z}(A_t,\overline{W}_t,\overline{P}_t,\kappa_t,\zeta_t)$, $Y =F_A(N,Z)$ with $(\overline{n},\overline{z})$ an admissible solution~\eqref{ct-eq:Nopti-firm-inputs}-\eqref{ct-eq:Zopti-firm-inputs} to~\eqref{ct-eq firm problem deterministic}, from Proposition~\ref{ct-pr firm foc} while $C$ and $H$ satisfy~\eqref{ct-eq:fH_first_order_1}-\eqref{ct-eq:sH_first_order_1} for $(\overline{W},\overline{P})$.
    \end{enumerate}
\end{definition}

In the case of the existence of a market equilibrium, we can derive equations that must be satisfied by the output production process~$Y$ and the consumption process~$C$.

\begin{proposition} \label{ct-prop:equilibrium}
Assume that there exists a market equilibrium as in Definition~\ref{ct-def:market equilibrium}. Then, for $t\geq 0$, $i\in\Jj$, it must hold that

\begin{equation}
\left\{
    \begin{array}{rl}
        Y_{t}^i & = \displaystyle C_{t}^i + \sum_{j \in\Jj}  \Lambda^{ij}(\dd_t) \left( \frac{C_{t}^j}{C_{t}^i}\right)^{-\sigma} Y_{t}^j,\\
        Y_{t}^i &= \displaystyle A^i_{t} \left[\Psi^i(\dd_t)(C_{t}^i)^{-\sigma} Y_{t}^i\right]^{\frac{\psi^i}{1+\varphi}} \prod_{j\in\Jj}  \left[\Lambda^{ji}(\dd_t) \left( \frac{C_{t}^i}{C_{t}^j}\right)^{-\sigma} Y_{t}^i \right]^{\llambda^{ji}},
    \end{array}
\right.\label{ct-eq:Y_it_C_it}
\end{equation}
where~$\Psi$ 
and~$\Lambda$ are defined in~\eqref{ct-eq:Psi_Lambda}, and $\dd_t$ is defined in~\eqref{ct-eq:emiss cost rate}.
\end{proposition}

\begin{proof} Let $i,j\in\Jj$ and $t\geq 0$.
Combining Proposition~\ref{ct-pr firm foc} and Proposition~\ref{ct-prop:HouseholdMax}, we obtain
\begin{align}\label{ct-eq Z equilibrium}
Z_{t}^{ji} = \llambda^{ji} \frac{1 - \tau_{t}^{i}\delta_t}{1+\zeta^{ji}_t\delta_t} \frac{1 + \kappa_{t}^{j}\delta_t}{1 + \kappa_{t}^{i}\delta_t} \left( \frac{C_{t}^i}{C_{t}^j}\right)^{-\sigma} Y_{t}^i.
\end{align}
From Propositions~\ref{ct-pr firm foc} and~\ref{ct-prop:HouseholdMax} again, we also have
\begin{align*}
    N_{t}^i
    = \psi^i \frac{1 - \tau_{t}^{i}\delta_t}{1 + \kappa_{t}^{i}\delta_t}  (H_{t}^{i})^{-\varphi} (C_{t}^i)^{-\sigma} Y_{t}^i.
\end{align*}
The labor market clearing condition in Definition~\ref{ct-def:market equilibrium} yields
\begin{equation}\label{ct-eq N equilibrium}
    N_{t}^i = \left[ \psi^i \frac{1 - \tau_{t}^{i}\delta_t}{1 + \kappa_{t}^{i}\delta_t} (C_{t}^i)^{-\sigma} Y_{t}^i\right]^{\frac{1}{1+\varphi}}.
\end{equation}
Then, by inserting  the expression of $N_{t}^i$ given in~\eqref{ct-eq N equilibrium}and~$Z_{t}^{ji}$ given in~\eqref{ct-eq Z equilibrium} into the production function $F$, we obtain the second equation in~\eqref{ct-eq:Y_it_C_it}.
The first equation in~\eqref{ct-eq:Y_it_C_it} is obtained by combining the market clearing condition with~\eqref{ct-eq Z equilibrium} (at index $(i,j)$ instead of $(j,i)$).
\end{proof}

\subsection{Output and consumption dynamics and associated growth}

For each time~$t\geq 0$ and noise realization, the system~\eqref{ct-eq:Y_it_C_it} is nonlinear with~$2I$ equations and~$2I$ variables,
and its well-posedness is hence relatively involved. 
Moreover, it is computationally heavy to solve this system for each price trajectory and productivity scenario. We thus consider a special value for the parameter $\sigma$ which allows to derive a unique solution in closed form. From now on, and following~\cite[page 63]{golosov2014optimal}, we assume that  $\sigma = 1$, namely $U(x, y) := \log(x) - \frac{y^{1+\varphi}}{1+\varphi}$
on $(0,\infty)^2$. 

\begin{theorem} \label{ct-eq:output_cons}
Assume that 
\begin{enumerate}
    \item $\sigma = 1$,
    \item $\Ir_I - \llambda$ is not singular,
    \item $\Ir_I - \Ll(\dd_t)^\top$ is not singular for all $t\in\RR_+$.
\end{enumerate}
Then for all $t\geq 0$, there exists a unique $(\Cc_{t},\Yy_{t})$ satisfying~\eqref{ct-eq:Y_it_C_it}.
Moreover, with $\ee_{t}^i := \frac{Y^i_t}{C^i_t}$ for $i\in\Jj$, we have 
\begin{align}
    \ee_{t} = \ee(\dd_t) := (\Ir_I-\Ll(\dd_t)^\top)^{-1} \bOne, \label{ct-eq:ee}
\end{align}
and using $\cB_t = (\cB^i_t)_{i \in \Jj} := \left[ \cA_{t}^{i} + \mfv^i(\dd_t)  \right]_{i\in\Jj}$ with
\begin{equation} 
    \mfv^i(\dd_t) := \log\left((\ee_{t}^{i})^{-\frac{\varphi\psi^i}{1+\varphi}}   \left(\Psi^i(\dd_t) \right)^{\frac{\psi^i}{1+\varphi}} \prod_{j\in\Jj}  \left(\Lambda^{ji}(\dd_t)\right)^{\llambda^{ji}} \right),\label{ct-eq:pricefunc}
\end{equation}
we obtain
\begin{equation}
    \Cc_{t} = \exp{\left( (\Ir_I-\llambda)^{-1}{\cB_t}\right)}.\label{ct-eq:consumption}
\end{equation} 
\end{theorem}

\begin{proof}
     Let~$t\geq 0$. When $\sigma = 1$, the system~\eqref{ct-eq:Y_it_C_it} becomes for all~$i\in\Jj$,
    \begin{align}\label{ct-eq equilibrium system}
\left\{
    \begin{array}{rl}
        Y_{t}^i & = \displaystyle C_{t}^i + \sum_{j \in\Jj}  \Lambda^{ij}(\dd_t) \left( \frac{C_{t}^i}{C_{t}^j}\right) Y_{t}^j,\\
        Y_{t}^i &= \displaystyle A_{t}^{i} \left[\Psi^i(\dd_t) \ee ^i_t\right]^{\frac{\psi^i}{1+\varphi}} \prod_{j\in\Jj}  \left[\Lambda^{ji}(\dd_t) C_{t}^j \ee_{t}^i \right]^{\llambda^{ji}}.
    \end{array}
\right.
\end{align}
For any~$i\in\Jj$,  dividing the first equation in~\eqref{ct-eq equilibrium system} by $C_{t}^i$, we get
\begin{equation*}
    \ee^i_t = 1 + \sum_{j \in\Jj}  \Lambda^{ij}(\dd_t) \ee^j_t,
\end{equation*}
which corresponds to~\eqref{ct-eq:ee},
thanks to~\eqref{ct-eq const ret to scale}.
Using $\sum_{j\in\Jj} \llambda^{ji} = 1-\psi^{i}$ and $Y_{t}^i = \ee_{t}^{i} C_{t}^i$ in the second equation in~\eqref{ct-eq equilibrium system}, we compute
\begin{equation*}
        C_{t}^i = A_{t}^{i} (\ee_{t}^{i})^{-\frac{\varphi \psi^i}{1+\varphi} }   \left[\Psi^i(\dd_t)\right]^{\frac{\psi^i}{1+\varphi}} \prod_{j\in\Jj}  \left[\Lambda^{ji}(\dd_t)\right]^{\llambda^{ji}} 
        \prod_{j\in\Jj} (C_{t}^j)^{\llambda^{ji}}.
\end{equation*}

Applying log and writing in matrix form, we obtain 
$(\Ir_I-\llambda) \log(\Cc_{t}) = {\cB_t}$,
implying~\eqref{ct-eq:consumption}.
\end{proof}

\begin{remark}
    The matrix $\llambda$ is generally not diagonal, and therefore, from~\eqref{ct-eq:consumption}, the sectors (in output and in consumption) are linked to each other through their respective productivity process. Similarly, an introduction of price in one sector affects the other ones.
\end{remark}

\begin{remark}
For any~$t\geq 0$, $i\in\Jj$, we observe that 
\begin{align}\label{ct-eq de v}
\cB^i_t = \cA^i_t + v^i(\dd_t),
\end{align}
where $v^i(\cdot)$ is defined using~\eqref{ct-eq:pricefunc}.
Namely, $\cB_t$ is the sum of the (random) productivity term and a term involving the price. The economy is therefore subject to fluctuations of two different natures: \textit{the first one comes from the productivity process while the second one comes from the price processes.}
\end{remark}
We now look at the dynamics of production and consumption growth.
\begin{theorem}\label{ct-prop:deltaY_C}
For any $t\geq 0$ and
for $\varpi \in \{\Yy, \Cc\}$. With the same assumptions as in Theorem~\ref{ct-eq:output_cons}, 
\begin{equation}
\dr\log{\varpi}_t \sim \cN\left({m}_{t}^{\varpi}, \widehat{\Sigma}_t \right),
\qquad\text{for }\varpi \in \{\Yy, \Cc\},\label{ct-eq:OutputLaw}
\end{equation}
with 
\begin{align}
    \widehat{\Sigma}_t &= \varsigma^2 (\Ir_I-\llambda)^{-1}\overline{\Sigma}(\Ir_I-\llambda^\top)^{-1} (\dr t)^2,\\
    m^\Cc_t &= (I-\llambda)^{-1}\left[ \mu\dr t+ \dr \mfv(\dd_t)\right],\\
    m^Y_t &= (I-\llambda)^{-1}\left[\mu\dr t + \dr\mfv(\dd_t) \right],
\end{align}
and 
\begin{equation}
    v(\dd_t) :=  \mfv(\dd_t) + (\Ir_I-\llambda)\log(\ee(\dd_t)), \label{ct-eq:taxfuncY}
\end{equation}
where  $\overline \mu$ and $\varsigma^2\overline{\Sigma}$ 
are the mean and the variance of the stationary process~$\cZ$ (Remark~\ref{ct-rem:VAR1}), 
$v$ is defined in~\eqref{ct-eq:pricefunc} and~$\ee$ in~\eqref{ct-eq:ee}.
\end{theorem}

\begin{proof}
Let $t\geq 0^*$, from~\eqref{ct-eq de v}, we have, for $i \in \Jj$, 
\begin{equation*}
    \dr \cB^i_t  = (\mu^i + \varsigma \cZ_t^i) \dr t + \dr v^i(\dd_t).
\end{equation*}
Combining the previous equality with~\eqref{ct-eq:consumption}, we get
\begin{equation}
    \dr\log{\Cc}_t = (\Ir_I-\llambda)^{-1} \left[(\mu + \varsigma \cZ_t)\dr t + \dr v(\dd_t)\right].\label{ct-eq:DeltalogC}
\end{equation}
Applying Remark~\ref{ct-rem:VAR1} leads to $\dr\log{\Cc}_t \sim \cN\left({m}_{t}^{\Cc}, \widehat{\Sigma}_t \right)$.
Using~\eqref{ct-eq:ee}, we observe that, for $i \in \Jj$,
\begin{align}
    (\dr\log{\Yy_t})^i = (\dr\log{\Cc_t})^i + \dr\log(\ee^i(\dd_t)),
\end{align}
which, using the previous characterization of the law of~$\dr\log{\Cc_t}$, allows to conclude.
\end{proof}

From the previous result, we observe that output and consumption growth processes have a stationary variance but a time-dependent mean. 

\begin{proof}\label{ct-proof:lem:approx firm value}
of Proposition~\ref{ct-lem:approx firm value}.\\

\noindent Let~$t\geq 0$, $n\in\OneN$, and $T > t_\star$. 
\begin{enumerate}
    \item we also introduce,
    \begin{equation}\label{ct-eq intro cVnKt}
        \cV^{n,K}_{t,\dd} :=F^n_{t,\dd} \int_{t}^{+\infty} e^{-r(s-t)}\EE_{t}\left[ \exp \left( (s-t)\af^{n\cdot} \mu
    +\af^{n\cdot}\left(v(\dd_{s})-v(\dd_t)\right)
    + \sigma_{n} (\cW^n_s - \cW^n_t)
    \right)\dr s\right].
    \end{equation}  
    Similar computations as (in fact easier than) the ones performed in the proof of Proposition 2.5. in~\cite{bouveret2023propagation} show that $\cV^n_{t,\dd} = \lim_{K\to+\infty} \cV^{n,K}_t $ is well defined in $\cL^q(\cH,\mathbb{E})$ for any $q\ge 1$.
    Furthermore,
    \begin{equation*}
        \cV^{n,K}_{t,\dd} = F^n_{t,\dd}\int_{s=0}^{K} e^{\varrho_n s}\exp{\left(\af^{n\cdot}\left(v(\dd_{t+s})-v(\dd_t)\right)
        \right)}\dr s = F^n_{t,\dd}e^{-\af^{n\cdot}v(\dd_t)} \int_{s=0}^{K} e^{\varrho_n s} \exp{(\af^{n\cdot}v(\dd_{t+s}))} \dr s,
    \end{equation*}  
    where $\varrho_n$ is defined in the lemma, and from Assumption~\ref{ct-ass:link} and Corollary~\ref{ct-cor:output_consc}, 
    \begin{equation*}
        F^n_{t} = F^n_0 \exp \left( \int_{u=0}^{t} \af^{n\cdot} (\Theta_u\dr u + \dr v(\dd_u)) + \sigma_{n} \dr \cW^n_t \dr u \right) 
        = F^n_0 e^{ \af^{n\cdot}\left(v(\dd_{t})-v(\dd_0)\right)} \exp \left( \af^{n\cdot}\cA^\circ_t+ \sigma_{n}\cW^n_t \right).
    \end{equation*}
    We then have
    \begin{equation*}
        F^n_{t,\dd} e^{-\af^{n\cdot}v(\dd_t)} \int_{s=0}^{K} e^{\varrho_n s} \exp{(\af^{n\cdot}v(\dd_{t+s}))} \dr s = F^n_0 e^{-\af^{n\cdot}v(\dd_0)} \exp \left( \af^{n\cdot}\cA^\circ_t+ \sigma_{n}\cW^n_t \right) \int_{s=0}^{K} e^{\varrho_n s}\exp{\left(\af^{n\cdot} v(\dd_{t+s}) 
        \right)} \dr s.
    \end{equation*}  
    \item Moreover,
    \begin{itemize}
        \item If $t < t_\circ$, then
    \begin{small}
    \begin{equation*}
        \begin{split}
    \mathfrak{R}^{n,K}_t(\dd) &:= \int_{s=0}^{K} e^{\varrho_n s}\exp{\left(\af^{n\cdot} v(\dd_{t+s})  \right)}\dr s\\ 
    &= \int_{s=0}^{t_\circ-t} e^{\varrho_n s}\exp{\left(\af^{n\cdot} v(\dd_{t+s})  \right)}\dr s + \int_{s=t_\circ-t}^{t_\star-t} e^{\varrho_n s}\exp{\left(\af^{n\cdot} v(\dd_{t+s})  \right)}\dr s + \int_{s=t_\star-t}^{K} e^{\varrho_n s}\exp{\left(\af^{n\cdot} v(\dd_{t+s})  \right)}\dr s\\
    &= e^{\af^{n\cdot}v(\dd_{t_\circ})} \frac{1-e^{\varrho_n (t_\circ-t)}}{-\varrho_n} + \int_{s=t_\circ-t}^{t_\star-t} e^{\varrho_n s}\exp{\left(\af^{n\cdot} v(\dd_{t+s})  \right)}\dr s + e^{\af^{n\cdot} v(\dd_{t_\star}) + \varrho_n (t_\star-t)} \frac{1-e^{\varrho_n (K-t_\star+t)}}{-\varrho_n}.
    \end{split}
    \end{equation*}
    \end{small}
        \item If $t_\circ \leq t < t_\star$, then
\begin{small}
\begin{equation*}
        \begin{split}
    \int_{s=0}^{K} e^{\varrho_n s}\exp{\left(\af^{n\cdot} v(\dd_{t+s})  \right)}\dr s &= \int_{s=0}^{t_\star-t} e^{\varrho_n s}\exp{\left(\af^{n\cdot} v(\dd_{t+s})  \right)}\dr s  + \int_{s=t_\star-t+1}^{K} e^{\varrho_n s}\exp{\left(\af^{n\cdot} v(\dd_{t+s})  \right)}\dr s \\
    &= \int_{s=0}^{t_\star-t} e^{\varrho_n s}\exp{\left(\af^{n\cdot} v(\dd_{t+s})  \right)}\dr s  + e^{\af^{n\cdot} v(\dd_{t_\star}) + \varrho_n (t_\star-t+1)} \frac{1-e^{\varrho_n (K-t_\star+t)}}{-\varrho_n}.
    \end{split}
    \end{equation*}
\end{small}
    \item If $t\geq t_\star$, then
    \begin{equation*}
        \begin{split}
    \int_{s=0}^{K} e^{\varrho_n s}\exp{\left(\af^{n\cdot} v(\dd_{t+s})  \right)}\dr s  &=\int_{s=0}^{K} e^{\varrho_n s}\exp{\left(\af^{n\cdot} v(\dd_{t_\star})  \right)}\dr s  = e^{\af^{n\cdot} v(\dd_{t_\star})} \frac{1-e^{\varrho_n (K+1)}}{-\varrho_n}.
    \end{split}
    \end{equation*}
    \end{itemize}
    Finally, $e^{\varrho_n (K+1)}$ and $e^{\varrho_n (K-t_\star+t)}$ converge to~$0$ for $\varrho_n < 0$ 
    as~$K$ tends to infinity, 
and the result follows.
\item We denote
    \begin{align*}
         V^{n,T}_{t,\dd} := \EE_{t}\left[\int_{t}^{T} e^{-r (s-t)} F^n_{s,\dd} \dr s \right].
    \end{align*}
    As we have from~\eqref{ct-eq:CF_vs_GDPGrowth with conso}, $F^n_{s,\dd} = F^n_{t,\dd}\exp\left( \af^{n\cdot}(\cA_s-\cA_t)+\af^{n\cdot}\left(v(\dd_{s})-v(\dd_t)\right)+ \sigma_{n}(\cW^n_s-\cW^n_t) \right)$, and given that for all $h,t\geq 0$,
    \begin{equation*}
    \cA_{t+h} = \cA_{t} + \mu h + \varsigma\Upsilon_{h}\cZ_t - \varsigma\Gamma^{-1} \int_{t}^{t+h} \left(e^{-\Gamma (t+h -s)} - \Ir_I \right) \Sigma \dr B_s^{\cZ}.
\end{equation*}
We obtain
    \begin{align*}
        V^{n,T}_{t,\dd}&= \EE_{t}\left[\int_{t}^{T} e^{-r (s-t)} F^n_{t}\exp\left( \af^{n\cdot}(\cA_s-\cA_t)+\af^{n\cdot}\left(v(\dd_{s})-v(\dd_t)\right)+ \sigma_{n}(\cW^n_s-\cW^n_t) \right) \dr s \right]\\
        &=F^n_{t,\dd}\int_{t}^{T} e^{\left(\frac{1}{2}\sigma_{n}^2-r\right) (s-t)} \exp\left(\af^{n\cdot}\left(v(\dd_{s})-v(\dd_t)\right) \right) \EE_{t}\left[\exp\left( \af^{n\cdot}(\cA_s-\cA_t) \right)\right] \dr s \\
         &=F^n_{t,\dd}\int_{t}^{T} e^{\left(\frac{1}{2}\sigma_{n}^2+\af^{n\cdot}\mu-r\right) (s-t)} \exp\left(\af^{n\cdot}v(\dd_{s})-v(\dd_t)\right) \exp\left( \varsigma\af^{n\cdot}\Upsilon_{s-t}\cZ_t +\frac{1}{2} \af^{n\cdot} \Sigma^{\cA,h}_{t} (\af^{n\cdot})^\top\right) \dr s.
    \end{align*}
    Then using Hölder's inequality (with $1= \frac{1}{p} + \frac{1}{q}$), we have
    \begin{small}
    \begin{align*}
        \lVert V^{n,T}_{t,\dd}\rVert_1 &\leq \lVert F^n_{t,\dd}\rVert_q \left\lVert \int_{t}^{T} e^{\left(\frac{1}{2}\sigma_{n}^2+\af^{n\cdot}\mu-r\right) (s-t)} \exp\left(\af^{n\cdot}v(\dd_{s})-v(\dd_t)\right) \exp\left( \varsigma\af^{n\cdot}\Upsilon_{s-t}\cZ_t +\frac{1}{2} \af^{n\cdot} \Sigma^{\cA,s-t}_{t} (\af^{n\cdot})^\top\right) \dr s\right\rVert_p\\
        &\leq \lVert F^n_{t,\dd}\rVert_q \int_{t}^{T} e^{\left(\frac{1}{2}\sigma_{n}^2+\af^{n\cdot}\mu-r\right) (s-t)} \exp\left(\af^{n\cdot}v(\dd_{s})-v(\dd_t)\right) \exp\left(\frac{1}{2} \af^{n\cdot} \Sigma^{\cA,s-t}_{t} (\af^{n\cdot})^\top\right)\left\lVert \exp\left( \varsigma\af^{n\cdot}\Upsilon_{s-t}\cZ_t\right)\right\rVert_p \dr s.
    \end{align*}
    \end{small}
    Observe that under Assumption~\ref{ct-sassc:price}, there exists a constant $\mathfrak{C}_\dd>0$ such that
\begin{align*}
    \sup_{n,s,t}\exp \left( \mathfrak{a}^{n\cdot}\left(v(\dd_{s})-v(\dd_t) \right)\right) \le \mathfrak{C}_\dd\,.
\end{align*}
Given that $\cZ$ is stationary and $\Upsilon_{s-t}$ is bounded (\eqref{ct-eq:Upsilon bound}), there exists $\mathfrak{C}_{n,p}>0$ so that $\leq \mathfrak{C}_{n,p}$
    \begin{align*}
        \left\lVert \exp\left( \varsigma\af^{n\cdot}\Upsilon_{s-t}\cZ_t\right)\right\rVert_p = \EE\left[\exp\left( \varsigma p \af^{n\cdot}\Upsilon_{s-t}\cZ_t\right)\right]^\frac{1}{p} \leq \mathfrak{C}_{n,p}. 
    \end{align*}
    Moreover,
    \begin{align*}
        \exp\left(\frac{1}{2} \af^{n\cdot} \Sigma^{\cA,h}_{t} (\af^{n\cdot})^\top\right) &= \exp\left(\frac{1}{2}  \varsigma^2\int_{0}^{s-t} \af^{n\cdot}\Upsilon_{u} \Sigma\Sigma^\top \Upsilon_{u}^\top  (\af^{n\cdot})^\top \dr u \right) \\
        &\leq \exp\left(\frac{1}{2}  \varsigma^2\int_{0}^{s-t} \lVert\af^{n\cdot}\rVert^2 \lVert\Sigma\rVert^2 \lVert\Upsilon_{u}\rVert^2 \dr u \right)\\
        &\leq \exp\left(\frac{1}{2}  \varsigma^2\frac{c_\Gamma^2}{\lambda_\Gamma^2}\lVert\af^{n\cdot}\rVert^2 \lVert\Sigma\rVert^2 (s-t) \right).
    \end{align*}
    Next, we can write
    \begin{align*}
        \lVert V^{n,T}_{t,\dd}\rVert_1 &\leq \mathfrak{C}_\dd \mathfrak{C}_{n,p} \lVert F^n_{t,\dd}\rVert_q \int_{t}^{T} \exp{\left(\frac{1}{2}\sigma_{n}^2+\af^{n\cdot}\mu+\frac{1}{2}  \varsigma^2\frac{c_\Gamma^2}{\lambda_\Gamma^2}\lVert\af^{n\cdot}\rVert^2 \lVert\Sigma\rVert^2-r\right) (s-t)} \dr s,
    \end{align*}
    and if \eqref{ct_eq:main technical ass} is satisfied and $T\to+\infty$, then $V^{n,K}_{t,\dd}$ converges to~$V^n_{t,\dd}$. Finally, similar methods must be used to show $\EE \left[\left|\frac{V^n_{t,\dd}}{F^n_{t,\dd}} - \frac{\cV^n_{t,\dd}}{F^n_{t,\dd}}\right|\right] \le C \varsigma$.
\end{enumerate}
\end{proof}

\end{document}